\documentclass[aps, pra, onecolumn, 10pt, tightenlines,
notitlepage, longbibliography, superscriptaddress,a4paper]{revtex4-2}

\usepackage{amsmath, amsthm, amsfonts, amssymb, amsbsy, mathtools, xcolor, bm, bbm, graphicx, changepage, hyperref, cleveref}
\usepackage[T1]{fontenc}
\usepackage[margin=3cm]{geometry}

\hypersetup{
	colorlinks,
	linkcolor={blue},
	citecolor={blue},
	urlcolor={blue},
	pdftitle={Efficient estimation of Pauli channels},
	pdfauthor={Steven T Flammia and Joel J Wallman}
}

\DeclareMathOperator{\tr}{Tr}
\theoremstyle{plain}
\newtheorem{theorem}{Theorem}
\newtheorem{lemma}[theorem]{Lemma}
\newtheorem{proposition}[theorem]{Proposition}

\newtheorem{corollary}[theorem]{Corollary}
\newtheorem{result}{Result}

\theoremstyle{definition}
\newtheorem{definition}{Definition}

\newenvironment{alg}[1][]{\medskip
\begin{adjustwidth}{1cm}{1cm}
{#1}
\begin{enumerate}}
{\end{enumerate}
\end{adjustwidth}
\medskip}

\newcommand{\mc}{\mathcal}
\renewcommand{\ss}[1]{{\sf{#1}}}

\newcommand{\W}{\bm{W}}
\newcommand{\noisy}[1]{\ensuremath{\tilde{\mathcal{#1}}}}

\begin{document}

\title{Efficient estimation of Pauli channels}

\author{Steven T. Flammia}
\affiliation{Centre for Engineered Quantum Systems, School of Physics, University of Sydney, Sydney, NSW 2006 Australia}
\affiliation{Yale Quantum Institute, Yale University, New Haven, CT 06520, USA}
\affiliation{Quantum Benchmark Inc., 100 Ahrens Street West, Suite 203, Kitchener, ON N2H 4C3, Canada}

\author{Joel J. Wallman}
\affiliation{Quantum Benchmark Inc., 100 Ahrens Street West, Suite 203, Kitchener, ON N2H 4C3, Canada}
\affiliation{Institute for Quantum Computing and Department of Applied Mathematics, University of Waterloo, Waterloo, Ontario N2L 3G1, Canada}

\date{\today}

\begin{abstract}
Pauli channels are ubiquitous in quantum information, both as a dominant noise source in many computing architectures and as a practical model for analyzing error correction and fault tolerance.
Here we prove several results on efficiently learning Pauli channels, and more generally the Pauli projection of a quantum channel.
We first derive a procedure for learning a Pauli channel on $n$ qubits with high probability to a \emph{relative} precision $\epsilon$ using $O\bigl(\epsilon^{-2} n 2^n\bigr)$ measurements, which is efficient in the Hilbert space dimension.
The estimate is robust to state preparation and measurement errors which, together with the relative precision, makes it especially appropriate for applications involving characterization of high-accuracy quantum gates.
Next we show that the error rates for an arbitrary set of $s$ Pauli errors can be estimated to a relative precision $\epsilon$ using $O\bigl(\epsilon^{-4} \log s\log s/\epsilon\bigr)$ measurements.
Finally, we show that when the Pauli channel is given by a Markov field with at most $k$-local correlations, we can learn an entire $n$-qubit Pauli channel to relative precision $\epsilon$ with only $O_k\bigl(\epsilon^{-2} n^2 \log n \bigr)$ measurements, which is efficient in the number of qubits.
These results enable a host of applications beyond just characterizing noise in a large-scale quantum system: they pave the way to tailoring quantum codes, optimizing decoders, and customizing fault tolerance procedures to suit a particular device.
\end{abstract}

\maketitle

\section{Introduction}

Pauli channels are among the most basic noise channels in quantum information science.
Nearly all theoretical studies of quantum error correction and fault tolerance, including most threshold and overhead estimates, rely on modeling noise as a Pauli channel~\cite{Terhal2015}.

Originally, the theoretical focus on Pauli channels was primarily motivated by the ease with which they can be theoretically analyzed and simulated on a classical computer.
However, Pauli channels are now justified by the technique of randomized compiling~\cite{Wallman2016}, which maps general quantum noise to its Pauli projection, a Pauli channel having the same average fidelity to an ideal channel.
This addresses the concern that coherent noise, which is in general as hard to model as full quantum computation, may create distorted comparisons with threshold error rates that were computed using Pauli noise~\cite{Sanders2015, Kueng2016}.
Moreover, quantum error correction of coherent noise will, under reasonable assumptions, lead to less coherent noise at the logical level, meaning that noise at that level is better approximated by a Pauli channel~\cite{Huang2018, Beale2018}.
Finally, recent experiments~\cite{Ware2018} have shown that enforcing the Pauli projection by randomized compiling works extremely well in practice, further solidifying the importance of Pauli channels.

Despite the central role played by Pauli channels, to date there have been few systematic studies of how to estimate them efficiently, meaning with a complexity that improves over what follows from a naive application of full channel tomography.
We will review the relevant literature below.
Filling this gap becomes even more pressing in light of work showing just how much the threshold changes under biased or correlated noise models~\cite{Bombin2012,Nickerson2017,Darmawan2017,Maskara2018}, in some cases by more than a factor of 4 for a code capacity threshold~\cite{Tuckett2018, Tuckett2019,Tuckett2020}.
Such factors can be substantial because the logical error rate depends exponentially on the distance below the threshold.
Accurate estimation of the Pauli noise in an architecture would allow for many applications, including tailoring codes and decoders to match the noise~\cite{Robertson2017}, customizing fault-tolerance schemes~\cite{Aliferis2008}, and accurate estimation of thresholds and overheads~\cite{Chubb2018}.

\subsection{Summary of Results}

In this paper, we give a comprehensive treatment of the sample complexity of Pauli channel estimation.
We will present three main results corresponding to efficient estimation procedures with error guarantees for three separate and complementary regimes: estimation of a complete Pauli channel, estimation of error rates for an arbitrary set of Pauli errors, and estimation of a Pauli channel that factorizes over a bounded-degree factor graph.
The latter is equivalent to considering the Pauli error probability as a tensor network state comprised of tensors with a bounded number of indices, but no restriction on the topology of the connections.
This section contains only informal statements of our results; the precise versions of these statements along with rigorous proofs are given in the subsequent sections.

The basic procedure that we develop and analyze is a variant of randomized benchmarking~\cite{Emerson2005, Dankert2009, Knill2008} and its recently introduced cousins, character benchmarking~\cite{Helsen2018} and cycle benchmarking~\cite{Erhard2019}.
The procedure, defined in \Cref{sec:estimation}, inherits the robustness to errors in state preparation and measurement (SPAM) enjoyed by all benchmarking variants.
It uses preparations and measurements in a stabilizer basis (or equivalently, in the computational basis together with a single Clifford group element), and repeated rounds of random Pauli gates to average over the noise.
Our results apply when the noise on the random Pauli gates, the state preparations, and the measurement all obey certain mild regularity assumptions such as being gate-independent, time-stationary, Markovian, and not too far from ideal.
These assumptions can be relaxed still further, particularly the gate-independence assumption~\cite{Chasseur2015,Proctor2017,Wallman2018,Merkel2018}.
The precise conditions on the noise for which our proofs hold are given in \Cref{def:GTMnoise,def:weakstable}.
In this section, we informally refer to these restrictions as ``nice'' noise.

Our first main result is that our procedure for estimating a complete Pauli channel on $n$ qubits requires only $O\bigl(\epsilon^{-2} n 2^n\bigr)$ measurements to estimate the channel to constant \textit{relative} precision $\epsilon$ and with a constant success probability.
This result forms a core subroutine in our subsequent two results.
The output is the vector $\bm{p}$ of all $4^n$ Pauli error probabilities, which we achieve by using the full power of $n$-bit measurements.
Here the relative precision and robustness to SPAM are crucial for applications, since average error rates in quantum gates are now routinely below 1\%, and in some cases as low as $10^{-6}$~\cite{Harty2014}.
In this regime, a meaningful additive error approximation would require at least $10^{12}$ samples, making it far outside the realm of practicality even if it is still technically ``efficient''.

\begin{result}[Informal summary of \Cref{prop:RatioProp,prop:errorbound}]\label{res:group}
For a nice Pauli channel on $n$ qubits, the Pauli error rates $\bm{p}$ can be estimated using $O\bigl(\epsilon^{-2} n 2^n\bigr)$ measurements by $\hat{\bm{p}}$ such that
\begin{align*}
    \| \hat{\bm{p}} - \bm{p} \|_2 
    \leq O(\epsilon)(1-p_0),
\end{align*}
holds with high probability, where $p_0$ is the probability of no error.
\end{result}

We can apply a randomized sampling routine to estimate a subnormalized probability distribution over an arbitrary set $\ss{E}$ of size $|\ss{E}| = s$. 
We note that the $\epsilon^{-4}$ scaling is an artifact of the proof technique and expect it can be improved to $\epsilon^{-2}$ by a more careful analysis of the bias in the output of the subroutine ${\bf Ratio}$ described in \cref{sec:analysis}.

\begin{result}[Informal statement of \Cref{thm:sparse}]
\label{res:sparse}
For a nice Pauli channel on $n$ qubits, a subnormalized distribution over any set $\ss{E}$ of $s$ Pauli errors can be estimated using $O(\epsilon^{-4}\log(s)\log(s/\epsilon^2))$ measurements by $\hat{\bm{p}}$ such that
\begin{align*}
	\| \hat{\bm{p}} - \bm{p} \|_\infty
	&\leq O(\epsilon)(1-p_0)
\end{align*}
holds with constant probability, where $p_0$ is the probability of no error.
\end{result}

The above sampling protocol can be applied directly to efficiently estimate the probability of all low-weight Pauli errors, which will account for the majority of the distribution under realistic physical assumptions in near-term high-performance quantum systems.
We also provide a tree-based search heuristic that can be used to identify sets of errors of interest.
Two natural candidates for interesting sets of errors are the best $s$-sparse approximation without assuming local errors, and high-weight Pauli errors that occur with probabilities that differ substantially from model predictions.
Although we do not give a formal analysis of the probability with which the search heuristic identifies a correct set, when it finds a sparse set with large measure it is certifiably correct as a description of the Pauli channel. 
This follows because \Cref{thm:sparse} can be applied to the final output without accounting for failure probabilities from intermediate steps.

Our final main result is the efficient reconstruction of any nice $n$-qubit Pauli channel in polynomial time in $n$ whenever the channel has bounded-degree correlations and the local marginals are positive, conditions which we quantify in \cref{sec:boundeddegree}.
We assume throughout that the correlations are modeled by a factor graph with fixed known topology.

\begin{result}[Informal statement of \Cref{prop:boundeddegree}.]
Let $\bm{p}$ be a nice $n$-qubit Pauli channel with a $k$-degree factor graph and positive marginals.
Then an estimate $\hat{\bm{p}}$ as a tensor network can be obtained using $O_k\bigl(\epsilon^{-2} n^2 \log n\bigr)$ measurements such that
\begin{align*}
	\| \hat{\bm{p}}-\bm{p} \|_1 \le O(\epsilon) \|\bm{1}_I-\bm{p}\|_\infty
\end{align*}
holds with high probability, where $O_k$ means the implicit big-$O$ constant depends exponentially on $k$.
Moreover, an estimate proportional to $\hat{\bm{p}}$ can be found in time $\mathrm{poly}(n)$.
\end{result}
Here we use the notation $\bm{1}_I$ to mean a probability distribution supported only on the identity Pauli operator, which corresponds to a noiseless channel.

\subsection{Proof techniques}

As mentioned above, we make use of a variant of randomized benchmarking (RB). 
RB was originally applied to the full unitary group and the Clifford group~\cite{Emerson2005, Dankert2009}, but subsequent work has developed it for many different groups~\cite{Emerson2007, Ryan2009, Gambetta2012, Barends2014a, Carignan-Dugas2015, Cross2016, Hashagen2018, Brown2018}; see Refs.~\cite{Franca2018, Helsen2018} for a general treatment. 
Our work focuses on RB over the Pauli group, and introduces several important advances over prior art. 

As shown by Harper et al.~\cite{Harper2018}, RB over the Clifford group provably gives relative precision estimates of a single parameter associated to a quantum channel, the average error rate.
Our proof extends this to the case of RB over the Pauli group.
The original proof of Ref.~\cite{Harper2018} isolates a single exponential decay from an RB signal using the method from Ref.~\cite{Fogarty2015}.
However, for the case of interest here, the Pauli group, there are $4^n$ parameters that must be learned instead of just one. 
While a general solution exists to the problem of isolating exponential decays~\cite{Helsen2018}, prior work estimates the individual decay rates serially using one-bit measurements. 
Our work greatly improves this by showing how to take the natural $n$-bit measurements arising from measuring the $n$ individual qubits and using these measurements to estimate up to $2^n$ parameters in parallel.
When sampled in this way, a Hadamard transformation can be done to isolate the exponentials on all data collected from the $n$-bit measurements simultaneously, avoiding the need for serial exponential fits. 
Our proof for the case of estimating probabilities for arbitrary sets of Pauli errors uses the essential ideas above with random sampling techniques. 

Finally, we assume that the Pauli channel error rates are given by a Markov random field over a known factor graph.
We can again apply the core subroutine for estimating Pauli channels to the marginal channel on each of the factors and then round these estimated marginals into a coherent global probability distribution using the Hammersley-Clifford theorem~\cite{Hammersley1971}.

\subsection{Comparison with Prior Work}

These results comprise the first proofs in the literature of recovery guarantees for quantum channels to relative precision while avoiding bias from SPAM errors.
As noted above, this is absolutely essential for applications.
Furthermore, Results 2 \& 3 are the first recovery guarantees known that apply to broad classes of non-trivial quantum channels on $n$ qubits with a scaling that is efficient in $n$ (when $s=\mathrm{poly}(n)$ or $k=O(1)$, respectively).
Lastly, although numerical practicality is not a focus of the present paper (it will be explored elsewhere~\cite{Harper2019}) we wish to stress that these results are \emph{practical} and immediately applicable to characterizing error rates in near-term devices, such as the 50- and 72-qubit processors announced by IBM and Google respectively.
To emphasize this last point, the results of our numerical simulations of these procedures on up to 100 qubits show excellent performance, and we have successfully implemented a variant of the method on a publicly accessible 14-qubit quantum device~\cite{Harper2019}. 
To put this in perspective, the previous record for quantum channel tomography is of a 3-qubit quantum gate~\cite{Rodionov2014}, which was only made possible by employing sophisticated methods from compressed sensing~\cite{Gross2010,Flammia2012,Shabani2011}. 
It is clear from these examples that our results represent a qualitative shift in the characterization of quantum devices.

Let us compare the scaling of our procedures to the best previously known results from the literature.
A naive, non-adaptive application of channel tomography with (single-copy) two-outcome measurements~\cite{Chuang1997} would require $\tilde{O}\bigl(\tfrac{d^4}{\epsilon^2}\bigr)$ state preparations and measurements to achieve an \emph{additive} precision $\epsilon$ in, say, the average fidelity of the channel, where $d=2^n$ for $n$ qubits.
Applying the best known procedures for state tomography~\cite{Haah2017, ODonnell2016, ODonnell2017} to the Choi state of the channel using a collective measurement requires $\tilde{\Theta}\bigl(\tfrac{d^4}{\epsilon^2}\bigr)$ copies, so this is roughly tight.
For sparse channels in a known superoperator basis (such as the Pauli basis), Ref.~\cite{Shabani2011} argues from the theory of compressed sensing that $O(s \log d)$ random product basis measurements should suffice to reconstruct an $s$-sparse channel in the Pauli basis, but it is not clear that the technical ``incoherence'' or restricted isometry properties that are used for efficiently reconstructing sparse vectors are satisfied in this case~\cite{Flammia2012}.
The recently introduced procedure of shadow tomography~\cite{Aaronson2017} could achieve an estimate of the Pauli error rates with a small ($\mathrm{poly}\log d$) number of copies of the Choi matrix for the channel, but at the cost of a large ($\mathrm{poly}(d)$) amount of nontrivial quantum computation and a copy complexity of $~1/\epsilon^5$ with respect to the additive error.
Our procedure uses much more trivial quantum resources, and with a circuit depth that is independent of the dimension.
The idea of tomography of matrix product states~\cite{Cramer2010a} has been extended to matrix product operators that are unitary in the numerical work of Ref.~\cite{Holzapfel2015}.
There is as yet no systematic mathematical treatment of matrix product operator tomography, or tomography of quantum channels with bounded correlations of the kind discussed in this paper.

The above tomographic techniques only provide \emph{additive} precision estimates rather than the \emph{relative} precision estimates achieved here.
Furthermore, these tomographic methods are not robust against SPAM errors, which create an unknown systematic error, or, equivalently, places an absolute floor on the achievable additive precision that can be obtained using these tomographic techniques~\cite{Merkel2012}.
The magnitude and complexity of SPAM errors scales with the system size, so that the achievable precision decreases for larger numbers of qubits.
These tomographic techniques could be made robust to SPAM errors using gate-set tomography~\cite{Merkel2012, Blume-Kohout2016}, at the cost of substantially increasing the resources required.
However, it is unclear how to obtain relative-precision estimates from any of these methods.

We achieve multiplicative precision and robustness to SPAM errors by using techniques originating from randomized benchmarking.
Randomized benchmarking tomography~\cite{Kimmel2013} also employs randomized benchmarking techniques to achieve robustness to SPAM, and can obtain tomographic reconstructions using $O\bigl(d^2\log(d)\bigr)$ randomized benchmarking experiments with an additive reconstruction error that scales as $O(d/\sqrt{\log d})$~\cite{Roth2018}.
Note that the additive error is essentially fundamental to the approach used in randomized benchmarking tomography because the decay rates are the fidelities between the noisy process and distinct Clifford gates.
These fidelities are significantly less than 1---even for ideal noise processes---resulting in rapid decays that are difficult to estimate precisely.

Besides the prior work on quantum tomography, there has also been a large amount of prior work on learning Markov random fields, both in terms of parameter estimation (including hardness results~\cite{Bresler2014}) and in terms of learning the structure of the factor graph~\cite{Chow1968, Abbeel2006, Bresler2013, Bresler2015, Hamilton2017, Klivans2017}.
These results are certainly relevant, but to make a direct comparison to the present results is complicated by the fact that our probability distribution contains variables that might be called ``quasi-latent'', in the sense that they are not observable to all queries, only to some, and queries that probe one variable must necessarily hide others.
This is due to the symplectic structure of the Pauli group: we can only simultaneously measure observables that commute.
Therefore learning a ``symplectic Markov field'' seems to be an inherently different task from what has been previously considered.

\section{Mathematical Preliminaries}

Given a set of $n$ qubits with Hilbert space dimension $d=2^n$, let us introduce the following notation.
Let $\mathbb{P}^n$ denote the group of Pauli operators acting on $n$ qubits, and $\ss{P}^n = \mathbb{P}^n/\langle i\rangle$ be the quotient of $\mathbb{P}^n$ with its center.
The group $\ss{P}^n$ is Abelian and isomorphic to $\mathbb{Z}_2^{2n}$, so we will label elements in $\ss{P}^n$ by bit strings of length $2n$.
We will abuse notation and talk about $a \in \ss{P}^n$ when we mean a matrix $P_a \in \ss{P}^n$ up to an overall phase.
For concreteness, for a bit string $a$ and a single-qubit Pauli operator $A$ we write $A[a] = \bigotimes_{a_i \in a} A^{a_i}$ and choose the elements of $\ss{P}^n$ to be
\begin{align}
	P_a = P_{(a_x,a_z)} = i^{a_x \cdot a_z} X[a_x]Z[a_z]
\end{align}
where $a\in \mathbb{Z}_2^{2n}$ and $X$ and $Z$ are the standard single-qubit Pauli matrices.
The choice of phase is arbitrary but this choice of phase ensures that every $P_a$ is Hermitian.
For any two Pauli matrices $P_a$ and $P_b$ with respective images $a$ and $b$ in $\mathbb{Z}_2^{2n}$, we have $P_a P_b = (-1)^{\langle a , b\rangle} P_b P_a$ where
\begin{align}
	\langle a,b\rangle = a_x\cdot b_z + a_z\cdot b_x \mod 2
\end{align}
is a binary symplectic form, and so it is symmetric and linear in each argument.
We will typically omit the mod 2 because it is clear from context.

We define a stabilizer group to be a linear subspace of $\mathbb{Z}_2^{2n}$ such that $\langle a,b\rangle=0$ for all $a,b\in\ss{S}$.
Therefore every stabilizer group is a group $\ss{S}\subset\sf{P}^n$ whose elements all commute as matrices.
Note that stabilizer groups are often defined in terms of the Pauli matrices with specific phase conventions~\cite{Gottesman1997}.

The commutant of a \textit{set} $\ss{G} \subseteq \ss{P}^n$ is the group $\ss{C_G} \subseteq \ss{P}^n$ that is orthogonal to $\ss{G}$ according to the binary symplectic form.
That is,
\begin{align}
    \ss{C_G} = \{a\in\ss{P}^n:\forall g\in\ss{G},\ \langle a,g\rangle = 0\}.
\end{align}
We refer to $\ss{C_G}$ as the commutant because every element of $\ss{C_G}$ commutes with every element of $\ss{G}$ as matrices in $\mathbb{P}^n$.
When $\ss{G}$ is a group (i.e., closed under multiplication), $\ss{C_{C_G}} = \ss{G}$ by the double commutant theorem.
We define the anti-commutant $\ss{A_G}$ of $\ss{G}\subseteq\ss{P}^n$ to be the quotient group $\ss{P}^n/\ss{C_G}$.
Any $b\in\ss{P}^n$ can be uniquely decomposed as
\begin{align}
b = a + c
\end{align}
for some $a\in\ss{A_G}$ and $c\in\ss{C_G}$.

Any element of the Pauli group (except those proportional to the identity) commutes with exactly half of the other elements in the group.
A similar statement is true for any subgroup, and we have the following lemma.
In what follows, we use the Iverson bracket notation, where if $x$ is any logical proposition (such as a set inclusion), then
\begin{align}
    1[x] = \begin{cases}
    1 & \text{if $x$ is true},\\
    0 & \text{if $x$ is false}.
    \end{cases}
\end{align}

\begin{lemma}\label{lem:pauliOrthogonality}
For any Pauli $a\in\ss{P}^n$ and any group $\ss{G}\subseteq\ss{P}^n$,
\begin{align*}
    \frac{1}{|\ss{G}|}\sum_{b\in\ss{G}} (-1)^{\langle a, b\rangle} = 1[a\in\ss{C_G}].
\end{align*}
\end{lemma}

Our Pauli estimation procedure varies over states and measurements whose stabilizer groups have elements that cover some set $\ss{X}$ of errors.
This motivates the following definition.
\begin{definition}[Stabilizer coverings]
A \textit{stabilizer covering} $\ss{O}$ of a set $\ss{X}\subseteq \ss{P}^n$ is a set of stabilizer groups $\ss{O} = \{\ss{S}_j\}$ such that $\ss{X}\subseteq \bigcup_j \ss{S}_j$.
\end{definition}

For important applications, $\ss{X}$ will not necessarily be a group.
For example, $\ss{X}$ could be the set of Pauli errors with weight $\le w$, which is not generally a group.
The next lemma bounds the size of a stabilizer covering for any set $\ss{X}$, and in some cases we can improve over the trivial bound.

\begin{lemma}\label{lem:covering}
For any set $\ss{X}\in\ss{P}^n$ there exists a stabilizer covering $\ss{O}$ of $\ss{X}$ with cardinality
\begin{align}
    |\ss{O}| \le \min\bigl\{ |\ss{X}|, \sqrt{|\langle \ss{X} \rangle/\ss{S_X}|}+1\bigr\}\,,
\end{align}
where $\langle \ss{X} \rangle$ is the group generated by the elements of $\ss{X}$ and $\ss{S_X}$ is any maximal stabilizer subgroup of $\langle \ss{X} \rangle$.
In particular, the bound is tight when $\ss{X}$ is a group, and $|\ss{O}| \le 2^n +1$ unconditionally.
\end{lemma}

\begin{proof}
For any set $\ss{X}$, a simple construction of a stabilizer covering is the set of all two-element stabilizer groups generated by each nontrivial element of $\ss{X}$. 
That is, $\ss{O} = \{\langle r \rangle: r \in \ss{X}\backslash\{0\}\}$.
Clearly this set has $|\ss{O}| \le |\ss{X}|$.

Now consider a stabilizer covering of $\langle\ss{X}\rangle$, which obviously will also cover $\ss{X}$.
As a subgroup of $\ss{P}^n$, the generators can be partitioned into $s$ generators of a stabilizer group $\ss{S_X}$ and $2k$ generators of the quotient $\langle \ss{X} \rangle/\ss{S_X}$.
The quotient comprises $k$ logical qubits, so it is isomorphic to $\ss{P}^k$ and has $4^k$ elements.
Then we can construct $\ss{O}$ to be the stabilizer groups of a complete set of $|\ss{O}| = 2^k+1$ mutually unbiased bases (MUBs)~\cite{Wootters1989} acting on the $k$ logical qubits, where each MUB is taken in direct sum with $\ss{S_X}$.
The cardinality of this choice for $\ss{O}$ is just $2^k+1 = \sqrt{4^k}+1$, as claimed.

The MUBs give a minimal covering when $\ss{X}$ is a group, as there are at most $2^k-1$ non-identity Pauli elements in a stabilizer group on $k$ qubits, so at least $2^k+1$ stabilizer groups are required to cover all $4^k-1$ non-identity elements of $\ss{P}^k$.
\end{proof}

\subsection{States and measurements}

We will use the following states and measurements constructed from stabilizer groups.
Recall that $d = 2^n$. 
For every $a\in\ss{P}^n$ and every stabilizer group $\ss{S}$,
\begin{align}
\label{eq:stabilizerstate}
	\rho_{\ss{S}, a} = \frac{1}{d}\sum_{s\in\ss{S}} (-1)^{\langle s, a\rangle} P_s
\end{align}
is a valid quantum state, known as a stabilizer state.
Note that for any $a\in\ss{P}^n$ and $b\in\ss{C_S}$, $\rho_{\ss{S}, a + b} = \rho_{\ss{S}, a}$ and so we can uniquely label stabilizer states by the error syndromes $e\in\ss{A_S}$.
Moreover,
\begin{align}
\label{eq:syndromemeasurement}
	E_\ss{S} = \left\{E_{\ss{S},e} = \frac{d}{\lvert \ss{A_S}\rvert}\rho_{\ss{S},e}: e\in\ss{A_s}\right\}
\end{align}
is a valid measurement, otherwise known as a syndrome measurement.
The normalization is such that each $E_{\ss{S},e}$ is an orthogonal projector.
If $\dim\ss{S}=n$, then $\rho_{\ss{S},e}$ is a pure state and $E_{\ss{S}}$ is a rank-1 projective measurement.

\subsection{Quantum channels}

We wish to partially characterize a general linear map $\mathcal{L}:\mathbb{C}^{d\times d}\to\mathbb{C}^{d\times d}$ which has a Kraus operator representation
\begin{align}
\label{eq:genlinKraus}
	\mathcal{L}(M) = \sum_k A_k M B_k^\dagger.
\end{align}
We will generally be interested in completely positive (CP) maps, where we can choose $A_k = B_k$ for all $k$.
A specific class of channels that we frequently use are CP maps with a single unitary Kraus operator $U$.
For any unitary matrix $U\in U(d)$, we implicitly define the corresponding ideal channel $\mathcal{U}$ that maps an arbitrary matrix $M$ to $UMU^\dagger$.
Note that the implicit function is many to one, as it maps phase multiples of a unitary matrix to the same channel.
For any $a,b\in\sf{P}^n$,
\begin{align}\label{eq:pauliTransform}
    \mathcal{P}_a(P_b) = P_a P_b P_a^\dag =  (-1)^{\langle a, b\rangle} P_b.
\end{align}

We will partially characterize a general linear map by performing a Pauli twirl of the channel to reduce the effective channel to a Pauli channel, that is, to a linear map with a Kraus representation
\begin{align}
	\mathcal{E}(\rho) = \sum_{a \in \ss{P}^n} p_a P_a \rho P_a
\end{align}
where the $p_a$ are \textit{Pauli error rates} with $p_a\geq 0$ for CP maps.
For any set $\mathbb{T}\subset U(d)$, we define the $\mathbb{T}$-twirl of a channel $\mathcal{L}$ to be
\begin{align}
    \mathcal{L}^\mathbb{T} = \frac{1}{|\mathbb{T}|} \sum_{T\in\mathbb{T}} \mathcal{T} \mathcal{L} \mathcal{T}^\dagger.
\end{align}
In Kraus form, we have
\begin{align}\label{eq:kraus_twirl}
    \mathcal{L}^\mathbb{T}(M) = \frac{1}{|\mathbb{T}|} \sum_{T\in\mathbb{T}} \sum_k \mathcal{T}(A_k) \rho \mathcal{T}(B_k)^\dagger.
\end{align}
As we now show, the  Pauli-twirled channel $\mc{L}^{\ss{P}^n}$ is a Pauli channel where the error rates are directly related to the Kraus operators of the untwirled channel $\mc{L}$.
This lemma has previously appeared in \cite[Lemma 5.2.4.]{DankertThesis}, but we provide a proof here for completeness. 

\begin{lemma}[Pauli error rates]\label{lem:pauliTwirl}
For a linear map $\mc L(M)$ as in \cref{eq:genlinKraus} with Kraus operators $A_k = \sum_{a\in\ss{P}^n} l_{k,a} P_a$ and $B_k = \sum_{b\in\ss{P}^n} r_{k,b} P_b$, the Pauli error rates $p_a$ of the twirled channel $\mc{L}^{\ss{P}^n}$ are given by
\begin{align}\label{eq:effectiveErrorRates}
	p_a = \sum_k l_{k,a} r_{k,a}^*.
\end{align}
\end{lemma}

\begin{proof}
From \cref{eq:kraus_twirl},
\begin{align*}
	\mc{L}^{\ss{P}^n}(M) = |\ss{P}^n|^{-1}\sum_{c\in\ss{P}^n}\sum_k \mathcal{P}_c(A_k) M \mathcal{P}_c(B_k)^\dagger.
\end{align*}
As the Pauli matrices are a Hermitian orthogonal basis for $\mathbb{C}^{d\times d}$, we can expand the Kraus operators as $A_k = \sum_{a\in\ss{P}^n} l_{k,a} P_a$ and $B_k = \sum_{b\in\ss{P}^n} r_{k,b} P_b$, so that
\begin{align}
\begin{split}
    \mc L^{\ss{P}^n}(M)
    &= |\ss{P}^n|^{-1} \sum_{k} \sum_{a,b,c \in\ss{P}^n}  l_{k,a} r_{k,b}^* (-1)^{\langle c, a+b\rangle} P_a M P_b\\
    &= \sum_{k} \sum_{a,b \in\ss{P}^n}  l_{k,a} r_{k,b}^* 1[a=b] P_a M P_b \\
    &= \sum_{a\in\ss{P}^n}\left(\sum_k l_{k,a} r_{k,a}^*\right) P_a M P_a,
\end{split}
\end{align}
where in the second line we have used \cref{lem:pauliOrthogonality} with the fact that $\ss{C}_{\ss{P}^n} = 0$.
\end{proof}

In the most important case of a CP map, $l_{k,a} = r_{k,a}$ and the error rates are manifestly nonnegative.
However, effective negative error rates are possible in the presence of initial correlations, and this would correspond to the inner product in \cref{eq:effectiveErrorRates} being negative.
It is also noteworthy that this definition is independent of the freedom in the Kraus operators of a CP map, since the unitary freedom on the index $k$ above does not change the inner product.

We can simplify the representation of Pauli channels using the vectorization map, a linear map $| \,\cdot\, ) : X \to |X)$ that acts by stacking the columns of $X$ in a given basis and satisfies the identity $|AXB^\dagger) = B^T\otimes A |X)$. 
Note that a consequence of this is that the vectorization map can be implemented as $|X) = I\otimes X|\Gamma)$, where $|\Gamma) = \sum_j |j\rangle\otimes|j\rangle$ is a maximally entangled state and the $|j\rangle$ form an orthonormal basis.
In this notation, the dual vectors map to scalars via the Hilbert-Schmidt inner product, $(A|B) = \mathrm{Tr}(A^\dag B)$.
Then we can write the action of a channel $\mc{L}$ as a superoperator acting via left multiplication so that $\mc{L}|\rho) = \bigl|\mc{L}(\rho)\bigr)$.
Note that we overload our notation by using the same symbol for the abstract channel and the superoperator matrix.

Any Pauli channel can be expressed as a superoperator in the Pauli basis as
\begin{align}\label{eq:pauliErrors}
	\mc{E} = \frac{1}{d}\sum_{a,b \in \ss{P}^n}  (-1)^{\langle a,b \rangle} p_b |P_a)(P_a|.
\end{align}
It is convenient to make the following two simplifications.
First, we note that $|P_a)$ is not a normalized vector, so we adopt the convention that
\begin{align}\label{eq:normalization}
	|a) = \tfrac{1}{\sqrt{d}} |P_a)
\end{align}
so that $(a|b) = \tfrac{1}{d} \mathrm{Tr}(P_a^\dagger P_b) = \delta_{a,b}$.
This notation is unambiguous as long as we distinguish carefully between $a$ and $P_a$ from the given context inside the vectorization map.
Second, for any two sets of Pauli matrices $\ss{A}, \ss{B} \subseteq \ss{P}^n$, we define a Walsh-Hadamard transform $\W_{\ss{A},\ss{B}}: \mathbb{C}^{|\ss{B}|} \to \mathbb{C}^{|\ss{A}|}$ whose matrix representation is
\begin{align}
\label{eq:WHtransform}
\W_\ss{A, B} = \sum_{a\in\ss{A}, b\in\ss{B}} (-1)^{\langle a, b\rangle} |a)(b|.
\end{align}
As we prove in \cref{lem:WHtransform} below, this transformation is proportional to an isometry for specific groups $\ss{A}$ and $\ss{B}$.

Given these definitions, the vector of Pauli error rates $\bm{p}$ is related to the superoperator representation in terms of the \textit{Pauli channel eigenvalues}, or just Pauli eigenvalues for short, which are defined as,
\begin{align}
\label{eq:p-ptilde}
	\bm{f} = \W \bm{p},
\end{align}
where we have simplified the notation by using $\W = \W_{\ss{P,P}}$.
After this change of variables we have simply
\begin{align}
\label{eq:Paulichannel}
	\mc{E} = \sum_{a \in \ss{P}^n} f_a  |a)(a|.
\end{align}
With this convention, $f_0 = \sum_a p_a \le 1$ with equality if the channel is trace preserving, and all other channel eigenvalues lie inside the interval $[-f_0, f_0]$ if the channel is completely positive.
We are most interested in channels that are close to the identity in the sense that the eigenvalues are in some small, strictly positive interval $[\epsilon,1]$, which is appropriate when the noise is weak.
Consequently, for any $h\in\ss{P}^n$ and any set $\ss{A}\subseteq\ss{P}^n$ we define
\begin{align}\label{eq:defr}
    r_h = 1 - f_h, \quad
    \bm{r}_\ss{A} = \sum_{a\in\ss{A}} r_a|a), \quad
    \bm{1}_\ss{A} = \sum_{a\in\ss{A}} |a).
\end{align}
We will omit the subscripts from the vector notation when $\ss{A} = \ss{P}^n$.

In what follows, our estimation strategy will be to infer both $\bm{r}$ and, by then doing the inverse transform, $\bm{p} = \W^{-1}( \bm{1} - \bm{r})$, where we can evaluate the inverse using the following lemma.

\begin{lemma}
\label{lem:WHtransform}
For any subgroups $\ss{A},\ss{B}\subseteq\ss{P}^n$, $\W_\ss{A, B}$ satisfies
\begin{align*}
	\W_\ss{A, B}^\dagger \W_\ss{A, B} = |\ss{A}|\sum_{b,b'\in\ss{B}:b+b'\in\ss{C_A}} |b')(b|.
\end{align*}
In particular, when $\ss{C_A \cap B}$ is trivial, we have that $\W_\ss{A, B}$ is proportional to an isometry, with $\W_\ss{A, B}^\dagger \W_\ss{A, B}= |\ss{A}| \Pi_\ss{B}$ where $\Pi_\ss{B} = \sum_{b \in \ss{B}} |b)(b|$ is the projector onto $\ss{B}$.
\end{lemma}

\begin{proof}
From \cref{eq:WHtransform},
\begin{align*}
	\W_\ss{A, B}^\dagger \W_\ss{A, B}
	&= \sum_{a,a'\in\ss{A}}\sum_{b,b'\in\ss{B}} (-1)^{\langle a,b\rangle + \langle a',b'\rangle} |b')(a'|a)(b| \\
\	&= \sum_{b,b'\in\ss{B}} \left(\sum_{a\in\ss{A}}(-1)^{\langle a,b+b'\rangle}\right) |b')(b| \\
    &= |\ss{A}|\sum_{b,b'\in\ss{B}} 1[b'+b\in \ss{C_A}] |b')(b|,
\end{align*}
where we have used \cref{lem:pauliOrthogonality} to obtain the third line.

When $\ss{B}$ is a group and $\ss{C_A \cap B}$ is trivial, $b+b'\in\ss{C_A}$ if and only if $b'= -b = b$.
\end{proof}

Finally, we can evaluate the quality of a reconstructed set of errors $\hat{\bm p}$ either in terms of a norm $\| \hat{\bm p} - \bm{p} \|$ or a figure of merit on the space of linear maps.
The natural norm would be the 1-norm, which is related to the average gate infidelity $r(\mc{E})$ and the diamond distance to the identity $\epsilon_\diamond(\mc{E})$ by
\begin{align}\label{eq:WF}
\tfrac{1}{2} \||0)-\bm{p}\|_1 = \epsilon_\diamond(\mc{E}) = (1+1/d)r(\mc{E}),
\end{align}
which holds for Pauli channels~\cite{Magesan2012a}. 
Here we have $r(\mc{E}) = 1-\int \mathrm{d}\psi \tr\bigl[\psi\mc{E}(\psi)\bigr]$, where $\mathrm{d}\psi$ is the uniform (Haar) measure over all pure state projectors $\psi$ and $\epsilon_\diamond(\mc{E}) = \sup_{\rho} \frac{1}{2}\|\rho-\mc{E}\otimes\mc{I}(\rho)\|_1$, where the supremum is over all valid quantum states on an enlarged space with an arbitrary-sized ancilla. 

However, motivated by physical considerations we will project reconstructions into the nearest point in the set of probability distributions, or potentially to the set of subnormalized probability distributions.
The nearest point is not uniquely defined in general according to some norms, such as the 1-norm, and so in some cases we will instead use the 2-norm for which the nearest point is uniquely defined.

\subsection{Model assumptions}\label{sec:modelassumptions}

Our procedure makes use of the following primitives: preparations of stabilizer states, Pauli transformations, and syndrome measurements.
An experimental implementation of our procedure will necessarily involve noisy versions of all primitives.
For clarity, we always denote the noisy version of a primitive with an overset $\sim$, e.g., $\tilde{A}$ is a noisy implementation of the ideal operation $A$ (whether $A$ is a state preparation, a measurement operator, or a channel).

For ease of analysis, we assume that the noise on the primitives is independent, time-stationary, and Markovian.
In particular, we assume that the noisy preparation of a state $\rho$ and a noisy syndrome measurement of a stabilizer group $\ss{H}$ are independent of the circuit to be applied, so that they can be written as a density operator $\tilde{\rho}$ and a positive-operator-valued measure (POVM) $\{\tilde{E}_{\ss{H},e}:e\in\ss{E_H}\}$ respectively.

Similarly, we require that the noise in the implementations of a set of twirling channels $\mathcal{T}$ are independent of the particular twirl being implemented and the remainder of the circuit and so can be written as $\noisy{T} = \mathcal{T}\Lambda$ for some fixed completely positive (CP) map $\Lambda$.
We codify this in the following definition.

\begin{definition}[GTM noise]
\label{def:GTMnoise}
A noise model is time-stationary if the noisy implementation $\noisy{U}(t)$ of a gate $\mathcal{U}$ at time $t$ is a linear map that is independent of $t$ and if state preparations and measurements are respectively described by fixed density operators and POVMs.
A noise model for the Pauli group is called \emph{GTM} (gate-independent, time-stationary, Markovian) if it is time-stationary and there exists a completely positive trace-preserving map $\Lambda$ such that $\noisy{P} = \mathcal{P}\Lambda$ for all $\mathcal{P}\in\ss{P}^n$.
\end{definition}

The above assumption is routinely assumed in analyses of RB for groups $\sf{X}$ that contain complex circuits of multi-qubit gates and can be relaxed with sufficient effort~\cite{Chasseur2015,Proctor2017,Wallman2018}.
In contrast, we only make the assumption for groups consisting of tensor products of channels acting on individual qubits.

The other assumption we will make on the noise is that it is sufficiently weak that we can ignore certain algebraic limitations in the use of our specific estimators.
In particular, we will assume that our noise and the state preparation and measurements (SPAM) satisfy the following definitions for choices of a parameter $c$ that we will specify later.

\begin{definition}[Weak, stable]
\label{def:weakstable}
A noise map $\mathcal{L}$ is \textbf{$\bm{c}$-weak} if the Pauli twirl $\mathcal{L}^{\ss{P}^n}$ is close to the identity channel in the operator norm, $\|\mathcal{I} - \mathcal{L}^{\ss{P}^n}\| \le c$.
A SPAM parameter $A$ is called \textbf{$\bm{c}$-stable} if $A \ge 1-c$.
\end{definition}

The first definition is equivalent to saying that the Pauli twirl $\mathcal{L}^{\ss{P}^n}$ has Pauli eigenvalues $\bm{f}$ that all lie in the interval $[1-c,1]$.
It is likely that this assumption too can be relaxed, but to do so would require different estimators than the ratio estimator that we use below.
We note that a simple sufficient condition to ensure that a CP map is $c$-weak is that the Pauli error rate for the identity, $p_0$, obeys $p_0 \ge 1-\tfrac{c}{2}$.
This can be verified from the inequality $f_i \ge 2p_0 - 1$~\cite{Erhard2019}.
The second definition will become clearer once we define (in \cref{eq:SPAM} below) how the SPAM parameters $A_j$ in the definition depend on the noisy state preparations and measurements. 
Intuitively, it just ensures that the measurements are sufficiently good to give a reasonable signal, so a $c$-stable SPAM parameter is $c$-close to being ideal.

\section{Estimation procedure and sample complexity}\label{sec:estimation}

We now specify a procedure for estimating the Pauli eigenvalues of a noisy implementation of the Pauli group $\ss{P}^n$, where we prepare states using a fixed stabilizer group $\ss{G}$ and we make syndrome measurements of a second, possibly different, stabilizer group $\ss{H}$.
As we will see below, for this procedure to yield useful information about the noise we will choose $\ss{H} \subseteq \ss{G}$, but we do not require this choice in our proofs.
The backbone of the procedure is running a generalized cycle benchmarking sequence of length $m$~\cite{Erhard2019} and recording the result of a syndrome measurement.
One of the core innovations of the present procedure is to consider measurements of general stabilizer groups, which enables us to obtain more information from each experiment.
Character benchmarking~\cite{Helsen2018} over the Pauli group, direct RB~\cite{Proctor2018} with the Pauli group as the generator distribution, and cycle benchmarking with no interleaved gate~\cite{Erhard2019} can all be rephrased in terms of the following more general procedure by choosing the measurement to be of a two-element stabilizer group, that is, a group generated by a single Pauli operator.

\begin{alg}[${\bf RunCB}(\ss{G},\ss{H},m)$]
\item Prepare the approximate stabilizer state $\tilde{\rho}_{\ss{G},0}$.

\item For each $i = 0,\ldots,m$, apply a Pauli gate $P_{a_i}\in\ss{P}^n$ uniformly at random.

\item Perform a syndrome measurement of $\ss{H}$ and record the outcome $b \in \ss{A_H}$ from the noisy POVM element $\tilde{E}_b$.

\item Return $z\in\ss{A_H}$ such that $b+\sum_i a_i\equiv z$.
\end{alg}

We will use the shorthand ${\bf RunCB}(\ss{G},m) = {\bf RunCB}(\ss{G}, \ss{G},m)$ for the special case $\ss{G} = \ss{H}$.
The following proposition exactly characterizes the output of ${\bf RunCB}(\ss{G}, \ss{H},m)$ under gate-independent and time-stationary noise.
We note that ${\bf RunCB}(\ss{G}, \ss{H},m)$ can be straightforwardly generalized to include interleaved gates~\cite{Erhard2019}; this will be explored in future work.

\begin{proposition}
\label{prop:PrZ}
Let $Z = Z(\ss{G},\ss{H},m)$ be the random variable taking values $z \in \ss{E_H}$ that results from one call to the subroutine ${\bf RunCB}(\ss{G},\ss{H},m)$.
Under GTM noise with $\tilde{\mc{P}} = \mc{P}\Lambda$ for all $P\in\ss{P}^n$, $Z$ has the probability distribution function
\begin{align}\label{eq:likelihood}
	\Pr\bigl(Z=z\bigr)
	= \frac{1}{|\ss{A_H}|}\sum_{h \in \ss{H}} f_h^m (-1)^{\langle h,z\rangle} A_h
\end{align}
where the $f_h^m = (h|\Lambda|h)^m$ are the $m$th powers of the Pauli eigenvalues of $\Lambda$, and
\begin{align}\label{eq:SPAM}
    A_h =  \sum_{b\in\ss{A_H}} (-1)^{\langle h,b\rangle} (\tilde{E}_{\ss{H},b}|h)(h|\Lambda|\tilde{\rho}_{\ss{G},0}).
\end{align}
\end{proposition}

\begin{proof}
By the linearity of the Born rule for density matrices, the probability of a given instantiation $Z=z$ is the sum of all the probabilities of random noisy Pauli gates and measurement outcomes that would lead to that result.
This is equivalent to first writing the entire joint distribution over any possible sequence, and then marginalizing over the specific sequences of Pauli gates and measurement outcomes that give the same value of $z$.

Under the assumption of GTM noise, the probability of choosing $a_0,a_1,\ldots,a_m$ and obtaining the outcome $b$ is
\begin{align*}
    \mathrm{Pr}(a_0,\ldots, a_m,b) = |\ss{P}^n|^{-m-1} (\tilde{E}_{\ss{H},b}|\left(\prod_{i=m}^0 \mc{P}_{a_i} \Lambda \right)|\tilde{\rho}_{\ss{G},0}),
\end{align*}
where for non-commutative products we use the convention $\prod_{i=a}^b x_i = x_a \ldots x_b$ (i.e., the lower index is on the left of the product).
We can rewrite the probability as
\begin{align*}
    \mathrm{Pr}(a_0,\ldots, a_m,b) = |\ss{P}^n|^{-m-1} (\tilde{E}_{\ss{H},b}|\mc{P}_{a'_m}\left(\prod_{i=m-1}^0 \mc{P}_{a'_i} \Lambda \mc{P}_{a'_i} \right)\Lambda|\tilde{\rho}_{\ss{G},0}),
\end{align*}
where we set $a'_0= a_0$ and recursively define $a'_i = a_i + a'_{i-1}$ for $i=1,\ldots, m-1$, where $a'_m = \sum_i a_i$ (recall that addition is mod 2).
Averaging uniformly and independently over the $a_i$ is equivalent to averaging uniformly and independently over the $a'_i$ because $\ss{P}^n$ is a group.
Therefore the marginal probability of $a'_m\in\ss{P}^n$ and $b\in\ss{E_H}$ is
\begin{align*}
    \mathrm{Pr}(a'_m,b) &= |\ss{P}^n|^{-m-1} \sum_{a'_0,\ldots,a'_m\in\ss{P}^n}  (\tilde{E}_{\ss{H},b}|\mc{P}_{a'_m}\left(\prod_{i=m-1}^0 \mc{P}_{a'_i} \Lambda \mc{P}_{a'_i} \right)\Lambda|\tilde{\rho}_{\ss{G},0}) \\
    &=  |\ss{P}^n|^{-1} (\tilde{E}_{\ss{H},b}|\mc{P}_{a'_m}\left(\prod_{i=m-1}^0 \Lambda^{\ss{P}^n} \right)\Lambda|\tilde{\rho}_{\ss{G},0}) \\
    &= |\ss{P}^n|^{-1} \sum_{h\in\ss{P}^n} f_h^m (\tilde{E}_{\ss{H},b}|\mc{P}_{a'_m}|h)(h|\Lambda|\tilde{\rho}_{\ss{G},0})\\
    &= |\ss{P}^n|^{-1} \sum_{h\in\ss{P}^n} (-1)^{\langle h, a'_m\rangle} f_h^m A_{b,h}
\end{align*}
by \cref{eq:Paulichannel,eq:pauliTransform}, where in particular we use  $(h|\Lambda^{\ss{P}^n}|h) = (h|\Lambda|h) = f_h$ and we define $A_{b,h} = (\tilde{E}_{\ss{H},b}|h)(h|\Lambda|\tilde{\rho}_{\ss{G},0})$.

Splitting $a'_m = a + c$ where $a\in \ss{A_H}$ and $c\in\ss{C_H}$ and averaging over $c$ gives the marginal probability
\begin{align*}
\mathrm{Pr}(a,b) &= \sum_{c\in\ss{C_H}} \mathrm{Pr}(a+c,b) \\
&= |\ss{A_H}|^{-1} \sum_{h\in\ss{P}^n} (-1)^{\langle h, a\rangle} f_h^m A_{b,h} |\ss{C_H}|^{-1} \sum_{c\in\ss{C_H}} (-1)^{\langle h, c\rangle} \\
&= |\ss{A_H}|^{-1} \sum_{h\in\ss{P}^n} (-1)^{\langle h, a\rangle} f_h^m A_{b,h}\, 1\bigl[h\in\ss{C_{C_H}}\bigr] 
\end{align*}
by \cref{lem:pauliOrthogonality}.
Now by the double commutant theorem we have $\ss{C_{C_H}} = \ss{H}$ for any group $\ss{H}\subseteq\ss{P}^n$, so
\begin{align*}
\mathrm{Pr}(a,b) &= |\ss{A_H}|^{-1} \sum_{h\in\ss{H}} (-1)^{\langle h, a\rangle} f_h^m A_{b,h}\,.
\end{align*}
The probability of obtaining the outcome $z\in\ss{A_H}$ is then
\begin{align*}
    \mathrm{Pr}(Z=z) &= \sum_{a,b\in\ss{A_H}:a+b=z} \mathrm{Pr}(a,b) \\
    &= \sum_{b\in\ss{A_H}} \mathrm{Pr}(z+b,b) \\
    &= |\ss{A_H}|^{-1} \sum_{b\in\ss{A_H}}\sum_{h\in\ss{H}} (-1)^{\langle h, b+z\rangle} f_h^m A_{b,h} \\
    &= |\ss{A_H}|^{-1}\sum_{h\in\ss{H}} (-1)^{\langle h, z\rangle} f_h^m A_h
\end{align*}
which completes the proof.
\end{proof}

We remark that the coefficients $A_h$ appearing in \cref{prop:PrZ} are $1$ for ideal states, measurements, and transformations.
Indeed, setting $\tilde{\rho}_{\ss{G},0} = \rho_{\ss{G},0}$, $\Lambda = 1$ and $\tilde{E}_{\ss{H},b} = E_{\ss{H},b}$ for any fixed outcome $b\in\ss{A_H}$ in \cref{eq:SPAM} and substituting in \cref{eq:stabilizerstate,eq:syndromemeasurement}, we have that
\begin{align*}
	(E_{\ss{H},b}|h)(h|\rho_{\ss{G},0}) & = \frac{1}{d|\ss{A_H}|}\sum_{g\in\ss{G}}\sum_{h'\in\ss{H}} (-1)^{\langle h',b\rangle} (P_{h'}|h)(h|P_g)\\
	& = \frac{(-1)^{\langle h, b\rangle}}{|\ss{A_H}|} 1[h\in\ss{G}\cap\ss{H}],
\end{align*}
where we have used the normalization from \cref{eq:normalization}.
Plugging this into \cref{eq:SPAM}, we find that in the ideal case,
\begin{align*}
	A_h & = 1[h\in\ss{G}\cap\ss{H}] . \tag{ideal case}
\end{align*}
We refer to the coefficients $A_h$ as SPAM coefficients (despite having a contribution from edge noise terms) because the coefficients are the only model parameters affected by errors in the state preparations and measurements.
We see that the information contained in the likelihood function is maximized when $\ss{H}=\ss{G}$.

\Cref{prop:PrZ} exactly quantifies the output of one call to ${\bf RunCB}(\ss{G},\ss{H},m)$.
One could naively obtain many samples of ${\bf RunCB}(\ss{G},\ss{H},m)$ for different values of $m$ and perform a multi-exponential fit to estimate the model parameters $\{A_h,f_h:h\in\ss{H}\}$.
However, there are exponentially many such parameters, making the fit and a rigorous theoretical treatment nigh impossible.
We can resolve this problem by transforming the output using the following lemma to obtain a new random variable with a single exponential decay as in \cite{Helsen2018}.
We note that considering syndrome measurements instead of the two-outcome measurements considered in  Refs.~\cite{Helsen2018,Erhard2019} allows us to estimate multiple random variables from the same data using the following procedure, which can significantly reduce the sample complexity.

\begin{alg}[${{\bf V}}(\ss{X}, \ss{G},t,m)$: Estimator function for vector of SPAM-dependent Pauli eigenvalues.]
\item Set $\bm{V} \coloneqq 0$.
\item For $k = 1,\ldots,t$, Do
\begin{itemize}
	\item Set $z \coloneqq {\bf RunCB}(\ss{G},m)$,
	\item Set $\bm{V} \mathrel{+}= \sum_{x\in\ss{X}} (-1)^{\langle x,z \rangle}|x)$.
\end{itemize}
\item Return $\hat{\bm{V}} = \frac{1}{t} \bm{V}$.
\end{alg}

The elements in the output of ${\bf V}(\ss{X},\ss{G},t,m)$ are correlated binomial variables, but their covariance can also be exactly computed using the following lemma.

\begin{lemma}
\label{lemma:Fmean}
Let $\bm{V}$ be the output of ${\bf V}(\ss{X}, \ss{G}, t, m)$.
Under GTM noise, we have
\begin{align*}
    \bm{\mu} = \mathbb{E}[\bm{V}] &= \sum_{x\in\ss{X}} A_{x^\bot} f_{x^\bot}^m|x) \\
    \mathbb{E}\bigl[\bm{V} \bm{V}^\dagger\bigr] &= \sum_{x,x'\in\ss{X}} 1[x+x'\in\sf{P}^n/\ss{G}]\, A_{x+x'} f_{x+x'}^m|x)(x'|
\end{align*}
where $x^\bot$ is the component of $x$ in $\ss{P}^n / \ss{G}$ and the $A_x$ are as in \cref{prop:PrZ}, and, in particular, are independent of $m$.
\end{lemma}

\begin{proof}
As the calls to ${\bf RunCB}(\ss{G},m)$ are independent, the expectation value of $\bm{V}$ is
\begin{align*}
    \bm{\mu} = \mathbb{E}[\bm{V}] = \sum_{x\in\ss{X}} \sum_{z\in\ss{A_G}} (-1)^{\langle x,z\rangle} \mathrm{Pr}(Z=z)|x).
\end{align*}
For any fixed $x\in\ss{X}$, substituting in \cref{prop:PrZ} with $\ss{H}=\ss{G}$, noting we ${\bf RunCB}(\ss{G},m)$), and using \cref{lem:pauliOrthogonality} gives
\begin{align*}
    (x|\bm{\mu} &= \frac{1}{|\ss{A_G}|}\sum_{z\in\ss{A_G}} \sum_{g\in\ss{G}} (-1)^{\langle x+g, z\rangle} f_g^m A_g \\
    &= \sum_{g\in\ss{G}} 1[x+g\in\ss{C_{A_G}}] f_g^m A_g.
\end{align*}
Note that for a stabilizer group $\ss{G}$, $\ss{C_{A_G}} = \ss{P}^n / \ss{G}$.
The result then follows as $\ss{G}$ is a group and $\ss{P}^n$ splits into $\ss{G} \oplus \ss{P}^n/\ss{G}$.

Similarly, for fixed $x,x'\in\ss{X}$ we have
\begin{align*}
    (x|\mathbb{E}\bigl[\bm{V} \bm{V}^\dagger\bigr]|x') = \sum_{z\in\ss{A_H}} (-1)^{\langle x+x',z\rangle} \mathrm{Pr}(Z=z) = (x+x'|\mathbb{E}[\bm{V}],
\end{align*}
which completes the proof.
\end{proof}

Since the elements of $\bm{V}$ are bounded, an average of sufficiently many independent samples will converge quickly to the mean.
The next proposition provides a simple and conservative tail bound on the probability of estimates obtained using $T$ samples will deviate from their mean by more than $\epsilon$.
Specifically, we bound the failure probability of the following procedure, where we allow $\ss{X}\subseteq\ss{G}$ both for later analysis and to illustrate the fundamental scaling.

\begin{proposition}
\label{prop:Tsuffice}
Let $\ss{G}$ be a stabilizer group, $\ss{X}\subseteq\ss{G}$, and $\hat{\bm{V}}$ be the output of $\bm{V}(\ss{X},\ss{G},t,m)$ for some fixed positive integers $t$ and $m$.
Then for any $\epsilon >0$,
\begin{align*}
    \mathrm{Pr}(\|\hat{\bm{V}} - \mathbb{E}(\hat{\bm{V}}) \|_\infty \geq \epsilon) \leq 2 |\ss{X}|\exp(-t\epsilon^2/2).
\end{align*}
\end{proposition}

\begin{proof}
The proof follows trivially from the union bound and Hoeffding's inequality~\cite{Hoeffding1963} where each element of $\hat{\bm{V}}$ is in the interval $[-1,1]$.
\end{proof}

Note that the above bound is very loose because, as shown in \cref{lemma:Fmean}, there are correlations between the elements of $\bm{\hat{V}}$ that are neglected when applying the union bound.
However, we still see that the sample complexity above in estimating $A_j f_j^m$ is independent of the Pauli eigenvalues.
The precision in estimating the eigenvalues $f_j$ with a fixed failure probability will, however, depend on the bare precision $\epsilon$ and on the specific values of $m$ that we choose.
In the next section we will analyze a specific procedure and estimator to get bounds on this dependency.
\section{Data fitting and error analysis}\label{sec:analysis}

To understand how the sequence lengths should be chosen to get a precise estimate of a Pauli channel, we must introduce a specific estimation strategy.
We use an estimator that samples $A f^m$ at exponentially increasing values of $m$, and then we use only the data at the endpoints to obtain a ratio estimator.
We also ignore repeated elements from distinct stabilizer groups in a stabilizer covering of a set.
Our estimation strategy throws away a lot of data, but it has two advantages.
First, it is easy to analyze explicitly, and our proof is a straightforward adaption of the theorem presented in Ref.~\cite{Harper2018}.
Second, the endpoints are where the data are most sensitive, so discarding the intermediate data is not as damaging as one might expect since not all points contribute equally to the variance in the estimator.
In what follows we use the notation $V_i(g) = (g|V_i$ for the $g$th element of $V_i$ since the subscript real estate is occupied by the iteration index $i$.

\begin{alg}
[${\bf Ratio}(\ss{O},\ss{X},t)$ Ratio estimator for exponential regression of a stabilizer covering $\ss{O}$ of a set $\ss{X}$.]
\item Set $\hat{r}_x = \texttt{NaN}$ for all $x\in\ss{X}$.
\item Set $\ss{A} \coloneqq \emptyset$.
\item For all $\ss{G} \in \ss{O}$, Do
\begin{enumerate}
\item Set $\ss{X}' \coloneqq \ss{G}\backslash \ss{A}$.
\item Set $\ss{A}\coloneqq\ss{A}\cup\ss{G}$.
\item Set $\hat{V} \coloneqq {\bf V}(\ss{X}',\ss{G},t,0)$.
\item Set $m \coloneqq 1$.
\item While $\exists x \in \ss{X}'$ such that $\hat{r}_x= \texttt{NaN}$, Do
\begin{itemize}
	\item Set $\hat{W} \coloneqq {\bf V}(\ss{X}',\ss{G},t,m)$,
	\item For all $x \in \ss{X}'$ such that $\hat{r}_x=\texttt{NaN}$, Do
	\begin{itemize}
	    \item Set $v \coloneqq \hat{V}(x)$ and $w \coloneqq \hat{W}(x)$.
	    \item If $w \leq v / 3$ and $w,v>0$, Set $\hat{r}_x \coloneqq 1 - \bigl(w/v\bigr)^{1/m}$;
	    \item Else If $w\le0$ or $v\le0$, Set $\hat{r}_x \coloneqq 1$.
	\end{itemize}
    \item Set $m \coloneqq 2 m$
\end{itemize}
\end{enumerate}
\item Return $\hat{\bm{r}}_\ss{X} = \sum_{x\in\ss{X}} \hat{r}_x |x)$.
\end{alg}

This ratio estimator guarantees a pointwise multiplicative precision estimate of the true vector of infidelities $\bm{r}_\ss{X}$ (recall \cref{eq:defr}), at least when certain mild assumptions hold that place us close to the regime of interest.
This is where the notions of $c$-weak and $c$-stable from \cref{sec:modelassumptions} come into play.
We will assume that the noise map is $\tfrac{1}{2}$-weak, which in particular implies that $f_h \geq \tfrac{1}{2}$, or equivalently, that $r_h \leq \tfrac{1}{2}$.
We will also assume that each SPAM parameter $A_h$ from \cref{eq:SPAM} is $\tfrac{1}{2}$-stable, so that $A_h \ge \tfrac{1}{2}$.
However, the only thing that our proof really requires is that both $A$ and $f$ are bounded from below by a positive constant, and the specific choice of $\tfrac{1}{2}$ is motivated only by convenience.

We note that an equivalent theorem to the following can be proven if we remove the $\tfrac{1}{2}$-stable assumption and instead assume that $\epsilon$ is chosen to give a relative precision proportional to $\min_h A_h f_h$ at each step instead of only a constant.
However, this makes the statements about sample complexity dependent on the SPAM parameters $A_h$.
This obfuscates the actual sample complexity of the procedure, so we prefer to add the stability assumption.

To understand how the set of sequence lengths $\kappa$ should be chosen to get a precise estimate of a Pauli channel, we must introduce the spectral gap of the channel, or of a subset of the channel corresponding to the Paulis in the set $\ss{X}$.
Since a Pauli channel is already diagonal in the Pauli basis, as in Eq.~(\ref{eq:Paulichannel}), and every trace-preserving channel has a largest absolute eigenvalue of 1, the spectral gap $\Delta_{\ss{X}} = \Delta_{\ss{X}}(\Lambda)$ of the Pauli channel $\Lambda$ over the set $\ss{X}$ is given by
\begin{align}
\label{eq:channelspectralgap}
	\Delta_{\ss{X}} = 1-\max_{j\in\ss{X}\backslash 0} |f_j| .
\end{align}
In the most interesting regime, $\Lambda$ is $\epsilon$-weak for some $\epsilon < 1$, and the absolute value is not necessary because then all eigenvalues are positive.
In that case, we have $\Delta_\ss{X} = \min_{x\in\ss{X}} r_x$.
When $\ss{X}=\ss{P}^n$, this is exactly the spectral gap of the superoperator for the channel.
Note that $\Delta_\ss{X} \ge \Delta_{\ss{P}^n}$.

As we will now show, one can completely learn all the Pauli eigenvalues in $\ss{X}$ with relative precision $\epsilon$ using a set $\kappa$ of sequence lengths with $|\kappa| = O\bigl(\log \tfrac{1}{\Delta_\ss{X}}\bigr)$ and $m_{\max} = \max \kappa = O\bigl(\tfrac{1}{\Delta_\ss{X}}\bigr)$.


\begin{proposition}
\label{prop:RatioProp}
Let $\ss{X}\subseteq\ss{P}^n$ and $\ss{O}$ be a stabilizer covering of $\ss{X}$.
For any sufficiently small $\epsilon,\delta>0$ and assuming $\tfrac{1}{2}$-weak, $\tfrac{1}{2}$-stable, GTM noise, the following holds with probability $1-\delta$.
Running ${\bf Ratio}(\ss{O},\ss{X},t)$ with $t = \frac{2}{\epsilon^2} \log\bigl(\frac{2|\ss{X}||\kappa|}{\delta}\bigr)$ uses a set $\kappa$ of at most $O\bigl(\log\tfrac{1}{\Delta_\ss{X}}\bigr)$ sequence lengths $m$ with $m_{\max} = O\bigl( \tfrac{1}{\Delta_\ss{X}}\bigr)$.
Moreover, the output $\hat{\bm{r}}_\ss{X}$ satisfies
\begin{align*}
    |\hat{r}_x - r_x| \leq O(\epsilon) r_x\,.
\end{align*}
\end{proposition}

\begin{proof}
The following proof of multiplicative precision is a modification of a similar result proven recently in Ref.~\cite{Harper2018}.

We begin by considering each iteration of the loop in step 3.
Let $t$ be a fixed integer, $\ss{G}\in\ss{O}$ be a fixed group, and $\ss{X}'$ be the value in step 3a in the corresponding iteration of ${\bf Ratio}(\ss{O},\ss{X},t)$.
The inner loop (e) yields estimators $\hat{\bm{V}}_m = {\bf V}(\ss{X}', \ss{G}, t, m)$ for each $m\in\kappa = \{0, 1, 2, 4,\ldots, m_{\rm max}\}$ for some yet-unspecified terminating sequence length $m_{\rm max}$.
Note that we do not store all these intermediate estimators in the procedure as we only use them to prove correctness.
 
By \cref{prop:Tsuffice} and the union bound, for any fixed $\epsilon>0$, 
\begin{align*}
|\hat{V}_m(x) - A_x f_x^m |\leq \epsilon
\end{align*}
with probability at least $1-\delta_{\ss{X}'}$ where $\delta_{\ss{X}'} = 2|\kappa||\ss{X}'| \exp(-t\epsilon^2/2)$.
To bound the terminating value $m_{\rm max}$ for each $\ss{G}\in\ss{O}$ and hence prove the stated sample complexity, fix $g\in\ss{G}$, let $m$ be the corresponding value in the procedure when $\hat{r}_x$ is assigned and $w = \hat{V}_m(\ss{X'}, \ss{G}, t, m)$ and $u = \hat{V}_m(\ss{X'}, \ss{G}, t, \lfloor m/2\rfloor)$.
Then $u,v,w$ satisfy the following inequalities
\begin{align}\label{eq:mess}
	\frac{w}{v} \le \frac{1}{3} 
	< \frac{u}{v}.
\end{align}

For simplicity, we now use the notation $f=f_g$ and $A=A_g$.
We also assume that $m > 1$ as otherwise the proof is trivial.
Let us define the deviations of $u,v,w$ from their mean,
\begin{align*}
    w-Af^m =:\epsilon_w\,,\quad u-Af^{m/2} =:\epsilon_u\,, \quad \text{and } v-A =:\epsilon_v\,.
\end{align*}
By \cref{prop:Tsuffice}, $|\epsilon_x| \le \epsilon$ for $x=u,v,w$. 
Using the fact that $A \ge 1/2$ for $\tfrac{1}{2}$-stable noise,
\begin{align}\label{eq:mess1}
    \frac{w}{v} 
    = \frac{Af^m + \epsilon_w}{A + \epsilon_v}
    = \left(1 + \tfrac{\epsilon_v}{A}\right)^{-1} \left(f^m + \tfrac{\epsilon_w}{A}\right)
    &\geq (1+2\epsilon)^{-1}\left(f^m - 2\epsilon\right).
\end{align}
Similarly,
\begin{align}\label{eq:mess2}
    \frac{u}{v}
    &\leq (1-2\epsilon)^{-1}\left(f^{m/2} + 2\epsilon\right) \\
    \frac{w}{v}
    &\leq (1-2\epsilon)^{-1}\left(f^m + 2\epsilon\right) \label{eq:mess3}.
\end{align}
Substituting \cref{eq:mess1,eq:mess2} into \cref{eq:mess}, we have
\begin{align}
\label{eq:fbounds}
	\biggl(\frac{1-8\epsilon}{3}\biggr)^2 < f^m \le \frac{1+8\epsilon}{3}.
\end{align}

We have that $r = 1-f \ge \Delta_{\ss{X}}$, and note that $f^m = (1-r)^m$ will satisfy \cref{eq:fbounds} for fixed $\epsilon$ if and only if $m = \Theta(1/r)$.

To show that $|\hat{r} - r|\le O(\epsilon) r$, we have to analyze the accuracy of the estimator $(w/v)^{1/m}$.
Using \cref{eq:mess1,eq:mess3}, we have
\begin{align}
	\biggl(\frac{f^m + 2\epsilon}{1 - 2\epsilon}\biggr)^{1/m} \geq \hat{f} \geq \biggl(\frac{f^m - 2\epsilon}{1 + 2\epsilon}\biggr)^{1/m}.
\end{align}
We can factor out the $f^m$ (since $f>0$ as we have assumed the noise is $1/2$-weak) and use the lower bound from \cref{eq:fbounds} to find
\begin{align}
	f \biggl(\frac{1 - 18\epsilon/(1-8\epsilon)^2}{1 +2\epsilon}\biggr)^{1/m} < \hat{f} < f\biggl(\frac{1 + 18\epsilon/(1-8\epsilon)^2}{1-2\epsilon}\biggr)^{1/m}\,.
\end{align}
For sufficiently small $\epsilon$, this implies
\begin{align}
	f \bigl(1-O(\epsilon)\bigr)^{1/m} \le \hat{f} \le f \bigl(1+O(\epsilon)\bigr)^{1/m}\,.
\end{align}
Now we Taylor expand around $\epsilon = 0$ for sufficiently small $\epsilon$ and use $\tfrac{1}{m} = \Theta(r)$ to obtain the stated accuracy.

Now let $g\in\ss{X}$ be the element that requires the largest value of $m_{\rm max}$.
Then the set of sequence lengths $\kappa$ used to estimate all the $f_a$ to within the stated precision with probability at least $1-\delta_{\ss{X}'}$ is $\kappa = \{2^i:i=0,\ldots,\ell=\log_2(m_{\rm max})\}$.
As $m_{\rm max} = \Theta(1/r_g)$, we have $|\kappa|-1 = \log_2 m_{\rm max} = O\bigl(\log \frac{1}{\Delta_\ss{X}}\bigr)$, which is well defined as the noise is $1/2$-weak by assumption.
The total failure probability is then at most $\delta$ where $\delta = 2|\kappa||\ss{X}|\exp(-2t\epsilon^2/2)$ by the union bound.
Rearranging gives the stated sample complexity.
\end{proof}

\section{Reconstructing error rates for a group}

We have seen how we can estimate each of the Pauli eigenvalues $f_a$ of a Pauli channel for all $a\in\ss{X}$ with a number of measurements that scales like $O\bigl(|\ss{O}| \log |\ss{X}|\bigr)$ for any stabilizer covering $\ss{O}$ of $\ss{X}$, assuming that $\epsilon$, $\delta$ and the channel spectral gap $\Delta_\ss{X}$ are fixed.
When $\ss{X}$ is a group, we can choose $\ss{O}$ such that $|\ss{O}| \le \sqrt{|\ss{X}|}+1$ using the construction in \cref{lem:covering}, so the total number of samples is at most $O(\sqrt{|\ss{X}|} \log |\ss{X}|)$ (or even less if $\ss{X}$ has a nontrivial stabilizer subgroup).
Moreover, the precision of these estimates is multiplicative.
We now show how the estimates of Pauli eigenvalues can be used to estimate Pauli error rates.

Suppose we estimate all the Pauli eigenvalues of some group $\ss{G}$ (which could be the full Pauli group).
Then \cref{eq:p-ptilde} becomes
\begin{align}
    \bm{f}_{\ss{G}} = \W_{\ss{G},\ss{P}^n} \bm{p},
\end{align}
where we use the shorthand $\bm{v}_{\ss{G}}$ to denote the vector with entries $\{v_g: g\in\ss{G}\}$.
The columns of $\W_{\ss{G},\ss{P}^n}$ for Paulis that differ by an element of $\ss{C_G}$ are identical and so the corresponding Pauli error rates cannot be distinguished using only the estimated Pauli eigenvalues.
That is, we can only reconstruct a \textit{marginal} or coarse-grained probability distribution $\bm{p}_{\ss{A_G}}$ over $\ss{A_G}$ with marginal probabilities
\begin{align}
    p_{\ss{A_G},a} = \sum_{c\in \ss{C_G}} p_{a+c}
\end{align}
for $a\in\ss{A_G}$ via
\begin{align}\label{eq:marginal}
    \bm{f}_{\ss{G}} = \W_{\ss{G},\ss{A_G}} \bm{p}_{\ss{A_G}}.
\end{align}
When $\ss{G}$ is the group generated by the stabilizers and logical operators of an error-correcting code, the marginal probabilities give exactly the distribution of logical errors for each syndrome.
This marginal distribution can therefore be used to construct the maximum likelihood decoder for the actual noise afflicting a device.

In practice, inverting \cref{eq:marginal} to find the marginal error rates might be solved using, for example, a non-negative least squares solver, being careful to take the nontrivial covariance into account as well when estimating the errors.
However, to prove concrete theorems, we will use the estimator
\begin{align}\label{eq:IPEstimator}
    \hat{\bm{p}}_{\ss{A_G}}
    = \bigl[\W_{\ss{G}, \ss{A_G}}^{-1} \hat{\bm{f}}_{\ss{G}}\bigr]_\mc{D}
    = \bigl[|\ss{G}|^{-1} \W_{\ss{A_G}, \ss{G}} \hat{\bm{f}}_{\ss{G}}\bigr]_\mc{D}
\end{align}
using \cref{lem:WHtransform}.
Here we use the notation $[\bm{v}]_\mc{D}$ to mean taking a $k$-dimensional vector $\bm{v}$ and projecting it to the nearest point (according to a Euclidean metric) in the $k$-point simplex, which we denote $\mc{D}$, with the dimension being inferred from the context.
The informal statement in \cref{res:group} follows as $\|\bm{r}_\ss{G}\|_\infty \leq 2(1-p_0)$ by \cite[Lemma 4]{Erhard2019}.

\begin{proposition}
\label{prop:errorbound}
Let $\ss{G}$ be a group of Pauli matrices with a maximal stabilizer subgroup $\ss{S}$.
Then, under the conditions of \cref{prop:RatioProp}, an estimator $\hat{\bm{p}}_\ss{A_G}$ of $\bm{p}_\ss{A_G}$ satisfying
\begin{align*}
    \| \hat{\bm{p}}_\ss{A_G} - \bm{p}_\ss{A_G} \|_2 
    \leq O(\epsilon) \|\bm{r}_\ss{G}\|_\infty
\end{align*}
with probability at least $1-\delta$ can be obtained using $t= \frac{2}{\epsilon^2}\log\bigl(\frac{2 |\kappa| |\ss{G}|}{\delta}\bigr)$ measurements per round for $|\kappa|(\sqrt{|\ss{G}/\ss{S}|}+1)$ rounds,
where $|\kappa| = O\bigl(\log \tfrac{1}{\Delta_\ss{G}}\bigr)$ and $m_{\rm max} = O\bigl(\tfrac{1}{\Delta_\ss{G}}\bigr)$ is the longest sequence length used.
\end{proposition}

\begin{proof}
Suppose we have an estimator $\hat{\bm{r}}_\ss{G}$ satisfying 
\begin{align*}
    \| \hat{\bm{r}}_\ss{G} - \bm{r}_\ss{G} \|_\infty \leq \epsilon_0.
\end{align*}
First note that projecting to the simplex can only decrease the error, that is,
\begin{align*}
    \| \hat{\bm{p}}_\ss{A_G} - \bm{p}_{\ss{A_G}} \|_2
    \leq \| |\ss{G}|^{-1} \W_{\ss{A_G}, \ss{G}}( \hat{\bm{r}}_\ss{G} - \bm{r}_\ss{G}) \|_2.
\end{align*}
By \cref{lem:WHtransform} and using $\ss{C_{A_G}} \cap \ss{G} = 0$, $|\ss{G}|^{-1/2} \W_{\ss{A_G}, \ss{G}}$ is a unitary transformation on the space $\bm{v}_\ss{G}$, so, using a standard norm equivalence,
\begin{align}\label{eq:prbound}
    \| \hat{\bm{p}}_{\ss{A_G}} - \bm{p}_{\ss{A_G}} \|_2
    &\leq |\ss{G}|^{-1/2} \| \hat{\bm{r}}_\ss{G} - \bm{r}_\ss{G} \|_2 
    \leq \|\hat{\bm{r}}_\ss{G} - \bm{r}_\ss{G} \|_\infty \leq \epsilon_0.
\end{align}
Under the conditions of \cref{prop:RatioProp}, $\epsilon_0 \le O(\epsilon) \|\bm{r}_{\ss{G}}\|_\infty$.
The bound on the number of rounds and the total number of measurements follows from the bound on the size of a stabilizer covering of $\ss{G}$ from \cref{lem:covering}.
\end{proof}

In the special case that the group $\ss{G} = \ss{P}^n$,  \cref{prop:errorbound} provides a direct method for estimating the Pauli projection of a noise channel and gives pointwise precision to within $O(\epsilon)(1-p_0)$, where $p_0$ is the probability of no error.
Because of its frequent use in applications, it may be desirable to state the following corollary of \cref{prop:errorbound} in terms of the diamond distance.

\begin{corollary}
Let $\Lambda = \mathcal{E}^{\ss{P}^n}\bigr|_\ss{G}$ be the restriction to a subgroup $\ss{G}\subseteq\ss{P}^n$ of the Pauli projection of a channel $\mathcal{E}$.
For any sufficiently small $\epsilon,\delta\geq 0$, an estimate $\hat{\Lambda}$ of any $\tfrac{1}{2}$-weak, $\tfrac{1}{2}$-stable, GTM noise model $\Lambda$ can be reconstructed with
\begin{align*}
    \|\hat{\Lambda} - \Lambda\|_\diamond \leq O(\epsilon) \|\mathcal{I}_{\ss{G}} - \Lambda\|_F
\end{align*}
using $O\Bigl[\tfrac{|\kappa||\ss{G}|^{1/2}}{\epsilon^2} \log\bigl(\tfrac{|\kappa| |\ss{G}|}{\delta}\bigr)\Bigr]$ samples with $|\kappa| = O\bigl(\log(1/\Delta)\bigr)$ and sequence lengths at most $O(1/\Delta)$.
\end{corollary}

\begin{proof}
This follows from \cref{eq:prbound} in \cref{prop:errorbound} using the norm equivalence
\begin{align*}
    |\ss{G}|^{-1/2} \| \hat{\bm{p}}_\ss{A_G} - \bm{p}_{\ss{A_G}} \|_1 
    \le \| \hat{\bm{p}}_\ss{A_G} - \bm{p}_{\ss{A_G}} \|_2\,,
\end{align*}
with $\|\hat{\bm{r}}_\ss{G} - \bm{r}_\ss{G}\|_2 \le O(\epsilon)\| \bm{r}_\ss{G} \|_2 $ from \cref{prop:RatioProp} and $\| \bm{r}_\ss{G} \|_2  = \|\mathcal{I}_{\ss{G}} - \Lambda\|_F$, where $\mathcal{I}_{\ss{G}}$ is the restriction of the identity channel to $\ss{G}$.
\end{proof}

\section{Reconstructing a subset of errors}

We have shown that all Pauli error rates for a marginal model can be reliably reconstructed with substantially fewer resources than might have been anticipated.
However, reconstructing an exponential number of probabilities is manifestly inefficient.
We now show how to reconstruct the dominant error rates in an approximately sparse noise model.
The motivating application is that errors can be divided into a ``background'' noise process due to, e.g., independent dephasing and amplitude damping errors, and some additional errors that arise from unknown couplings.
The goal is to learn these couplings and either determine the physical mechanism causing them and eliminate them via re-engineering the device or to introduce compensating pulses.

To set the scale for the number of dominant error rates, consider a background local depolarizing process on $n$ qubits where independent single-qubit errors occur with probability $p$.
Then a weight $w$ error (that is, an error that acts nontrivially on $w$ qubits) occurs with probability $\binom{n}{w}(1-p)^{n-w}p^w$.
For $np\ll 1$ (where, e.g., 50 qubits with $p=0.001$ is approximately the current state of the art), the probability of an error with more than weight 1 is approximately $(np)^2/2$.
Now suppose an additional noise process is mixed in that, with probability $p'$ uniformly at random applies one of $t$ additional errors with no \textit{a priori} structure, representing, e.g., unknown correlated errors.
Then $3n + t$ errors account for almost all of the errors in the device.

We first show how to use the ${\bf Ratio}$ subroutine to estimate the probabilities of all errors within an arbitrary set $\ss{E}$.
Formally, we will prove pointwise precision to within $O(\epsilon)(1-p_0)$, where $p_0$ is the probability of no error.
By \cref{eq:WF} we have $1-p_0 = (1+1/d)r$ where $r$ is the so-called average error rate.
We will then present a divide-and-conquer routine to efficiently identify a set $\ss{E}$ of fixed size that can be used to identify correlated errors.
The informal version of the following theorem stated as \cref{res:sparse} in the introduction follows from using the trivial stabilizer covering $\ss{O}$ of $\ss{X}$ with $|\ss{O}|=|\ss{X}|$.

\begin{theorem}\label{thm:sparse}
For any set $\ss{E}\subseteq\ss{P}^n$,
an estimator $\hat{\bm{p}}_\ss{E}$ of $\bm{p}_\ss{E}$ satisfying
\begin{align*}
    \| \hat{\bm{p}}_\ss{E} - \bm{p}_\ss{E} \|_\infty 
    \leq O(\epsilon)(1-p_0)
\end{align*}
with probability at least $1-\delta$ can be obtained from the output of ${\bf Ratio}(\ss{O},\ss{X},t)$ using a random set $\ss{X}\subseteq\ss{P}^n$ with $|\ss{X}| = \tfrac{1}{\epsilon^2}\log(4|\ss{E}|/\delta)$, any stabilizer covering $\ss{O}$ of $\ss{X}$, and $t = \tfrac{1}{\epsilon^2}\log(4|\ss{X}||\kappa|/\delta)$ measurements per round for $|\ss{O}|$ rounds.
\end{theorem}

\begin{proof}
By \cref{eq:p-ptilde,lem:WHtransform,lem:pauliOrthogonality}, we have
\begin{align*}
    p_a = (a|\W^{-1} \bm{f} &= \frac{1}{|\ss{P}^n|} \W_{a,\ss{P}} \bm{f} \\
    &= \mathbb{E}_{b\in\ss{P}^n}\left[ (-1)^{\langle a, b\rangle}f_b\right] \\
    &=1[a=0] -\mathbb{E}_{b\in\ss{P}^n}\left[ (-1)^{\langle a, b\rangle}r_b\right]
\end{align*}
for any $a\in\ss{P}^n$.
Let $\ss{X}\subseteq\ss{P}^n$ be a set with $|\ss{X}|$ elements sampled independently and uniformly at random without replacement and
\begin{align*}
    \tilde{p}_a = 1[a=0] -\mathbb{E}_{b\in\ss{X}}\left[ (-1)^{\langle a, b\rangle}r_b\right].
\end{align*}
By \cite[Lemma 4]{Erhard2019}, we have $r_b\in[0,2(1-p_0)]$ for all $b\in\ss{P}^n$.
Applying Serfling's inequality~\cite{Serfling1974} to the first term gives
\begin{align}
\label{eq:ptildebound}
    \Pr\bigl(|\tilde{p}_a - p_a | \geq \epsilon (1-p_0)\bigr) \leq 2\mathrm{e}^{-|\ss{X}|\epsilon^2/2 (1-|\ss{X}|/|\ss{P}^n|)} \leq 2\mathrm{e}^{-|\ss{X}|\epsilon^2/2}=:\delta/2.
\end{align}

Now suppose we have estimates $\hat{r}_b$ of $r_b$ for all $b\in\ss{X}$ and let
\begin{align}
    \hat{p}_a = 1[a=0] -\mathbb{E}_{b\in\ss{X}} \left[ (-1)^{\langle a, b\rangle}\hat{r}_b\right].
\end{align}
By the triangle inequality and \cref{eq:ptildebound}, we have that with probability at least $1-\delta/2$,
\begin{align*}
    |\hat{p}_a - p_a | &\leq \epsilon(1-p_0) + \bigl|\mathbb{E}_{b\in\ss{X}}  (-1)^{\langle a, b\rangle}(\hat{r}_b - r_b)\bigr|\,.
\end{align*}
By \cref{prop:RatioProp}, $|\hat{r}_b-r_b| \le O(\epsilon)r_b$ with probability at least $1-\delta/2$ using $t=\frac{2}{\epsilon^2} \log\bigl(\frac{4 |\ss{X}||\kappa|}{\delta}\bigr)$ measurements per round.
Therefore by the union bound, we have that with probability at least $1-\delta$,
\begin{align*}
    |\hat{p}_a - p_a |\leq O(\epsilon) (1-p_0)
\end{align*}
for any $a\in\ss{P}^n$. 
The final result holds for any set $\ss{E}\subseteq\ss{P}^n$ by the union bound, redefining $\delta\to\delta/|\ss{E}|$.
\end{proof}

We now provide a search heuristic to identify sets $\ss{E}$ of interest using a function ${\bf Select}$ that returns the indices (as Pauli operators) of entries in a vector $\bm{p}$ of probabilities satisfying a desired condition.
We use three additional subroutines in what follows.
${\bf Choose}(\ss{A}, b)$ returns a set of $\min\{b,|\ss{A}|\}$ elements of $\ss{A}$ chosen uniformly at random without replacement.
${\bf Cover}(\ss{F})$ returns any valid stabilizer covering of a set $\ss{F}\subseteq\ss{P}^n$.
${\bf Ratio}$ is as described in \cref{sec:analysis}.
For any $\ss{A}\subset\mathbb{N}$, let $\ss{P}_\ss{A}\cong \ss{P}^{|\ss{A}|}$ be the set of Pauli operators that act trivially on all qubits not in $\ss{A}$ and the \textit{support} of a set $\mathrm{supp}(\ss{X})\subseteq\ss{P}^n$ be the set of qubits that are acted on nontrivially  by some element of $\ss{X}$.

\begin{alg}
[${\bf TreeReconstruction}(t,{\bf Select},u,n)$ Reconstruct an $n$-qubit error model using a function ${\bf Select}$ to select error terms at each stage of the iteration.]
\item Set $E'\coloneqq\{\ss{P}_{\{j\}}:j\in\mathbb{Z}_n\}$.
\item For $j=0,1,\ldots,\lceil\log_2 n\rceil$, Do
\begin{enumerate}
    \item Set $E\coloneqq E'$.
    \item Set $F\coloneqq \emptyset$.
    \item For each $e\in E$, Do
    \begin{itemize}
        \item Set $F \coloneqq F \cup {\bf Choose}(\ss{P}_{\mathrm{supp}(e)}, u)$.
    \end{itemize}
    \item Set $\ss{O} \coloneqq {\bf Cover}\bigl(\cup_{f\in F} f\bigr)$.
    \item Set $\bm{r} \coloneqq {\bf Ratio}(\ss{O}, \cup_{f\in F} f, t)$.
    \item Set $E'\coloneqq \emptyset$.
    \item For $k = 0, 1, \ldots, |E|-1$, Do
    \begin{itemize}
        \item Set $\displaystyle \bm{p} \coloneqq \sum_{e\in E_k} \mathbb{E}_{x\in F_k}(-1)^{\langle e, x\rangle}r_x|e)$.
        \item Set $E_k \coloneqq {\bf Select}(\bm{p})$.
        \item If $k$ is odd, Set $E'\coloneqq E' \cup \{E_{k-1} \otimes E_k\}$.
    \end{itemize}
    \item If $|E|$ is odd, set $E' \coloneqq E' \cup \{E_{|E|-1}\}$.
\end{enumerate}
\item Return $\displaystyle \bm{p} \coloneqq \sum_{e\in E_0} \mathbb{E}_{x\in F_0}(-1)^{\langle e, x\rangle}r_x|e)$.
\end{alg}

\begin{theorem}
Let ${\bf Select}$ be any function, $u$ be any positive integer and ${\bf Cover}(\ss{X})$ return the trivial cover from \cref{lem:covering}.
Then for sufficiently large $t$, ${\bf TreeReconstruction}(t,{\bf Select},u,n)$ uses $O\bigl(u t n\log(1/\Delta)\bigr)$ measurements.
\end{theorem}

\begin{proof}
At step 2(c) in the $j$th iteration, we are guaranteed that each of the at most $2^j$ subsets of $F$ each contain at most $u$ terms.
Therefore $|\cup_{f\in\ss{F}} f| = |\ss{O}| \leq 2^{\lceil \log_2(n) \rceil - j} u$ in each iteration.
For sufficiently large $t$, ${\bf Ratio}(\ss{O}, \cup_{f\in F} f, t)$ uses $t$ measurements for each of $|\ss{O}||\kappa|$ rounds where $|\kappa| = O\bigl(\log\tfrac{1}{\Delta}\bigr)$.
Therefore the $j$th iteration uses $O\bigl(2^{\lceil \log_2(n) \rceil - j} u t \log\tfrac{1}{\Delta}\bigr)$ measurements.
Summing over $j=0,\ldots, \lceil \log_2 n\rceil$ completes the proof.
\end{proof}

A natural choice of ${\bf Select}$ is the function that returns the indices of the $s$ largest elements of the input $\bm{p}$.
With this function, ${\bf TreeReconstruction}$ will heuristically return a set $\ss{E}$ that is a good $s$-sparse approximation of the Pauli channel.
We note that we do not attempt to prove that the resulting set is a best $s$-sparse approximation, that is, that it is a maximizer (or near-maximizer) of $\|\bm{p}_\ss{E}\|_1$ over all $\ss{E}$ of size $s$.
It would be possible to prove an upper bound on how close it gets to an optimal approximation by estimating the precision of each error in the set to high precision and backtracking if $\sum_{e\in E_k}\bm{p}_e$ is below some threshold value at any iteration.
However, a rigorous statement along these lines would likely only hold with impractical amounts of resources or a tighter analysis of the distribution of the output of ${\bf Ratio}$.
We leave the question of provable approximation ratios to future work. 

We also note that a better run-time may be achieved in this setting using methods for reconstructing a sparse signal from few measurements~\cite{Candes2006, Candes2006a, Donoho2006}.
Specifically, the present variant can be cast in a similar form to sparse Fourier and Hadamard transforms~\cite{Scheibler2015, Cheraghchi2017, Li2015,Lu2018}.
However, it is an open question if the methods in those references can be adapted to the symplectic structure of the Pauli group to achieve near-optimal sparse reconstruction.
Moreover, while finding such a set is desirable, we do not require any property of the set $\ss{E}$ for the reconstruction of $\bm{p}_\ss{E}$ to be pointwise convergent.
Thus, when the heuristic converges to a sparse set with large measure it is a certificate of correctness since each estimate in the set is pointwise accurate.

Another interesting choice of ${\bf Select}$ is a function that returns the indices of errors that do not fit a background model.
The most natural background model is for independent single-qubit errors, which can be directly estimated using \cref{thm:sparse}.

\section{Bounded degree graphical models}\label{sec:boundeddegree}

In this section we show how the reconstruction procedures above for complete Pauli channels can be used to learn a Pauli channel on $n$ qubits with bounded degree correlations.
It will be convenient to adopt a slightly different notation than what was used previously that treats the probability distribution $p$ over Pauli errors as a function of a random variable rather than a vector.
That is, we write $p(\bm{x})$ for the probability of the string of Paulis $\bm{x}$, $p(\bm{x}|\bm{y})$ for the conditional probability of $\bm{x}$ given $\bm{y}$, and so on.
In particular, if $p(\bm{x},\bm{y})$ is a joint distribution, then we write $p(\bm{x})$ for the marginal.
We will now review some concepts from the field of probabilistic graphical models; see Ref.~\cite{Koller2009} for an introduction.

Probabilistic graphical models are, at the most basic level, probability distributions over collections of random variables where the dependency structure between the variables is specified by a bipartite graph called a factor graph.
Consider a collection of random variables with $x_j$ denoting the $j$th random variable.
The variable $\bm{x}$ will be a random $n$-qubit Pauli, with $x_j$ being the $j$th tensor factor (i.e., the single-qubit Pauli acting on qubit $j$), but we leave the discussion more general for the moment.
Suppose that there exist \emph{factors} $C_k$ that are subsets of the random variables such that the joint probability distribution $p$ over all the $x_j$ obeys
\begin{align}
	p(\bm{x}) = \frac{1}{Z} \prod_{k} \phi_k(\bm{x}_{C_k}),
\end{align}
where $\phi_k : \bm{x}_{C_k} \to \mathbb{R}^+$ are strictly positive functions called \emph{factor potentials} supported on the factors; the argument of $\phi_k$ is shorthand notation for the subset of variables contained in a factor, $\bm{x}_{C} =\{ x_j : j \in C\}$; and $Z$ is a normalization constant called the \emph{partition function}.
(Note that some authors use the term factor potential to refer instead to $\log \phi_k$.)
Then the factor graph for this model is the bipartite graph with one set of nodes labeled by the variable labels $j$ and the other set of nodes labeled by the factors $C_k$, and an edge between $j$ and $C_k$ if and only if $x_j \in C_k$.
Such a factorized strictly positive probability distribution is known as a \emph{Gibbs random field} since the factor potentials play the role of exponentiated energy potentials in the Gibbs distribution of a statistical mechanical model.
Note that there always exists a Gibbs random field for any probability distribution, although the factor graph may be trivial, that is, contain only one factor.

\begin{figure}[tb]
\includegraphics{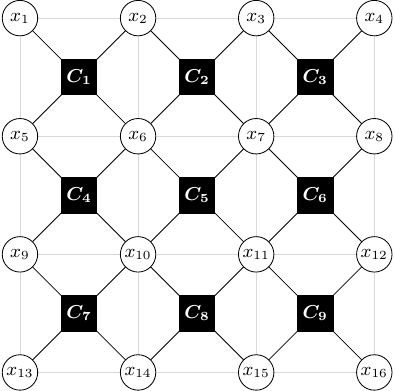}
\caption{\label{fig:factorgraph}
An example of a factor graph.
The factors are the black squares and the variable nodes are the white circles.
This is also an example of a bounded degree factor graph because each each variable belongs to at most 4 factors, and each factor couples only 4 variables.}
\end{figure}

The structure of a nontrivial factor graph implies that not all variables in $\bm{x}$ can be arbitrarily correlated.
In particular, the correlation is controlled by the \emph{Markov blanket} $\bm{x}_{\partial S}$ of a subset of variables $\bm{x}_S$, where the set $\partial S$ is defined as
\begin{align}
	\partial S = \bigcup_k\{C_k : C_k \cap S \not= \emptyset\} - S.
\end{align}
That is, it is the set of variables that are adjacent to $S$ via a factor node in the factor graph, but excluding $S$ itself.
We will refer to both $\partial S$ and $\bm{x}_{\partial S}$ as the Markov blankets of both $S$ and $\bm{x}_S$ with context resolving any ambiguity.
The closure $\bar{S}$ is the union of a set and its Markov blanket, $\bar{S} = S \cup \partial S$.

The celebrated Hammersley-Clifford theorem~\cite{Hammersley1971, Besag1974} describes how the Markov blanket controls the dependency between variables in a Gibbs random field.
The variables $\bm{x}_A$ and $\bm{x}_{A^c}$, where $A^c$ is the complement of $A$, obey the following conditional independence relation in terms of the Markov blanket of $A$,
\begin{align}
	p(\bm{x}_A | \bm{x}_{A^c}) = p(\bm{x}_A | \bm{x}_{\partial A}).
\end{align}
This relation, known as the \textit{local Markov property}, states that the variables in a set $A$ have only bounded correlation in the sense that they are conditionally independent of any other variables outside of their Markov blanket $\partial A$.

To illustrate some of these concepts, consider the example factor graph given in \cref{fig:factorgraph}.
Here the factors are the neighborhoods of the black squares, so for example $C_1 = \{x_1, x_2, x_5, x_6\}$.
The Markov blanket of $x_1$ is given by $\{x_2, x_5, x_6\}$, and the Markov blanket of $\{x_1, x_5\}$ is $\{x_2, x_6, x_9, x_{10}\}$.
The closure of $\{x_1, x_2\}$ is $\{x_1, x_2, x_3, x_5, x_6, x_7\}$.
The variable $x_1$ is conditionally independent of every variable except $\{x_2, x_5, x_6\}$ and the set of random variables $\{x_6, x_7, x_{10}, x_{11}\}$ depends conditionally on every other variable.

We will provide an estimate of a Gibbs random field by estimating the factor potentials.
We will not provide estimates for the partition function because estimating general partition functions is believed to be computationally hard~\cite{Barahona1982,Jerrum1993}, and so properly normalizing our probability estimate will in general be intractable.
However, there are still several cases of practical interest where this is provably not the case. 
For example, when the factor graph is a tree, and when $p(\bm{0})$ is not too small to be estimated via sampling.
Even without being able to efficiently compute the normalization exactly, either heuristics could be used, or ratios of probabilities can be used unconditionally.
Ratios of probabilities are all that is needed in, for example, most Monte Carlo methods.
However, as we now show, the error on the renormalized distribution can be bounded by the error on the factor potentials without needing to obtain the partition functions.

\begin{lemma}
\label{lemma:factorBound}
Let $p$ and $q$ be Gibbs random fields with the same factor graph and factorizations  $p(\bm{x}) = \tfrac{1}{Z_p}\prod_C \phi_{p,C}(\bm{x}_C)$ and $q(\bm{x}) = \tfrac{1}{Z_q}\prod_C \phi_{q,C}(\bm{x}_C)$ respectively.
Then
\begin{align*}
	\| p-q \|_1 \le \sum_C \max_{\bm{x}_C}\lvert\log \tfrac{\phi_{p,C}(\bm{x}_C)}{\phi_{q,C}(\bm{x}_C)}\rvert.
\end{align*}
\end{lemma}

\begin{proof}
We begin by recalling Pinsker's inequality~\cite{Pinsker1964}, one form of which is given by
\begin{align}
\label{eq:Pinsker}
	\|p-q\|_1^2 \le D(p\|q) + D(q\|p)
\end{align}
where $D(p\|q)$ is the relative entropy, defined by
\begin{align*}
	D(p\|q) = \sum_{\bm{x}} p(\bm{x}) \log \tfrac{p(\bm{x})}{q(\bm{x})}.
\end{align*}
Thus it suffices to bound the symmetric relative entropy, which simplifies to
\begin{align*}
    D(p\|q) + D(q\|p) = \sum_{\bm{x}} \left[p(\bm{x}) - q(\bm{x})\right] \log \tfrac{p(\bm{x})}{q(\bm{x})}
\end{align*}
for any probability distributions $p$ and $q$.
Substituting the definitions of the Gibbs random fields into the logarithms, we have
\begin{align*}
    D(p\|q) + D(q\|p) &= \sum_{\bm{x}} \left[p(\bm{x}) - q(\bm{x})\right] \left[\log\tfrac{Z_p}{Z_q} +\sum_C \log\tfrac{\phi_{p,C}}{\phi_{q,C}}\right] \\
    &= \sum_C\sum_{\bm{x}} \left[p(\bm{x}) - q(\bm{x})\right]  \log\tfrac{\phi_{p,C}}{\phi_{q,C}},
\end{align*}
where to obtain the second line we have used the fact that $Z_p$ and $Z_q$ are independent of $\bm{x}$ and that $\sum_{\bm{x}} p(x) - q(x) = 0$, which holds as $p$ and $q$ are probability distributions.
By applying Holder's inequality separately for each $C$, we obtain
\begin{align*}
    D(p\|q) + D(q\|p) 
    &\le \sum_C \|p - q \|_1 \max_{\bm{x}}|\log\tfrac{\phi_{p,C}}{\phi_{q,C}}| \\
    &\le \sum_C \|p - q \|_1 \max_{\bm{x}_C}|\log\tfrac{\phi_{p,C}}{\phi_{q,C}}|,
\end{align*}
where the second line follows from the bounded dependence of the factors.
Substituting the bound on the symmetric relative entropy into \cref{eq:Pinsker} and canceling the common factor of $\|p-q\|_1$ gives the desired result.
\end{proof}

Given a probability distribution $p(\bm{x})$ that is a Gibbs random field, the choice of factor potentials and partition function are not unique.
However, the Hammersley-Clifford theorem also gives an explicit description of a certain canonical choice for the factor potentials as a function of the marginal probability distributions on the factors and their Markov blankets, as well as an explicit partition function.
It is this description that we will use for our estimators below.
We use the parameterization from \citet{Abbeel2006} that applies directly to factor graphs (as opposed to less general formulations like Bayesian networks or Markov random fields) and makes explicit use of the local Markov property.
We first fix a fiducial assignment to $\bm{x}$ which we call $\bm{0}$.
This is simply a reference value and can be any arbitrary fixed assignment to $\bm{x}$; however we will see below that a convenient choice for our purposes will be to choose as a reference value the ``identity Pauli'' outcome for each variable.
Next we will augment the factors in our graph and consider $2^{C_k}$, the set of all subsets of $C_k$.
The union of all of these, minus the empty set, defines a new, larger collection of factors $\ss{C}^*$, given explicitly by
\begin{align}
	\ss{C}^* = \bigcup_k 2^{C_k} -\emptyset.
\end{align}
Next, given a factor $C^*_k \in \ss{C}^*$, define the \emph{canonical factor potential} for $C^*_k$ by
\begin{align}
\label{eq:logphi}
	\log \phi_k(\bm{x}_{C^*_k}) = \sum_{S \subseteq C^*_k} (-1)^{|C^*_k| - |S|} \log p(\bm{x}_{S}|\bm{0}_{\partial S}).
\end{align}
Let $C_j$ be a factor such that $C^*_k\in 2^{C_j}$.
Then every $S\subset C^*_k$ is also a subset of $C_j$ and so every probability on the right-hand-side of \cref{eq:logphi} depends only upon $\bar{C}_j$, the closure of $C_j$ with its Markov blanket.
That is, for all $S$ in the sum we have
\begin{align}
	S\subseteq C^*_k \cup \partial C^*_k \subseteq C_j \cup \partial C_j = \bar{C}_j.
\end{align}
Therefore each $\phi_k$ is determined by $p(\bm{x}_{\bar{C}_j})$ for some factor $C_j$.
If each $\bar{C}_j$ has constant size, each $\log\phi_k$ is completely determined by a small amount of data.

The Gibbs random field can now be expressed in terms of the canonical factor potentials as
\begin{align}
\label{eq:canonicalfactors}
	p(\bm{x}) = p(\bm{0}) \prod_{C^*_j \in \ss{C}^*} \phi_j(\bm{x}_{C^*_j}).
\end{align}
We note that $1/p(\bm{0})$ plays the role of the partition function $Z$ and can be independently estimated via \cref{thm:sparse}.


The error bounds below are derived assuming that an independent estimate is obtained for each of the raw factors in the canonical factor decomposition.
This adds a constant factor overhead to the overall sample complexity (when the degree of the factor graph is bounded).
This overhead could be avoided by using estimates from a covering set of marginals and then directly computing estimates of the marginals on the subfactors, that is, on the subsets of the factors. 
However, this might bias the overall estimate somewhat and would complicate the analysis presented below, so we leave an understanding of this more efficient protocol to future work.

Importantly, the estimates for the individual raw factors will not generally be self-consistent, and expressing this requires careful notation.
Let $\hat{p}_S(\bm{x}_S)$ denote an empirical estimate of the marginal distribution $p(\bm{x}_S)$ for a set $S$, which in our case will be obtained from \cref{prop:errorbound}.
Then for two sets $A,B$ such that $S\subset A$, $S\subset B$ and $A\neq B$, the two marginal distributions obtained over $S$ from the empirical estimates do not need to agree, that is,
\begin{align}
    \sum_{\bm{x}_{A-S}} \hat{p}_A(\bm{x}_{A-S},\bm{x}_S) \neq \sum_{\bm{x}_{B-S}} \hat{p}_B(\bm{x}_{B-S},\bm{x}_S)\,.
\end{align}
Indeed, they will generically not agree in the presence of sampling errors.
Moreover, neither will generally agree with the corresponding marginal of the global reconstruction $q$ obtained by substituting the $\hat{p}_S$ into \cref{eq:logphi}.
Thus, the canonical factor potentials are essential to round each of the local empirical marginals into a global coherent probability distribution with self-consistent local marginals. 

As a final difficulty, we cannot use an estimate of $p(\bm{0})$ as the partition function to exactly normalize the empirically reconstructed distribution. 
Fortunately, as shown in \cref{lemma:factorBound}, we do not need to know the exact value of the partition function for the empirical distribution to bound the error between our empirical estimate and the true distribution.

Our estimator of the global Gibbs distribution is as follows.
We will simply use the expressions in \cref{eq:logphi,eq:canonicalfactors} where empirical estimates of the local conditional probabilities are used in place of the exact distributions.
These can be obtained by using \cref{prop:errorbound} to estimate the complete marginal on the closure of each factor, $\bar{C}$.
When an estimate $q$ is obtained in this way, we say that $q$ is a \textit{canonical estimator} of $p$.

\begin{definition}[Canonical estimator]
Given a Gibbs random field $p$ with known factor graph $\ss{C}$, a distribution $q = q\bigl(\{\hat{p}_{\bar{C}}:C\in\ss{C}^\star\}\bigr)$ obtained by substituting empirical estimates $\{\hat{p}_{\bar{C}}:C\in\ss{C}^\star\}$ into \cref{eq:logphi,eq:canonicalfactors} is called a \textit{canonical estimator} of $p$. 
\end{definition}

We now bound the error on our empirical estimate assuming it has been properly normalized.
In order to obtain a nontrivial bound, we require that the empirical estimates of the marginal probabilities used to estimate the factor potentials are sufficiently close to their true values.
We also require that the true and estimated values are strictly positive, a point that we will quantify below.
Finally, we will also use the bound on the number of variables that are in any given factor, defining 
\begin{align}
\label{eq:NuDef}
    \nu \coloneqq \max_{C\in\ss{C}} |C|
\end{align}
to be the maximum degree of all the factors. 

\begin{lemma}
\label{lemma:first1normbound}
Let $p$ be a Gibbs random field with a factor graph $\ss{C}$ and factorization $\tfrac{1}{Z_p}\prod_{C\in\ss{C}} \phi_{p,C}$, let $N$ be the number of factors in $\ss{C}$, and let $\nu$ be the maximum degree of all the factors.
Then any canonical estimator $q$ satisfies
\begin{align}
	\| p-q \|_1
    &\le N 3^\nu \max_{\bm{x}_{C}} \bigl\lvert \log \tfrac{p(\bm{x}_{C}|\bm{0}_{\partial C})}{\hat{p}_{\bar{C}}(\bm{x}_{C}|\bm{0}_{\partial C})} \bigr\rvert,
\end{align}
where the maximum is over all variables $\bm{x}_C$ in either a factor or a subfactor in $\ss{C}$.
\end{lemma}

\begin{proof}
Applying \cref{lemma:factorBound} to the canonical factor decomposition gives
\begin{align*}
    \| p-q \|_1 &\le \sum_{C_k \in\ss{C}} \sum_{S\subseteq C_k} \max_{\bm{x}_S} \bigl\lvert \log \tfrac{\phi_{p,S}(\bm{x}_S)}{\phi_{q,S}(\bm{x}_S)} \bigr\rvert.
\end{align*}
Using the canonical factor potentials from \cref{eq:logphi}, we have
\begin{align*}
    (-1)^{|S|}\log \tfrac{\phi_{p,S}(\bm{x}_S)}{\phi_{q,S}(\bm{x}_S)}
    &= \sum_{R\subseteq S} (-1)^{|R|}\log\tfrac{p(\bm{x}_R|\bm{0}_{\partial R})}{\hat{p}_{\bar{R}}(\bm{x}_R|\bm{0}_{\partial R})}\,.
\end{align*}
Taking the absolute value and using the triangle inequality over the $2^{|S|}$ subsets of $S$, we obtain
\begin{align*}
    \| p-q \|_1 &\le \sum_{C_k \in\ss{C}} \sum_{S\subseteq C_k} \sum_{R\subseteq S} 
    \max_{\bm{x}_R}\bigl\lvert \log\tfrac{p(\bm{x}_R|\bm{0}_{\partial R})}{\hat{p}_{\bar{R}}(\bm{x}_R|\bm{0}_{\partial R})}\bigr\rvert \\
    &\le N 3^{\nu} \max_{\bm{x}_C}\bigl\lvert \log\tfrac{p(\bm{x}_C|\bm{0}_{\partial C})}{\hat{p}_{\bar{C}}(\bm{x}_C|\bm{0}_{\partial C})}\bigr\rvert\,,
\end{align*}
where in the last line we used the bound on the number of factors $N$, the bound on the degree of the factors $|C_k| \le \nu$, and the identity $\sum_{A\subseteq B\subseteq C} 1 = 3^{|C|}$.
The maximization in the last line is over all variables $\bm{x}_C$ where $C$ is a factor or a subfactor, and the result is immediate.
\end{proof}

The next lemmas let us translate the previous bound in terms of the logarithms of the conditional marginals and their empirical estimates into a bound in terms of the quantities that we naturally have control over, namely the (non-conditional) marginals and their empirical estimates.

\begin{lemma}
\label{lemma:basicloglemma}
For $a b > 0$, $\bigl|\log \tfrac{a}{b}\bigr| \le \tfrac{|a-b|}{\sqrt{a b}}$.
\end{lemma}
\begin{proof}
This result is a consequence of a simpler inequality in terms of a single variable.
For $x \ge 1$ we can use the integral representation of $\log x$ and the Cauchy-Schwarz inequality to obtain
\begin{align*}
	\log x = \int_1^x \frac{1}{z}\mathrm{d}z \le \left(\int_1^x \mathrm{d}z\right)^{1/2} \left(\int_1^x \frac{1}{z^2}\mathrm{d}z\right)^{1/2} = \sqrt{x} - \frac{1}{\sqrt{x}}  \qquad (x \ge 1).
\end{align*}
Multiplying by $-1$, we find
\begin{align*}
	-\log x = \log \tfrac{1}{x} \ge -\sqrt{x}+\frac{1}{\sqrt{x}},
\end{align*}
and by letting $0 < z = 1/x \le 1$, we see that the reverse inequality is true in the interval $x \in (0,1]$.
This establishes that for all $x > 0$,
\begin{align*}
	\lvert\log x\rvert \le \bigl\lvert \sqrt{x} - \tfrac{1}{\sqrt{x}}\bigr\rvert.
\end{align*}
The lemma follows by letting $x = a/b$ for $a b > 0$ and using some basic algebra.
\end{proof}

\begin{lemma}
\label{lemma:logboundlemma}
For any two strictly positive probability distributions $p(x,y)$, $q(x,y)$ on alphabets $X, Y$ and for any fixed elements $x\in X$ and $y\in Y$ we have
\begin{align*}
	\bigl\lvert\log \tfrac{p(x|y)}{q(x|y)} \bigr\rvert
	\le \frac{\lvert p(x,y) - q(x,y)\rvert + \lvert p(y) - q(y) \rvert}{p(y)\,\sqrt{p(x|y) q(x|y)}}.
\end{align*}
\end{lemma}

\begin{proof}
By \cref{lemma:basicloglemma}, we have
\begin{align*}
	\bigl\lvert \log \tfrac{p(x|y)}{q(x|y)} \bigr\rvert
	\le \frac{\lvert p(x|y) - q(x|y) \rvert}{\sqrt{p(x|y) q(x|y)}}.
\end{align*}
From the definition of conditional probability, the triangle inequality, and the fact that $p$ and $q$ are probability distributions, we then have
\begin{align*}
	\lvert p(x|y) - q(x|y) \rvert
	& = \frac{1}{p(y)}\bigl\lvert p(x,y) - q(x,y) + q(x|y)[q(y)-p(y)] \bigr\rvert \\
	& \le \frac{\lvert p(x,y) - q(x,y)\rvert + \lvert q(y) - p(y) \rvert}{p(y)}.
\end{align*}
Combining these two inequalities yields the desired bound.
We remark that this proof works equally well with $p \leftrightarrow q$, so the $p(y)$ in the denominator could also be replaced with $\max\{p(y),q(y)\}$ to get a stronger bound.
\end{proof}

The following definitions help facilitate a direct application of \cref{prop:RatioProp} to the problem of learning the parameters of a Pauli channel with a bounded-degree factor graph.
Here, as before, the minimizations in these definitions are over every factor or subfactor in the factor graph of $p$ and $q$.
We first introduce a quantitative notion of positivity given by the minimum geometric mean marginal probability over a factor $C$, defined by
\begin{align}
\label{eq:GeometricMeanDef}
	G(p,q) \coloneqq \min_{\bm{x}_C} \sqrt{p(\bm{x}_{C}|\bm{0}_{\partial C}) q(\bm{x}_{C}|\bm{0}_{\partial C})}
\end{align}
for the fixed but arbitrary outcome $\bm{0}_{\partial C}$.
We will also quantify the deviation of the local marginals from their empirical estimates by introducing the following error parameters.
Let
\begin{align}
\label{eq:EpsilonDefs}
    \epsilon_1 \coloneqq \max_{\bm{x}_{C}} \bigl\lvert p(\bm{x}_{C},\bm{0}_{\partial C}) - \hat{p}_{\bar{C}}(\bm{x}_{C},\bm{0}_{\partial C}) \bigr\rvert
    \quad \text{and} \quad
    \epsilon_2 \coloneqq \max_{\partial C} \bigl\lvert p(\bm{0}_{\partial C}) - \hat{p}_{\bar{C}}(\bm{0}_{\partial C}) \bigr\rvert\,.
\end{align}
For Pauli channels with maximum degree $\nu$, the triangle inequality shows that $\epsilon_2 \le 4^\nu \epsilon_1$.
Finally, we introduce
\begin{align}
\label{eq:GammaDef}
    \gamma \coloneqq \min_{\partial C} p(\bm{0}_{\partial C})\,.
\end{align}
From these definitions, we have the following lemma.

\begin{lemma}
\label{lemma:second1normbound}
Let $p$ be a Gibbs random field with a factor graph $\ss{C}$ and factorization $\tfrac{1}{Z_p}\prod_{C\in\ss{C}} \phi_{p,C}$, let $N$ be the number of factors in $\ss{C}$, and let $\nu$ be the maximum degree of all the factors.
Then any canonical estimator $q$ satisfies
\begin{align}
	\| p-q \|_1
    &\le \frac{N 3^\nu}{\gamma G} (\epsilon_1 + \epsilon_2),
\end{align}
where $G$ is as in \cref{eq:GeometricMeanDef}, $\epsilon_1$, $\epsilon_2$ are as in \cref{eq:EpsilonDefs}, and $\gamma$ is as in \cref{eq:GammaDef}. 
\end{lemma}

\begin{proof}
The result follows immediately by combining \cref{lemma:first1normbound} with \cref{lemma:logboundlemma} and the definitions in \cref{eq:GeometricMeanDef,eq:EpsilonDefs,eq:GammaDef}.
\end{proof}

\Cref{lemma:second1normbound} can be used together with our procedure for learning local marginal distributions (that is, complete Pauli noise models on a subset of qubits) to obtain a global guarantee on the 1-norm error between our estimate and the true distribution.
As we see in the lemma, the quality of the estimate will depend on two quantities related to the local marginals as well as the precision of our local estimate, $\epsilon$ from \cref{prop:errorbound}.
The first additional quantity is the local marginal that is furthest from the noiseless case, as quantified by the 2-norm,
\begin{align}
\label{eq:globalrbound}
	r^\star = \max_C \| \bm{1}_I - p(\bm{x}_{\bar{C}})\|_\infty.
\end{align}
Note that by \cref{eq:WF} we have $\| \bm{1}_I - p(\bm{x}_{\bar{C}})\|_\infty \le \| \bm{1}_I - p(\bm{x}_{\bar{C}})\|_1 = 2 \tfrac{d+1}{d} r \le 3r$ where $r$ is the average error rate of the local noise channel supported on $\bar{C}$, so we expect that this term is small.
The second is the minimum geometric mean marginal probability from \cref{eq:GeometricMeanDef}.
Finally, the geometry of the factor graph is also important.
We will need to use the bounded degree assumption to get a bound on $|\bar{C}|$ that is independent of $n$ and $N$.
We will quantify the geometry of the factor graph by defining, in analogy with the definition of $\nu$ from \cref{eq:NuDef},
\begin{align}
\label{eq:globalnubound}
	\bar{\nu} \coloneqq \max_{C} |\bar{C}|\,.
\end{align}
In general $\bar{\nu}$ can depend on both $n$ and $N$, but it is $O(1)$ for factor graphs having factors of bounded size and where each variable participates in a bounded number of factors.
From these quantities, we have the following global guarantee.

\begin{proposition}
\label{prop:boundeddegree}
Let $p$ be a Gibbs distribution with a known factor graph $\ss{C}$ having $N$ factors, $n$ variables, and $\bar{\nu} = O(1)$ such that $p$ corresponds to the Pauli error rates for a quantum channel.
Suppose that this noisy channel is GTM, $\tfrac{1}{2}$-weak, $\tfrac{1}{2}$-stable, and that $\Delta_{\bar{C}} = O(1)$.
Then there exists a canonical estimator $q$ such that for all sufficiently small $\epsilon > 0$, using $T$ samples with
\begin{align}
	T = O\biggl(\frac{N^2}{G^2 \epsilon^2} \log\biggl(\frac{N}{\delta}\biggr)\biggr)
\end{align}
we have, with probability at least $1-\delta$, 
\begin{align}
	\| p-q \|_1 \le \epsilon r^\star\,.
\end{align}
Moreover, an estimate proportional to $q$ can be found in time $\mathrm{poly}(n)$.
\end{proposition}

\begin{proof}
We construct the following canonical estimator $q$. 
Using \textbf{RunCB} and \textbf{Ratio}, we obtain estimates of the marginal distribution $\hat{p}_{\bar{C}}(\bm{x}_{\bar{C}})$ for every factor $C$ of the graph. 
By \cref{prop:errorbound}, each of these estimates can be obtained such that 
\begin{align*}
    \|\hat{p}_{\bar{C}}(\bm{x}_{\bar{C}}) - p_{\bar{C}}(\bm{x}_{\bar{C}})\|_2 \le O(1)\epsilon r^\star\,,
\end{align*}
with probability $\delta$ using $T_{\bar{C}} = O\bigl(\tfrac{1}{\epsilon^{2}} \log \tfrac{1}{\delta}\bigr)$ samples per marginal. 
Here the big-$O$ notation hides factors of $\kappa = O\bigl(\log \tfrac{1}{\Delta_{\bar{C}}}\bigr)$ where $\Delta_{\bar{C}}$ is the spectral gap of the marginal channel as defined in \cref{eq:channelspectralgap}, but these are $O(1)$ by assumption.
By the triangle inequality, these bounds apply to any marginals that are computed from these estimates as well, so they can also be applied to the subfactors in the canonical factor decomposition.

The above bounds for the empirical marginals imply that $\epsilon_1$ and $\epsilon_2$ from \cref{eq:EpsilonDefs} are also both bounded as $\epsilon_1, \epsilon_2 \le O(1) \epsilon r^\star$ with probability at least $1 - N \delta$.
Because the noise is assumed to be 1/2-weak, we also have that $\gamma$ from \cref{eq:GammaDef} is bounded as $\gamma \ge 1/2$. Now applying \cref{lemma:second1normbound}, we have a canonical estimator $q$ such that with probability at least $1-\delta$ 
\begin{align}
    \|p-q\|_1 \le O(1) \frac{N}{G} \epsilon r^\star
\end{align}
using $T = O\bigl(\tfrac{N}{\epsilon^{2}} \log \tfrac{N}{\delta}\bigr)$ samples. 

This sample complexity can be improved by a factor of $N$ as follows. 
Consider the graph whose vertices are the closures of each factor, $\bar{C}$, and whose edges are present if and only if $\bar{C}_1 \cap \bar{C}_2 \not= \emptyset$.
Because the factor graph is assumed to have constant degree, this adjacency graph also has constant degree, and so the $\bar{C}$ vertices can be partitioned into a constant number of independent sets.
By construction, nonadjacent factors are conditionally independent, and so they can be estimated from the same sample. 
That is to say, each independent set can be chosen to comprise $\Omega(N)$ vertices, and the constituent marginals associated to these factors can be estimated in parallel. 
Such a choice can be found in polynomial time in $N$ using, e.g., a greedy algorithm. 
The estimation is repeated across the $O(1)$ distinct independent sets, which reduces the sample complexity by a factor of $O(N)$, as claimed.

The final result is obtained by rescaling $\epsilon$ to $\epsilon/N$.
As all the factors are assumed to have size $O(1)$, the canonical estimator $q$ can be constructed in polynomial time in $n$ (or equivalently $N$) except for the normalizing factor.
\end{proof}

Finally, we demonstrate that significantly better bounds that scale as $\sqrt{n}$ can be obtained in terms of the differences between the marginal distributions of $p$ and $q$.
The following bound is difficult to apply theoretically because we cannot directly compute or bound the marginals of the empirical distribution without first computing or bounding the partition function for the reconstructed distribution.
However, the marginals of the true probability distribution can be accurately estimated by \cref{prop:RatioProp} and the marginals of the empirical distribution can be estimated heuristically by Monte Carlo sampling from the factor potentials.

\begin{lemma}
Given two Gibbs random fields $p$ and $q$ on $n$ variables with the same factor graph, we have the bound
\begin{align}
	\|p-q\|_1 \le \sqrt{n} \max_u \frac{1}{\sqrt{G^\star_u}} \bigl\|p(\bm{x}_{\bar{u}})-q(\bm{x}_{\bar{u}})\bigr\|_2,
\end{align}
where $G^\star_u = \min_{\bm{x}_{\bar{u}}} \sqrt{p(\bm{x}_{\bar{u}})q(\bm{x}_{\bar{u}})}$ and $\bar{u}$ denotes the closure $u\cup \partial u$ of $u$.
\end{lemma}

\begin{proof}
We will bound the relative entropy and then use Pinsker's inequality.
Recall the chain rule for relative entropy~\cite{Cover1991},
\begin{align*}
D\bigl(p(x,y|z)\|q(x,y|z)\bigr) = D\bigl(p(y|z)\|q(y|z)\bigr) + D\bigl(p(x|y,z)\|q(x|y,z)\bigr).
\end{align*}
By iteratively applying the chain rule to each variable, we have
\begin{align*}
	D\bigl(p\|q\bigr)
	& = \sum_{u=1}^n D\bigl(p(x_u|\bm{x}_{> u})\|q(x_u|\bm{x}_{> u})\bigr)\\
	& = \sum_{u=1}^n D\bigl(p(x_u|\bm{x}_{\partial^+ u})\|q(x_u|\bm{x}_{\partial^+ u})\bigr),
\end{align*}
where in the second step we use the local Markov property to restrict the conditioning to the remaining part of the Markov blanket of $u$, which we denote $\bm{x}_{\partial^+ u}$.
Because the relative entropy is positive, the chain rule also proves that marginalizing or conditioning on random variables only decreases the relative entropy.
Therefore, we can add back the variables from the Markov blanket of $u$ to get the upper bound
\begin{align*}
	D\bigl(p\|q\bigr)
	& \le \sum_{u=1}^n D\bigl(p(\bm{x}_{\bar{u}})\|q(\bm{x}_{\bar{u}})\bigr).
\end{align*}
Now we note that from \cref{lemma:basicloglemma} we have
\begin{align*}
	D\bigl(p\|q\bigr) + D\bigl(q\|p\bigr) = \sum_x \biggl(p(x)-q(x)\biggr) \log \tfrac{p(x)}{q(x)} \le 
	\sum_x \frac{\bigl(p(x)-q(x)\bigr)^2}{\sqrt{p(x)q(x)}} \le \frac{1}{\min_x \sqrt{p(x)q(x)}} \bigl\|p-q\bigr\|_2^2.
\end{align*}
The result then follows by applying this bound term by term, taking a maximum, and then substituting into \cref{eq:Pinsker}.
\end{proof}

\section{Conclusion}\label{sec:conclusion}

We have shown that Pauli channels can be learned with high precision using far fewer resources than previous methods.
As noise can be engineered to be effectively Pauli noise using randomized compiling~\cite{Wallman2016}, the methods described here should be broadly applicable to learning the residual noise under randomized compiling for large-scale devices, and, as such, represent a paradigm shift for characterizing quantum hardware.

We envision hardware developers using the description of the noise to efficiently discover previously unknown properties of devices beyond regimes that have been explored at present.
This should enable engineering effects to be concentrated on removing the most relevant residual errors.
Moreover, we envision theoretical efforts that will result in error mitigation and correction techniques that are tailored to the specific noise afflicting a device as reconstructed using the methods in this paper.
Such techniques will result in substantially better device performance by enabling optimized design of codes and decoders, bespoke fault tolerance schemes, and error-aware compiling for quantum simulation.

While we have proven rigorous upper bounds, many interesting open questions remain.
For example, we have not attempted to prove lower bounds on the problems that we study in this work.
It would be interesting to prove such bounds or to find procedures with even better asymptotic scaling.

Other natural avenues for future work are to explicitly prove robustness to gate-dependent noise and to generalize the method to estimate the noise on an interleaved gate~\cite{Magesan2012b}.
With such a generalization, a natural question would be to compare the predicted circuit performance under the reconstructed noise model with assessments of circuit-level performance obtained through other protocols, such as the accreditation protocol of Ref.~\cite{Ferracin2018}.

We have also proven our convergence guarantees in the idealized ``single-shot'' regime of Ref.~\cite{Granade2015}.
It should be straightforward to generalize our arguments to handle the reuse of individual sequences instead of using a fresh random sequence each time. 
We have also made exclusive use of the \textit{qubit} Pauli group, but it would be interesting to extend this work to $d$-level quantum systems as well and more formally incorporate techniques from compressed sensing.

For the reconstruction of Pauli channels as Markov random fields, it should also be possible to incorporate structure learning techniques~\cite{Chow1968, Abbeel2006, Bresler2013, Bresler2015, Hamilton2017} to learn the dependency structure.
If the correlations are promised to be bounded degree, then we expect that structure learning can be also be done efficiently, in time $\mathrm{poly}(n)$ in the number of qubits.
The sample complexity can be improved by making use of parallel samples on the same system so long as the measurements being used in parallel lie in disjoint neighborhoods of the factor graph.
This can be used to improve the efficiency still further.
Finally, a natural generalization of this idea is to learn a matrix product operator quantum channel~\cite{Verstraete2004, Zwolak2004}.
The state and unitary analogs of this idea have been explored before~\cite{Cramer2010a, Holzapfel2015}, but not yet in a way that is robust to SPAM errors.
Perhaps combining the ideas in the present paper with those of Kimmel \textit{et al}.~\cite{Kimmel2013} could lead to a more general procedure. 

Finally, there is the practical application of our methods to real world quantum devices. 
A first exploration along these lines can be found in a companion paper by R.~Harper and the present authors~\cite{Harper2019}.

\begin{acknowledgments}
We thank Robin Harper and Wenjun Yu for comments on an earlier draft.
This work was supported in part by the US Army Research Office grant numbers W911NF-14-1-0098 and W911NF-14-1-0103, the Australian Research Council Centre of Excellence for Engineered Quantum Systems (EQUS) CE170100009, the Government of Ontario, and the Government of Canada through the Canada First Research Excellence Fund (CFREF) and Transformative Quantum Technologies (TQT), the Natural Sciences and Engineering Research Council (NSERC), Industry Canada.
\end{acknowledgments}

\bibliographystyle{apsrev4-1}
\bibliography{refs}

\begin{thebibliography}{87}%
\makeatletter
\providecommand \@ifxundefined [1]{%
 \@ifx{#1\undefined}
}%
\providecommand \@ifnum [1]{%
 \ifnum #1\expandafter \@firstoftwo
 \else \expandafter \@secondoftwo
 \fi
}%
\providecommand \@ifx [1]{%
 \ifx #1\expandafter \@firstoftwo
 \else \expandafter \@secondoftwo
 \fi
}%
\providecommand \natexlab [1]{#1}%
\providecommand \enquote  [1]{``#1''}%
\providecommand \bibnamefont  [1]{#1}%
\providecommand \bibfnamefont [1]{#1}%
\providecommand \citenamefont [1]{#1}%
\providecommand \href@noop [0]{\@secondoftwo}%
\providecommand \href [0]{\begingroup \@sanitize@url \@href}%
\providecommand \@href[1]{\@@startlink{#1}\@@href}%
\providecommand \@@href[1]{\endgroup#1\@@endlink}%
\providecommand \@sanitize@url [0]{\catcode `\\12\catcode `\$12\catcode
  `\&12\catcode `\#12\catcode `\^12\catcode `\_12\catcode `\%12\relax}%
\providecommand \@@startlink[1]{}%
\providecommand \@@endlink[0]{}%
\providecommand \url  [0]{\begingroup\@sanitize@url \@url }%
\providecommand \@url [1]{\endgroup\@href {#1}{\urlprefix }}%
\providecommand \urlprefix  [0]{URL }%
\providecommand \Eprint [0]{\href }%
\providecommand \doibase [0]{http://dx.doi.org/}%
\providecommand \selectlanguage [0]{\@gobble}%
\providecommand \bibinfo  [0]{\@secondoftwo}%
\providecommand \bibfield  [0]{\@secondoftwo}%
\providecommand \translation [1]{[#1]}%
\providecommand \BibitemOpen [0]{}%
\providecommand \bibitemStop [0]{}%
\providecommand \bibitemNoStop [0]{.\EOS\space}%
\providecommand \EOS [0]{\spacefactor3000\relax}%
\providecommand \BibitemShut  [1]{\csname bibitem#1\endcsname}%
\let\auto@bib@innerbib\@empty
\bibitem [{\citenamefont {Terhal}(2015)}]{Terhal2015}%
  \BibitemOpen
  \bibfield  {author} {\bibinfo {author} {\bibfnamefont {B.~M.}\ \bibnamefont
  {Terhal}},\ }\href {\doibase 10.1103/RevModPhys.87.307} {\bibfield  {journal}
  {\bibinfo  {journal} {Rev. Mod. Phys.}\ }\textbf {\bibinfo {volume} {87}},\
  \bibinfo {pages} {307} (\bibinfo {year} {2015})},\ \Eprint
  {http://arxiv.org/abs/1302.3428} {arXiv:1302.3428} \BibitemShut {NoStop}%
\bibitem [{\citenamefont {Wallman}\ and\ \citenamefont
  {Emerson}(2016)}]{Wallman2016}%
  \BibitemOpen
  \bibfield  {author} {\bibinfo {author} {\bibfnamefont {J.~J.}\ \bibnamefont
  {Wallman}}\ and\ \bibinfo {author} {\bibfnamefont {J.}~\bibnamefont
  {Emerson}},\ }\href {\doibase 10.1103/PhysRevA.94.052325} {\bibfield
  {journal} {\bibinfo  {journal} {Phys. Rev. A}\ }\textbf {\bibinfo {volume}
  {94}},\ \bibinfo {pages} {052325} (\bibinfo {year} {2016})},\ \Eprint
  {http://arxiv.org/abs/1512.01098} {arXiv:1512.01098} \BibitemShut {NoStop}%
\bibitem [{\citenamefont {Sanders}\ \emph {et~al.}(2015)\citenamefont
  {Sanders}, \citenamefont {Wallman},\ and\ \citenamefont
  {Sanders}}]{Sanders2015}%
  \BibitemOpen
  \bibfield  {author} {\bibinfo {author} {\bibfnamefont {Y.~R.}\ \bibnamefont
  {Sanders}}, \bibinfo {author} {\bibfnamefont {J.~J.}\ \bibnamefont
  {Wallman}}, \ and\ \bibinfo {author} {\bibfnamefont {B.~C.}\ \bibnamefont
  {Sanders}},\ }\href {\doibase 10.1088/1367-2630/18/1/012002} {\bibfield
  {journal} {\bibinfo  {journal} {New J. Phys.}\ }\textbf {\bibinfo {volume}
  {18}},\ \bibinfo {pages} {012002} (\bibinfo {year} {2015})},\ \Eprint
  {http://arxiv.org/abs/1501.04932} {arXiv:1501.04932} \BibitemShut {NoStop}%
\bibitem [{\citenamefont {Kueng}\ \emph {et~al.}(2016)\citenamefont {Kueng},
  \citenamefont {Long}, \citenamefont {Doherty},\ and\ \citenamefont
  {Flammia}}]{Kueng2016}%
  \BibitemOpen
  \bibfield  {author} {\bibinfo {author} {\bibfnamefont {R.}~\bibnamefont
  {Kueng}}, \bibinfo {author} {\bibfnamefont {D.~M.}\ \bibnamefont {Long}},
  \bibinfo {author} {\bibfnamefont {A.~C.}\ \bibnamefont {Doherty}}, \ and\
  \bibinfo {author} {\bibfnamefont {S.~T.}\ \bibnamefont {Flammia}},\ }\href
  {\doibase 10.1103/PhysRevLett.117.170502} {\bibfield  {journal} {\bibinfo
  {journal} {Phys. Rev. Lett.}\ }\textbf {\bibinfo {volume} {117}},\ \bibinfo
  {pages} {170502} (\bibinfo {year} {2016})},\ \Eprint
  {http://arxiv.org/abs/1510.05653} {arXiv:1510.05653} \BibitemShut {NoStop}%
\bibitem [{\citenamefont {Huang}\ \emph {et~al.}(2019)\citenamefont {Huang},
  \citenamefont {Doherty},\ and\ \citenamefont {Flammia}}]{Huang2018}%
  \BibitemOpen
  \bibfield  {author} {\bibinfo {author} {\bibfnamefont {E.}~\bibnamefont
  {Huang}}, \bibinfo {author} {\bibfnamefont {A.~C.}\ \bibnamefont {Doherty}},
  \ and\ \bibinfo {author} {\bibfnamefont {S.}~\bibnamefont {Flammia}},\ }\href
  {\doibase 10.1103/PhysRevA.99.022313} {\bibfield  {journal} {\bibinfo
  {journal} {Phys. Rev. A}\ }\textbf {\bibinfo {volume} {99}},\ \bibinfo
  {pages} {022313} (\bibinfo {year} {2019})},\ \Eprint
  {http://arxiv.org/abs/1805.08227} {arXiv:1805.08227} \BibitemShut {NoStop}%
\bibitem [{\citenamefont {Beale}\ \emph {et~al.}(2018)\citenamefont {Beale},
  \citenamefont {Wallman}, \citenamefont {Guti\'{e}rrez}, \citenamefont
  {Brown},\ and\ \citenamefont {Laflamme}}]{Beale2018}%
  \BibitemOpen
  \bibfield  {author} {\bibinfo {author} {\bibfnamefont {S.}~\bibnamefont
  {Beale}}, \bibinfo {author} {\bibfnamefont {J.}~\bibnamefont {Wallman}},
  \bibinfo {author} {\bibfnamefont {M.}~\bibnamefont {Guti\'{e}rrez}}, \bibinfo
  {author} {\bibfnamefont {K.~R.}\ \bibnamefont {Brown}}, \ and\ \bibinfo
  {author} {\bibfnamefont {R.}~\bibnamefont {Laflamme}},\ }\href {\doibase
  10.1103/PhysRevLett.121.190501} {\bibfield  {journal} {\bibinfo  {journal}
  {Phys. Rev. Lett.}\ }\textbf {\bibinfo {volume} {121}},\ \bibinfo {pages}
  {190501} (\bibinfo {year} {2018})},\ \Eprint
  {http://arxiv.org/abs/1805.08802} {arXiv:1805.08802} \BibitemShut {NoStop}%
\bibitem [{\citenamefont {Ware}\ \emph {et~al.}(2018)\citenamefont {Ware},
  \citenamefont {Ribeill}, \citenamefont {Riste}, \citenamefont {Ryan},
  \citenamefont {Johnson},\ and\ \citenamefont {{da Silva}}}]{Ware2018}%
  \BibitemOpen
  \bibfield  {author} {\bibinfo {author} {\bibfnamefont {M.}~\bibnamefont
  {Ware}}, \bibinfo {author} {\bibfnamefont {G.}~\bibnamefont {Ribeill}},
  \bibinfo {author} {\bibfnamefont {D.}~\bibnamefont {Riste}}, \bibinfo
  {author} {\bibfnamefont {C.~A.}\ \bibnamefont {Ryan}}, \bibinfo {author}
  {\bibfnamefont {B.}~\bibnamefont {Johnson}}, \ and\ \bibinfo {author}
  {\bibfnamefont {M.~P.}\ \bibnamefont {{da Silva}}},\ }\href@noop {} {\
  (\bibinfo {year} {2018})},\ \Eprint {http://arxiv.org/abs/1803.01818}
  {1803.01818} \BibitemShut {NoStop}%
\bibitem [{\citenamefont {Bombin}\ \emph {et~al.}(2012)\citenamefont {Bombin},
  \citenamefont {Andrist}, \citenamefont {Ohzeki}, \citenamefont {Katzgraber},\
  and\ \citenamefont {Martin-Delgado}}]{Bombin2012}%
  \BibitemOpen
  \bibfield  {author} {\bibinfo {author} {\bibfnamefont {H.}~\bibnamefont
  {Bombin}}, \bibinfo {author} {\bibfnamefont {R.~S.}\ \bibnamefont {Andrist}},
  \bibinfo {author} {\bibfnamefont {M.}~\bibnamefont {Ohzeki}}, \bibinfo
  {author} {\bibfnamefont {H.~G.}\ \bibnamefont {Katzgraber}}, \ and\ \bibinfo
  {author} {\bibfnamefont {M.~A.}\ \bibnamefont {Martin-Delgado}},\ }\href
  {\doibase 10.1103/physrevx.2.021004} {\bibfield  {journal} {\bibinfo
  {journal} {Physical Review X}\ }\textbf {\bibinfo {volume} {2}},\ \bibinfo
  {pages} {021004} (\bibinfo {year} {2012})},\ \Eprint
  {http://arxiv.org/abs/1202.1852} {arXiv:1202.1852} \BibitemShut {NoStop}%
\bibitem [{\citenamefont {Nickerson}\ and\ \citenamefont
  {Brown}(2019)}]{Nickerson2017}%
  \BibitemOpen
  \bibfield  {author} {\bibinfo {author} {\bibfnamefont {N.~H.}\ \bibnamefont
  {Nickerson}}\ and\ \bibinfo {author} {\bibfnamefont {B.~J.}\ \bibnamefont
  {Brown}},\ }\href {\doibase 10.22331/q-2019-04-08-131} {\bibfield  {journal}
  {\bibinfo  {journal} {Quantum}\ }\textbf {\bibinfo {volume} {3}},\ \bibinfo
  {pages} {131} (\bibinfo {year} {2019})},\ \Eprint
  {http://arxiv.org/abs/1712.00502} {arXiv:1712.00502} \BibitemShut {NoStop}%
\bibitem [{\citenamefont {Darmawan}\ and\ \citenamefont
  {Poulin}(2017)}]{Darmawan2017}%
  \BibitemOpen
  \bibfield  {author} {\bibinfo {author} {\bibfnamefont {A.~S.}\ \bibnamefont
  {Darmawan}}\ and\ \bibinfo {author} {\bibfnamefont {D.}~\bibnamefont
  {Poulin}},\ }\href {\doibase 10.1103/PhysRevLett.119.040502} {\bibfield
  {journal} {\bibinfo  {journal} {Phys. Rev. Lett.}\ }\textbf {\bibinfo
  {volume} {119}},\ \bibinfo {pages} {040502} (\bibinfo {year} {2017})},\
  \Eprint {http://arxiv.org/abs/1607.06460} {arXiv:1607.06460} \BibitemShut
  {NoStop}%
\bibitem [{\citenamefont {Maskara}\ \emph {et~al.}(2019)\citenamefont
  {Maskara}, \citenamefont {Kubica},\ and\ \citenamefont
  {Jochym-O'Connor}}]{Maskara2018}%
  \BibitemOpen
  \bibfield  {author} {\bibinfo {author} {\bibfnamefont {N.}~\bibnamefont
  {Maskara}}, \bibinfo {author} {\bibfnamefont {A.}~\bibnamefont {Kubica}}, \
  and\ \bibinfo {author} {\bibfnamefont {T.}~\bibnamefont {Jochym-O'Connor}},\
  }\href {\doibase 10.1103/PhysRevA.99.052351} {\bibfield  {journal} {\bibinfo
  {journal} {Phys. Rev. A}\ }\textbf {\bibinfo {volume} {99}},\ \bibinfo
  {pages} {052351} (\bibinfo {year} {2019})},\ \Eprint
  {http://arxiv.org/abs/1802.08680} {arXiv:1802.08680} \BibitemShut {NoStop}%
\bibitem [{\citenamefont {Tuckett}\ \emph {et~al.}(2018)\citenamefont
  {Tuckett}, \citenamefont {Bartlett},\ and\ \citenamefont
  {Flammia}}]{Tuckett2018}%
  \BibitemOpen
  \bibfield  {author} {\bibinfo {author} {\bibfnamefont {D.~K.}\ \bibnamefont
  {Tuckett}}, \bibinfo {author} {\bibfnamefont {S.~D.}\ \bibnamefont
  {Bartlett}}, \ and\ \bibinfo {author} {\bibfnamefont {S.~T.}\ \bibnamefont
  {Flammia}},\ }\href {\doibase 10.1103/PhysRevLett.120.050505} {\bibfield
  {journal} {\bibinfo  {journal} {Phys. Rev. Lett.}\ }\textbf {\bibinfo
  {volume} {120}},\ \bibinfo {pages} {050505} (\bibinfo {year} {2018})},\
  \Eprint {http://arxiv.org/abs/1708.08474} {arXiv:1708.08474} \BibitemShut
  {NoStop}%
\bibitem [{\citenamefont {Tuckett}\ \emph {et~al.}(2019)\citenamefont
  {Tuckett}, \citenamefont {Darmawan}, \citenamefont {Chubb}, \citenamefont
  {Bravyi}, \citenamefont {Bartlett},\ and\ \citenamefont
  {Flammia}}]{Tuckett2019}%
  \BibitemOpen
  \bibfield  {author} {\bibinfo {author} {\bibfnamefont {D.~K.}\ \bibnamefont
  {Tuckett}}, \bibinfo {author} {\bibfnamefont {A.~S.}\ \bibnamefont
  {Darmawan}}, \bibinfo {author} {\bibfnamefont {C.~T.}\ \bibnamefont {Chubb}},
  \bibinfo {author} {\bibfnamefont {S.}~\bibnamefont {Bravyi}}, \bibinfo
  {author} {\bibfnamefont {S.~D.}\ \bibnamefont {Bartlett}}, \ and\ \bibinfo
  {author} {\bibfnamefont {S.~T.}\ \bibnamefont {Flammia}},\ }\href@noop {}
  {\bibfield  {journal} {\bibinfo  {journal} {Phys. Rev. X}\ }\textbf {\bibinfo
  {volume} {9}},\ \bibinfo {pages} {041031} (\bibinfo {year} {2019})},\ \Eprint
  {http://arxiv.org/abs/1812.08186} {arXiv:1812.08186} \BibitemShut {NoStop}%
\bibitem [{\citenamefont {Tuckett}\ \emph {et~al.}(2020)\citenamefont
  {Tuckett}, \citenamefont {Bartlett}, \citenamefont {Flammia},\ and\
  \citenamefont {Brown}}]{Tuckett2020}%
  \BibitemOpen
  \bibfield  {author} {\bibinfo {author} {\bibfnamefont {D.~K.}\ \bibnamefont
  {Tuckett}}, \bibinfo {author} {\bibfnamefont {S.~D.}\ \bibnamefont
  {Bartlett}}, \bibinfo {author} {\bibfnamefont {S.~T.}\ \bibnamefont
  {Flammia}}, \ and\ \bibinfo {author} {\bibfnamefont {B.~J.}\ \bibnamefont
  {Brown}},\ }\href {\doibase 10.1103/physrevlett.124.130501} {\bibfield
  {journal} {\bibinfo  {journal} {Physical Review Letters}\ }\textbf {\bibinfo
  {volume} {124}},\ \bibinfo {pages} {130501} (\bibinfo {year} {2020})},\
  \Eprint {http://arxiv.org/abs/1907.02554} {arXiv:1907.02554} \BibitemShut
  {NoStop}%
\bibitem [{\citenamefont {Robertson}\ \emph {et~al.}(2017)\citenamefont
  {Robertson}, \citenamefont {Granade}, \citenamefont {Bartlett},\ and\
  \citenamefont {Flammia}}]{Robertson2017}%
  \BibitemOpen
  \bibfield  {author} {\bibinfo {author} {\bibfnamefont {A.}~\bibnamefont
  {Robertson}}, \bibinfo {author} {\bibfnamefont {C.}~\bibnamefont {Granade}},
  \bibinfo {author} {\bibfnamefont {S.~D.}\ \bibnamefont {Bartlett}}, \ and\
  \bibinfo {author} {\bibfnamefont {S.~T.}\ \bibnamefont {Flammia}},\ }\href
  {\doibase 10.1103/PhysRevApplied.8.064004} {\bibfield  {journal} {\bibinfo
  {journal} {Phys. Rev. Applied}\ }\textbf {\bibinfo {volume} {8}},\ \bibinfo
  {pages} {064004} (\bibinfo {year} {2017})},\ \Eprint
  {http://arxiv.org/abs/1703.08179} {arXiv:1703.08179} \BibitemShut {NoStop}%
\bibitem [{\citenamefont {Aliferis}\ and\ \citenamefont
  {Preskill}(2008)}]{Aliferis2008}%
  \BibitemOpen
  \bibfield  {author} {\bibinfo {author} {\bibfnamefont {P.}~\bibnamefont
  {Aliferis}}\ and\ \bibinfo {author} {\bibfnamefont {J.}~\bibnamefont
  {Preskill}},\ }\href {\doibase 10.1103/PhysRevA.78.052331} {\bibfield
  {journal} {\bibinfo  {journal} {Phys. Rev. A}\ }\textbf {\bibinfo {volume}
  {78}},\ \bibinfo {pages} {052331} (\bibinfo {year} {2008})},\ \Eprint
  {http://arxiv.org/abs/0710.1301} {arXiv:0710.1301} \BibitemShut {NoStop}%
\bibitem [{\citenamefont {Chubb}\ and\ \citenamefont
  {Flammia}(2018)}]{Chubb2018}%
  \BibitemOpen
  \bibfield  {author} {\bibinfo {author} {\bibfnamefont {C.~T.}\ \bibnamefont
  {Chubb}}\ and\ \bibinfo {author} {\bibfnamefont {S.~T.}\ \bibnamefont
  {Flammia}},\ }\href@noop {} {\  (\bibinfo {year} {2018})},\ \Eprint
  {http://arxiv.org/abs/1809.10704} {arXiv:1809.10704} \BibitemShut {NoStop}%
\bibitem [{\citenamefont {Emerson}\ \emph {et~al.}(2005)\citenamefont
  {Emerson}, \citenamefont {Alicki},\ and\ \citenamefont
  {\.{Z}yczkowski}}]{Emerson2005}%
  \BibitemOpen
  \bibfield  {author} {\bibinfo {author} {\bibfnamefont {J.}~\bibnamefont
  {Emerson}}, \bibinfo {author} {\bibfnamefont {R.}~\bibnamefont {Alicki}}, \
  and\ \bibinfo {author} {\bibfnamefont {K.}~\bibnamefont {\.{Z}yczkowski}},\
  }\href {\doibase 10.1088/1464-4266/7/10/021} {\bibfield  {journal} {\bibinfo
  {journal} {J. Opt. B}\ }\textbf {\bibinfo {volume} {7}},\ \bibinfo {pages}
  {S347} (\bibinfo {year} {2005})},\ \Eprint
  {http://arxiv.org/abs/quant-ph/0503243} {quant-ph/0503243} \BibitemShut
  {NoStop}%
\bibitem [{\citenamefont {Dankert}\ \emph {et~al.}(2009)\citenamefont
  {Dankert}, \citenamefont {Cleve}, \citenamefont {Emerson},\ and\
  \citenamefont {Livine}}]{Dankert2009}%
  \BibitemOpen
  \bibfield  {author} {\bibinfo {author} {\bibfnamefont {C.}~\bibnamefont
  {Dankert}}, \bibinfo {author} {\bibfnamefont {R.}~\bibnamefont {Cleve}},
  \bibinfo {author} {\bibfnamefont {J.}~\bibnamefont {Emerson}}, \ and\
  \bibinfo {author} {\bibfnamefont {E.}~\bibnamefont {Livine}},\ }\href
  {\doibase 10.1103/PhysRevA.80.012304} {\bibfield  {journal} {\bibinfo
  {journal} {Phys. Rev. A}\ }\textbf {\bibinfo {volume} {80}},\ \bibinfo
  {pages} {012304} (\bibinfo {year} {2009})},\ \Eprint
  {http://arxiv.org/abs/quant-ph/0606161} {arXiv:quant-ph/0606161} \BibitemShut
  {NoStop}%
\bibitem [{\citenamefont {Knill}\ \emph {et~al.}(2008)\citenamefont {Knill},
  \citenamefont {Leibfried}, \citenamefont {Reichle}, \citenamefont {Britton},
  \citenamefont {Blakestad}, \citenamefont {Jost}, \citenamefont {Langer},
  \citenamefont {Ozeri}, \citenamefont {Seidelin},\ and\ \citenamefont
  {Wineland}}]{Knill2008}%
  \BibitemOpen
  \bibfield  {author} {\bibinfo {author} {\bibfnamefont {E.}~\bibnamefont
  {Knill}}, \bibinfo {author} {\bibfnamefont {D.}~\bibnamefont {Leibfried}},
  \bibinfo {author} {\bibfnamefont {R.}~\bibnamefont {Reichle}}, \bibinfo
  {author} {\bibfnamefont {J.}~\bibnamefont {Britton}}, \bibinfo {author}
  {\bibfnamefont {R.~B.}\ \bibnamefont {Blakestad}}, \bibinfo {author}
  {\bibfnamefont {J.~D.}\ \bibnamefont {Jost}}, \bibinfo {author}
  {\bibfnamefont {C.}~\bibnamefont {Langer}}, \bibinfo {author} {\bibfnamefont
  {R.}~\bibnamefont {Ozeri}}, \bibinfo {author} {\bibfnamefont
  {S.}~\bibnamefont {Seidelin}}, \ and\ \bibinfo {author} {\bibfnamefont
  {D.~J.}\ \bibnamefont {Wineland}},\ }\href {\doibase
  10.1103/PhysRevA.77.012307} {\bibfield  {journal} {\bibinfo  {journal} {Phys.
  Rev. A}\ }\textbf {\bibinfo {volume} {77}},\ \bibinfo {pages} {012307}
  (\bibinfo {year} {2008})},\ \Eprint {http://arxiv.org/abs/0707.0963}
  {arXiv:0707.0963} \BibitemShut {NoStop}%
\bibitem [{\citenamefont {Helsen}\ \emph {et~al.}(2019)\citenamefont {Helsen},
  \citenamefont {Xue}, \citenamefont {Vandersypen},\ and\ \citenamefont
  {Wehner}}]{Helsen2018}%
  \BibitemOpen
  \bibfield  {author} {\bibinfo {author} {\bibfnamefont {J.}~\bibnamefont
  {Helsen}}, \bibinfo {author} {\bibfnamefont {X.}~\bibnamefont {Xue}},
  \bibinfo {author} {\bibfnamefont {L.~M.}\ \bibnamefont {Vandersypen}}, \ and\
  \bibinfo {author} {\bibfnamefont {S.}~\bibnamefont {Wehner}},\ }\href
  {\doibase 10.1038/s41534-019-0182-7} {\bibfield  {journal} {\bibinfo
  {journal} {npj Quantum Inf.}\ }\textbf {\bibinfo {volume} {5}},\ \bibinfo
  {pages} {71} (\bibinfo {year} {2019})},\ \Eprint
  {http://arxiv.org/abs/1806.02048} {arXiv:1806.02048} \BibitemShut {NoStop}%
\bibitem [{\citenamefont {{Erhard}}\ \emph {et~al.}(2019)\citenamefont
  {{Erhard}}, \citenamefont {{Wallman}}, \citenamefont {{Postler}},
  \citenamefont {{Meth}}, \citenamefont {{Stricker}}, \citenamefont
  {{Martinez}}, \citenamefont {{Schindler}}, \citenamefont {{Monz}},
  \citenamefont {{Emerson}},\ and\ \citenamefont {{Blatt}}}]{Erhard2019}%
  \BibitemOpen
  \bibfield  {author} {\bibinfo {author} {\bibfnamefont {A.}~\bibnamefont
  {{Erhard}}}, \bibinfo {author} {\bibfnamefont {J.~J.}\ \bibnamefont
  {{Wallman}}}, \bibinfo {author} {\bibfnamefont {L.}~\bibnamefont
  {{Postler}}}, \bibinfo {author} {\bibfnamefont {M.}~\bibnamefont {{Meth}}},
  \bibinfo {author} {\bibfnamefont {R.}~\bibnamefont {{Stricker}}}, \bibinfo
  {author} {\bibfnamefont {E.~A.}\ \bibnamefont {{Martinez}}}, \bibinfo
  {author} {\bibfnamefont {P.}~\bibnamefont {{Schindler}}}, \bibinfo {author}
  {\bibfnamefont {T.}~\bibnamefont {{Monz}}}, \bibinfo {author} {\bibfnamefont
  {J.}~\bibnamefont {{Emerson}}}, \ and\ \bibinfo {author} {\bibfnamefont
  {R.}~\bibnamefont {{Blatt}}},\ }\href@noop {} {\bibfield  {journal} {\bibinfo
   {journal} {arXiv e-prints}\ ,\ \bibinfo {eid} {arXiv:1902.08543}} (\bibinfo
  {year} {2019})},\ \Eprint {http://arxiv.org/abs/1902.08543} {arXiv:1902.08543
  [quant-ph]} \BibitemShut {NoStop}%
\bibitem [{\citenamefont {Chasseur}\ and\ \citenamefont
  {Wilhelm}(2015)}]{Chasseur2015}%
  \BibitemOpen
  \bibfield  {author} {\bibinfo {author} {\bibfnamefont {T.}~\bibnamefont
  {Chasseur}}\ and\ \bibinfo {author} {\bibfnamefont {F.~K.}\ \bibnamefont
  {Wilhelm}},\ }\href {\doibase 10.1103/PhysRevA.92.042333} {\bibfield
  {journal} {\bibinfo  {journal} {Physical Review A}\ }\textbf {\bibinfo
  {volume} {92}},\ \bibinfo {pages} {042333} (\bibinfo {year} {2015})},\
  \Eprint {http://arxiv.org/abs/1505.00580v2} {arXiv:1505.00580v2} \BibitemShut
  {NoStop}%
\bibitem [{\citenamefont {Proctor}\ \emph {et~al.}(2017)\citenamefont
  {Proctor}, \citenamefont {Rudinger}, \citenamefont {Young}, \citenamefont
  {Sarovar},\ and\ \citenamefont {Blume-Kohout}}]{Proctor2017}%
  \BibitemOpen
  \bibfield  {author} {\bibinfo {author} {\bibfnamefont {T.}~\bibnamefont
  {Proctor}}, \bibinfo {author} {\bibfnamefont {K.}~\bibnamefont {Rudinger}},
  \bibinfo {author} {\bibfnamefont {K.}~\bibnamefont {Young}}, \bibinfo
  {author} {\bibfnamefont {M.}~\bibnamefont {Sarovar}}, \ and\ \bibinfo
  {author} {\bibfnamefont {R.}~\bibnamefont {Blume-Kohout}},\ }\href {\doibase
  10.1103/PhysRevLett.119.130502} {\bibfield  {journal} {\bibinfo  {journal}
  {Phys. Rev. Lett.}\ }\textbf {\bibinfo {volume} {119}},\ \bibinfo {pages}
  {130502} (\bibinfo {year} {2017})},\ \Eprint
  {http://arxiv.org/abs/1702.01853} {arXiv:1702.01853} \BibitemShut {NoStop}%
\bibitem [{\citenamefont {Wallman}(2018)}]{Wallman2018}%
  \BibitemOpen
  \bibfield  {author} {\bibinfo {author} {\bibfnamefont {J.}~\bibnamefont
  {Wallman}},\ }\href {\doibase 10.22331/q-2018-01-29-47} {\bibfield  {journal}
  {\bibinfo  {journal} {Quantum}\ }\textbf {\bibinfo {volume} {2}},\ \bibinfo
  {pages} {47} (\bibinfo {year} {2018})},\ \Eprint
  {http://arxiv.org/abs/1703.09835} {arXiv:1703.09835} \BibitemShut {NoStop}%
\bibitem [{\citenamefont {Merkel}\ \emph {et~al.}(2018)\citenamefont {Merkel},
  \citenamefont {Pritchett},\ and\ \citenamefont {Fong}}]{Merkel2018}%
  \BibitemOpen
  \bibfield  {author} {\bibinfo {author} {\bibfnamefont {S.~T.}\ \bibnamefont
  {Merkel}}, \bibinfo {author} {\bibfnamefont {E.~J.}\ \bibnamefont
  {Pritchett}}, \ and\ \bibinfo {author} {\bibfnamefont {B.~H.}\ \bibnamefont
  {Fong}},\ }\href@noop {} {\  (\bibinfo {year} {2018})},\ \Eprint
  {http://arxiv.org/abs/1804.05951} {arXiv:1804.05951} \BibitemShut {NoStop}%
\bibitem [{\citenamefont {Harty}\ \emph {et~al.}(2014)\citenamefont {Harty},
  \citenamefont {Allcock}, \citenamefont {Ballance}, \citenamefont {Guidoni},
  \citenamefont {Janacek}, \citenamefont {Linke}, \citenamefont {Stacey},\ and\
  \citenamefont {Lucas}}]{Harty2014}%
  \BibitemOpen
  \bibfield  {author} {\bibinfo {author} {\bibfnamefont {T.~P.}\ \bibnamefont
  {Harty}}, \bibinfo {author} {\bibfnamefont {D.~T.~C.}\ \bibnamefont
  {Allcock}}, \bibinfo {author} {\bibfnamefont {C.~J.}\ \bibnamefont
  {Ballance}}, \bibinfo {author} {\bibfnamefont {L.}~\bibnamefont {Guidoni}},
  \bibinfo {author} {\bibfnamefont {H.~A.}\ \bibnamefont {Janacek}}, \bibinfo
  {author} {\bibfnamefont {N.~M.}\ \bibnamefont {Linke}}, \bibinfo {author}
  {\bibfnamefont {D.~N.}\ \bibnamefont {Stacey}}, \ and\ \bibinfo {author}
  {\bibfnamefont {D.~M.}\ \bibnamefont {Lucas}},\ }\href {\doibase
  10.1103/PhysRevLett.113.220501} {\bibfield  {journal} {\bibinfo  {journal}
  {Phys. Rev. Lett.}\ }\textbf {\bibinfo {volume} {113}},\ \bibinfo {pages}
  {220501} (\bibinfo {year} {2014})},\ \Eprint {http://arxiv.org/abs/1403.1524}
  {arXiv:1403.1524} \BibitemShut {NoStop}%
\bibitem [{\citenamefont {Emerson}\ \emph {et~al.}(2007)\citenamefont
  {Emerson}, \citenamefont {Silva}, \citenamefont {Moussa}, \citenamefont
  {Ryan}, \citenamefont {Laforest}, \citenamefont {Baugh}, \citenamefont
  {Cory},\ and\ \citenamefont {Laflamme}}]{Emerson2007}%
  \BibitemOpen
  \bibfield  {author} {\bibinfo {author} {\bibfnamefont {J.}~\bibnamefont
  {Emerson}}, \bibinfo {author} {\bibfnamefont {M.}~\bibnamefont {Silva}},
  \bibinfo {author} {\bibfnamefont {O.}~\bibnamefont {Moussa}}, \bibinfo
  {author} {\bibfnamefont {C.}~\bibnamefont {Ryan}}, \bibinfo {author}
  {\bibfnamefont {M.}~\bibnamefont {Laforest}}, \bibinfo {author}
  {\bibfnamefont {J.}~\bibnamefont {Baugh}}, \bibinfo {author} {\bibfnamefont
  {D.~G.}\ \bibnamefont {Cory}}, \ and\ \bibinfo {author} {\bibfnamefont
  {R.}~\bibnamefont {Laflamme}},\ }\href {\doibase 10.1126/science.1145699}
  {\bibfield  {journal} {\bibinfo  {journal} {Science}\ }\textbf {\bibinfo
  {volume} {317}},\ \bibinfo {pages} {1893} (\bibinfo {year} {2007})},\ \Eprint
  {http://arxiv.org/abs/0707.0685} {arXiv:0707.0685} \BibitemShut {NoStop}%
\bibitem [{\citenamefont {Ryan}\ \emph {et~al.}(2009)\citenamefont {Ryan},
  \citenamefont {Laforest},\ and\ \citenamefont {Laflamme}}]{Ryan2009}%
  \BibitemOpen
  \bibfield  {author} {\bibinfo {author} {\bibfnamefont {C.~A.}\ \bibnamefont
  {Ryan}}, \bibinfo {author} {\bibfnamefont {M.}~\bibnamefont {Laforest}}, \
  and\ \bibinfo {author} {\bibfnamefont {R.}~\bibnamefont {Laflamme}},\ }\href
  {\doibase 10.1088/1367-2630/11/1/013034} {\bibfield  {journal} {\bibinfo
  {journal} {New J. Phys.}\ }\textbf {\bibinfo {volume} {11}},\ \bibinfo
  {pages} {013034} (\bibinfo {year} {2009})},\ \Eprint
  {http://arxiv.org/abs/0808.3973} {arXiv:0808.3973} \BibitemShut {NoStop}%
\bibitem [{\citenamefont {Gambetta}\ \emph {et~al.}(2012)\citenamefont
  {Gambetta}, \citenamefont {C\'{o}rcoles}, \citenamefont {Merkel},
  \citenamefont {Johnson}, \citenamefont {Smolin}, \citenamefont {Chow},
  \citenamefont {Ryan}, \citenamefont {Rigetti}, \citenamefont {Poletto},
  \citenamefont {Ohki}, \citenamefont {Ketchen},\ and\ \citenamefont
  {Steffen}}]{Gambetta2012}%
  \BibitemOpen
  \bibfield  {author} {\bibinfo {author} {\bibfnamefont {J.~M.}\ \bibnamefont
  {Gambetta}}, \bibinfo {author} {\bibfnamefont {A.~D.}\ \bibnamefont
  {C\'{o}rcoles}}, \bibinfo {author} {\bibfnamefont {S.~T.}\ \bibnamefont
  {Merkel}}, \bibinfo {author} {\bibfnamefont {B.~R.}\ \bibnamefont {Johnson}},
  \bibinfo {author} {\bibfnamefont {J.~A.}\ \bibnamefont {Smolin}}, \bibinfo
  {author} {\bibfnamefont {J.~M.}\ \bibnamefont {Chow}}, \bibinfo {author}
  {\bibfnamefont {C.~A.}\ \bibnamefont {Ryan}}, \bibinfo {author}
  {\bibfnamefont {C.}~\bibnamefont {Rigetti}}, \bibinfo {author} {\bibfnamefont
  {S.}~\bibnamefont {Poletto}}, \bibinfo {author} {\bibfnamefont {T.~A.}\
  \bibnamefont {Ohki}}, \bibinfo {author} {\bibfnamefont {M.~B.}\ \bibnamefont
  {Ketchen}}, \ and\ \bibinfo {author} {\bibfnamefont {M.}~\bibnamefont
  {Steffen}},\ }\href {\doibase 10.1103/PhysRevLett.109.240504} {\bibfield
  {journal} {\bibinfo  {journal} {Phys. Rev. Lett.}\ }\textbf {\bibinfo
  {volume} {109}},\ \bibinfo {pages} {240504} (\bibinfo {year} {2012})},\
  \Eprint {http://arxiv.org/abs/1204.6308} {arXiv:1204.6308} \BibitemShut
  {NoStop}%
\bibitem [{\citenamefont {Barends}\ \emph {et~al.}(2014)\citenamefont
  {Barends}, \citenamefont {Kelly}, \citenamefont {Veitia}, \citenamefont
  {Megrant}, \citenamefont {Fowler}, \citenamefont {Campbell}, \citenamefont
  {Chen}, \citenamefont {Chen}, \citenamefont {Chiaro}, \citenamefont
  {Dunsworth}, \citenamefont {Hoi}, \citenamefont {Jeffrey}, \citenamefont
  {Neill}, \citenamefont {O'Malley}, \citenamefont {Mutus}, \citenamefont
  {Quintana}, \citenamefont {Roushan}, \citenamefont {Sank}, \citenamefont
  {Wenner}, \citenamefont {White}, \citenamefont {Korotkov}, \citenamefont
  {Cleland},\ and\ \citenamefont {Martinis}}]{Barends2014a}%
  \BibitemOpen
  \bibfield  {author} {\bibinfo {author} {\bibfnamefont {R.}~\bibnamefont
  {Barends}}, \bibinfo {author} {\bibfnamefont {J.}~\bibnamefont {Kelly}},
  \bibinfo {author} {\bibfnamefont {A.}~\bibnamefont {Veitia}}, \bibinfo
  {author} {\bibfnamefont {A.}~\bibnamefont {Megrant}}, \bibinfo {author}
  {\bibfnamefont {A.~G.}\ \bibnamefont {Fowler}}, \bibinfo {author}
  {\bibfnamefont {B.}~\bibnamefont {Campbell}}, \bibinfo {author}
  {\bibfnamefont {Y.}~\bibnamefont {Chen}}, \bibinfo {author} {\bibfnamefont
  {Z.}~\bibnamefont {Chen}}, \bibinfo {author} {\bibfnamefont {B.}~\bibnamefont
  {Chiaro}}, \bibinfo {author} {\bibfnamefont {A.}~\bibnamefont {Dunsworth}},
  \bibinfo {author} {\bibfnamefont {I.-C.}\ \bibnamefont {Hoi}}, \bibinfo
  {author} {\bibfnamefont {E.}~\bibnamefont {Jeffrey}}, \bibinfo {author}
  {\bibfnamefont {C.}~\bibnamefont {Neill}}, \bibinfo {author} {\bibfnamefont
  {P.~J.~J.}\ \bibnamefont {O'Malley}}, \bibinfo {author} {\bibfnamefont
  {J.}~\bibnamefont {Mutus}}, \bibinfo {author} {\bibfnamefont
  {C.}~\bibnamefont {Quintana}}, \bibinfo {author} {\bibfnamefont
  {P.}~\bibnamefont {Roushan}}, \bibinfo {author} {\bibfnamefont
  {D.}~\bibnamefont {Sank}}, \bibinfo {author} {\bibfnamefont {J.}~\bibnamefont
  {Wenner}}, \bibinfo {author} {\bibfnamefont {T.~C.}\ \bibnamefont {White}},
  \bibinfo {author} {\bibfnamefont {A.~N.}\ \bibnamefont {Korotkov}}, \bibinfo
  {author} {\bibfnamefont {A.~N.}\ \bibnamefont {Cleland}}, \ and\ \bibinfo
  {author} {\bibfnamefont {J.~M.}\ \bibnamefont {Martinis}},\ }\href {\doibase
  10.1103/PhysRevA.90.030303} {\bibfield  {journal} {\bibinfo  {journal} {Phys.
  Rev. A}\ }\textbf {\bibinfo {volume} {90}},\ \bibinfo {pages} {030303}
  (\bibinfo {year} {2014})},\ \Eprint {http://arxiv.org/abs/1406.3364}
  {arXiv:1406.3364} \BibitemShut {NoStop}%
\bibitem [{\citenamefont {Carignan-Dugas}\ \emph {et~al.}(2015)\citenamefont
  {Carignan-Dugas}, \citenamefont {Wallman},\ and\ \citenamefont
  {Emerson}}]{Carignan-Dugas2015}%
  \BibitemOpen
  \bibfield  {author} {\bibinfo {author} {\bibfnamefont {A.}~\bibnamefont
  {Carignan-Dugas}}, \bibinfo {author} {\bibfnamefont {J.~J.}\ \bibnamefont
  {Wallman}}, \ and\ \bibinfo {author} {\bibfnamefont {J.}~\bibnamefont
  {Emerson}},\ }\href {\doibase 10.1103/PhysRevA.92.060302} {\bibfield
  {journal} {\bibinfo  {journal} {Phys. Rev. A}\ }\textbf {\bibinfo {volume}
  {92}},\ \bibinfo {pages} {060302} (\bibinfo {year} {2015})},\ \Eprint
  {http://arxiv.org/abs/1508.06312} {arXiv:1508.06312} \BibitemShut {NoStop}%
\bibitem [{\citenamefont {Cross}\ \emph {et~al.}(2016)\citenamefont {Cross},
  \citenamefont {Magesan}, \citenamefont {Bishop}, \citenamefont {Smolin},\
  and\ \citenamefont {Gambetta}}]{Cross2016}%
  \BibitemOpen
  \bibfield  {author} {\bibinfo {author} {\bibfnamefont {A.~W.}\ \bibnamefont
  {Cross}}, \bibinfo {author} {\bibfnamefont {E.}~\bibnamefont {Magesan}},
  \bibinfo {author} {\bibfnamefont {L.~S.}\ \bibnamefont {Bishop}}, \bibinfo
  {author} {\bibfnamefont {J.~A.}\ \bibnamefont {Smolin}}, \ and\ \bibinfo
  {author} {\bibfnamefont {J.~M.}\ \bibnamefont {Gambetta}},\ }\href {\doibase
  10.1038/npjqi.2016.12} {\bibfield  {journal} {\bibinfo  {journal} {npj
  Quantum Information}\ }\textbf {\bibinfo {volume} {2}} (\bibinfo {year}
  {2016}),\ 10.1038/npjqi.2016.12},\ \Eprint {http://arxiv.org/abs/1510.02720}
  {arXiv:1510.02720} \BibitemShut {NoStop}%
\bibitem [{\citenamefont {{Hashagen}}\ \emph {et~al.}(2018)\citenamefont
  {{Hashagen}}, \citenamefont {{Flammia}}, \citenamefont {{Gross}},\ and\
  \citenamefont {{Wallman}}}]{Hashagen2018}%
  \BibitemOpen
  \bibfield  {author} {\bibinfo {author} {\bibfnamefont {A.~K.}\ \bibnamefont
  {{Hashagen}}}, \bibinfo {author} {\bibfnamefont {S.~T.}\ \bibnamefont
  {{Flammia}}}, \bibinfo {author} {\bibfnamefont {D.}~\bibnamefont {{Gross}}},
  \ and\ \bibinfo {author} {\bibfnamefont {J.~J.}\ \bibnamefont {{Wallman}}},\
  }\href {\doibase 10.22331/q-2018-08-22-85} {\bibfield  {journal} {\bibinfo
  {journal} {Quantum}\ }\textbf {\bibinfo {volume} {2}},\ \bibinfo {pages} {85}
  (\bibinfo {year} {2018})},\ \Eprint {http://arxiv.org/abs/1801.06121}
  {arXiv:1801.06121} \BibitemShut {NoStop}%
\bibitem [{\citenamefont {Brown}\ and\ \citenamefont
  {Eastin}(2018)}]{Brown2018}%
  \BibitemOpen
  \bibfield  {author} {\bibinfo {author} {\bibfnamefont {W.~G.}\ \bibnamefont
  {Brown}}\ and\ \bibinfo {author} {\bibfnamefont {B.}~\bibnamefont {Eastin}},\
  }\href {\doibase 10.1103/PhysRevA.97.062323} {\bibfield  {journal} {\bibinfo
  {journal} {Phys. Rev. A}\ }\textbf {\bibinfo {volume} {97}},\ \bibinfo
  {pages} {062323} (\bibinfo {year} {2018})},\ \Eprint
  {http://arxiv.org/abs/1801.04042} {arXiv:1801.04042} \BibitemShut {NoStop}%
\bibitem [{\citenamefont {Fran{\c{c}}a}\ and\ \citenamefont
  {Hashagen}(2018)}]{Franca2018}%
  \BibitemOpen
  \bibfield  {author} {\bibinfo {author} {\bibfnamefont {D.~S.}\ \bibnamefont
  {Fran{\c{c}}a}}\ and\ \bibinfo {author} {\bibfnamefont {A.~K.}\ \bibnamefont
  {Hashagen}},\ }\href {\doibase 10.1088/1751-8121/aad6fa} {\bibfield
  {journal} {\bibinfo  {journal} {Journal of Physics A: Mathematical and
  Theoretical}\ }\textbf {\bibinfo {volume} {51}},\ \bibinfo {pages} {395302}
  (\bibinfo {year} {2018})},\ \Eprint {http://arxiv.org/abs/1803.03621}
  {arXiv:1803.03621} \BibitemShut {NoStop}%
\bibitem [{\citenamefont {Harper}\ \emph
  {et~al.}(2019{\natexlab{a}})\citenamefont {Harper}, \citenamefont {Hincks},
  \citenamefont {Ferrie}, \citenamefont {Flammia},\ and\ \citenamefont
  {Wallman}}]{Harper2018}%
  \BibitemOpen
  \bibfield  {author} {\bibinfo {author} {\bibfnamefont {R.}~\bibnamefont
  {Harper}}, \bibinfo {author} {\bibfnamefont {I.}~\bibnamefont {Hincks}},
  \bibinfo {author} {\bibfnamefont {C.}~\bibnamefont {Ferrie}}, \bibinfo
  {author} {\bibfnamefont {S.~T.}\ \bibnamefont {Flammia}}, \ and\ \bibinfo
  {author} {\bibfnamefont {J.~J.}\ \bibnamefont {Wallman}},\ }\href@noop {}
  {\bibfield  {journal} {\bibinfo  {journal} {Phys. Rev. A}\ } (\bibinfo {year}
  {2019}{\natexlab{a}})},\ \Eprint {http://arxiv.org/abs/1901.00535}
  {arXiv:1901.00535} \BibitemShut {NoStop}%
\bibitem [{\citenamefont {Fogarty}\ \emph {et~al.}(2015)\citenamefont
  {Fogarty}, \citenamefont {Veldhorst}, \citenamefont {Harper}, \citenamefont
  {Yang}, \citenamefont {Bartlett}, \citenamefont {Flammia},\ and\
  \citenamefont {Dzurak}}]{Fogarty2015}%
  \BibitemOpen
  \bibfield  {author} {\bibinfo {author} {\bibfnamefont {M.~A.}\ \bibnamefont
  {Fogarty}}, \bibinfo {author} {\bibfnamefont {M.}~\bibnamefont {Veldhorst}},
  \bibinfo {author} {\bibfnamefont {R.}~\bibnamefont {Harper}}, \bibinfo
  {author} {\bibfnamefont {C.~H.}\ \bibnamefont {Yang}}, \bibinfo {author}
  {\bibfnamefont {S.~D.}\ \bibnamefont {Bartlett}}, \bibinfo {author}
  {\bibfnamefont {S.~T.}\ \bibnamefont {Flammia}}, \ and\ \bibinfo {author}
  {\bibfnamefont {A.~S.}\ \bibnamefont {Dzurak}},\ }\href {\doibase
  10.1103/PhysRevA.92.022326} {\bibfield  {journal} {\bibinfo  {journal} {Phys.
  Rev. A}\ }\textbf {\bibinfo {volume} {92}},\ \bibinfo {pages} {022326}
  (\bibinfo {year} {2015})},\ \Eprint {http://arxiv.org/abs/1502.05119}
  {arXiv:1502.05119} \BibitemShut {NoStop}%
\bibitem [{\citenamefont {Hammersley}\ and\ \citenamefont
  {Clifford}(1971)}]{Hammersley1971}%
  \BibitemOpen
  \bibfield  {author} {\bibinfo {author} {\bibfnamefont {J.~M.}\ \bibnamefont
  {Hammersley}}\ and\ \bibinfo {author} {\bibfnamefont {P.}~\bibnamefont
  {Clifford}},\ }\href@noop {} {\enquote {\bibinfo {title} {{M}arkov fields on
  finite graphs and lattices},}\ } (\bibinfo {year} {1971}),\ \bibinfo {note}
  {available at
  \url{http://www.statslab.cam.ac.uk/~grg/books/hammfest/hamm-cliff.pdf}}\BibitemShut
  {NoStop}%
\bibitem [{\citenamefont {Harper}\ \emph
  {et~al.}(2019{\natexlab{b}})\citenamefont {Harper}, \citenamefont {Flammia},\
  and\ \citenamefont {Wallman}}]{Harper2019}%
  \BibitemOpen
  \bibfield  {author} {\bibinfo {author} {\bibfnamefont {R.}~\bibnamefont
  {Harper}}, \bibinfo {author} {\bibfnamefont {S.~T.}\ \bibnamefont {Flammia}},
  \ and\ \bibinfo {author} {\bibfnamefont {J.~J.}\ \bibnamefont {Wallman}},\
  }\href@noop {} {\  (\bibinfo {year} {2019}{\natexlab{b}})},\ \Eprint
  {http://arxiv.org/abs/1907.13022} {arXiv:1907.13022} \BibitemShut {NoStop}%
\bibitem [{\citenamefont {Rodionov}\ \emph {et~al.}(2014)\citenamefont
  {Rodionov}, \citenamefont {Veitia}, \citenamefont {Barends}, \citenamefont
  {Kelly}, \citenamefont {Sank}, \citenamefont {Wenner}, \citenamefont
  {Martinis}, \citenamefont {Kosut},\ and\ \citenamefont
  {Korotkov}}]{Rodionov2014}%
  \BibitemOpen
  \bibfield  {author} {\bibinfo {author} {\bibfnamefont {A.~V.}\ \bibnamefont
  {Rodionov}}, \bibinfo {author} {\bibfnamefont {A.}~\bibnamefont {Veitia}},
  \bibinfo {author} {\bibfnamefont {R.}~\bibnamefont {Barends}}, \bibinfo
  {author} {\bibfnamefont {J.}~\bibnamefont {Kelly}}, \bibinfo {author}
  {\bibfnamefont {D.}~\bibnamefont {Sank}}, \bibinfo {author} {\bibfnamefont
  {J.}~\bibnamefont {Wenner}}, \bibinfo {author} {\bibfnamefont {J.~M.}\
  \bibnamefont {Martinis}}, \bibinfo {author} {\bibfnamefont {R.~L.}\
  \bibnamefont {Kosut}}, \ and\ \bibinfo {author} {\bibfnamefont {A.~N.}\
  \bibnamefont {Korotkov}},\ }\href {\doibase 10.1103/PhysRevB.90.144504}
  {\bibfield  {journal} {\bibinfo  {journal} {Phys. Rev. B}\ }\textbf {\bibinfo
  {volume} {90}},\ \bibinfo {pages} {144504} (\bibinfo {year} {2014})},\
  \Eprint {http://arxiv.org/abs/1407.0761} {arXiv:1407.0761} \BibitemShut
  {NoStop}%
\bibitem [{\citenamefont {Gross}\ \emph {et~al.}(2010)\citenamefont {Gross},
  \citenamefont {Liu}, \citenamefont {Flammia}, \citenamefont {Becker},\ and\
  \citenamefont {Eisert}}]{Gross2010}%
  \BibitemOpen
  \bibfield  {author} {\bibinfo {author} {\bibfnamefont {D.}~\bibnamefont
  {Gross}}, \bibinfo {author} {\bibfnamefont {Y.-K.}\ \bibnamefont {Liu}},
  \bibinfo {author} {\bibfnamefont {S.~T.}\ \bibnamefont {Flammia}}, \bibinfo
  {author} {\bibfnamefont {S.}~\bibnamefont {Becker}}, \ and\ \bibinfo {author}
  {\bibfnamefont {J.}~\bibnamefont {Eisert}},\ }\href {\doibase
  10.1103/PhysRevLett.105.150401} {\bibfield  {journal} {\bibinfo  {journal}
  {Phys. Rev. Lett.}\ }\textbf {\bibinfo {volume} {105}},\ \bibinfo {pages}
  {150401} (\bibinfo {year} {2010})},\ \Eprint {http://arxiv.org/abs/0909.3304}
  {arXiv:0909.3304} \BibitemShut {NoStop}%
\bibitem [{\citenamefont {Flammia}\ \emph {et~al.}(2012)\citenamefont
  {Flammia}, \citenamefont {Gross}, \citenamefont {Liu},\ and\ \citenamefont
  {Eisert}}]{Flammia2012}%
  \BibitemOpen
  \bibfield  {author} {\bibinfo {author} {\bibfnamefont {S.~T.}\ \bibnamefont
  {Flammia}}, \bibinfo {author} {\bibfnamefont {D.}~\bibnamefont {Gross}},
  \bibinfo {author} {\bibfnamefont {Y.-K.}\ \bibnamefont {Liu}}, \ and\
  \bibinfo {author} {\bibfnamefont {J.}~\bibnamefont {Eisert}},\ }\href
  {\doibase 10.1088/1367-2630/14/9/095022} {\bibfield  {journal} {\bibinfo
  {journal} {New J. Phys.}\ }\textbf {\bibinfo {volume} {14}},\ \bibinfo
  {pages} {095022} (\bibinfo {year} {2012})},\ \Eprint
  {http://arxiv.org/abs/1205.2300} {arXiv:1205.2300} \BibitemShut {NoStop}%
\bibitem [{\citenamefont {Shabani}\ \emph {et~al.}(2011)\citenamefont
  {Shabani}, \citenamefont {Kosut}, \citenamefont {Mohseni}, \citenamefont
  {Rabitz}, \citenamefont {Broome}, \citenamefont {Almeida}, \citenamefont
  {Fedrizzi},\ and\ \citenamefont {White}}]{Shabani2011}%
  \BibitemOpen
  \bibfield  {author} {\bibinfo {author} {\bibfnamefont {A.}~\bibnamefont
  {Shabani}}, \bibinfo {author} {\bibfnamefont {R.~L.}\ \bibnamefont {Kosut}},
  \bibinfo {author} {\bibfnamefont {M.}~\bibnamefont {Mohseni}}, \bibinfo
  {author} {\bibfnamefont {H.}~\bibnamefont {Rabitz}}, \bibinfo {author}
  {\bibfnamefont {M.~A.}\ \bibnamefont {Broome}}, \bibinfo {author}
  {\bibfnamefont {M.~P.}\ \bibnamefont {Almeida}}, \bibinfo {author}
  {\bibfnamefont {A.}~\bibnamefont {Fedrizzi}}, \ and\ \bibinfo {author}
  {\bibfnamefont {A.~G.}\ \bibnamefont {White}},\ }\href {\doibase
  10.1103/PhysRevLett.106.100401} {\bibfield  {journal} {\bibinfo  {journal}
  {Phys. Rev. Lett.}\ }\textbf {\bibinfo {volume} {106}},\ \bibinfo {pages}
  {100401} (\bibinfo {year} {2011})},\ \Eprint {http://arxiv.org/abs/0910.5498}
  {arXiv:0910.5498} \BibitemShut {NoStop}%
\bibitem [{\citenamefont {Chuang}\ and\ \citenamefont
  {Nielsen}(1997)}]{Chuang1997}%
  \BibitemOpen
  \bibfield  {author} {\bibinfo {author} {\bibfnamefont {I.~L.}\ \bibnamefont
  {Chuang}}\ and\ \bibinfo {author} {\bibfnamefont {M.~A.}\ \bibnamefont
  {Nielsen}},\ }\href {\doibase 10.1080/09500349708231894} {\bibfield
  {journal} {\bibinfo  {journal} {J. Mod. Opt.}\ }\textbf {\bibinfo {volume}
  {44}},\ \bibinfo {pages} {2455} (\bibinfo {year} {1997})},\ \Eprint
  {http://arxiv.org/abs/quant-ph/9610001} {quant-ph/9610001} \BibitemShut
  {NoStop}%
\bibitem [{\citenamefont {Haah}\ \emph {et~al.}(2017)\citenamefont {Haah},
  \citenamefont {Harrow}, \citenamefont {Ji}, \citenamefont {Wu},\ and\
  \citenamefont {Yu}}]{Haah2017}%
  \BibitemOpen
  \bibfield  {author} {\bibinfo {author} {\bibfnamefont {J.}~\bibnamefont
  {Haah}}, \bibinfo {author} {\bibfnamefont {A.~W.}\ \bibnamefont {Harrow}},
  \bibinfo {author} {\bibfnamefont {Z.}~\bibnamefont {Ji}}, \bibinfo {author}
  {\bibfnamefont {X.}~\bibnamefont {Wu}}, \ and\ \bibinfo {author}
  {\bibfnamefont {N.}~\bibnamefont {Yu}},\ }\href {\doibase
  10.1109/tit.2017.2719044} {\bibfield  {journal} {\bibinfo  {journal} {{IEEE}
  Transactions on Information Theory}\ ,\ \bibinfo {pages} {1}} (\bibinfo
  {year} {2017})},\ \Eprint {http://arxiv.org/abs/1508.01797}
  {arXiv:1508.01797} \BibitemShut {NoStop}%
\bibitem [{\citenamefont {O'Donnell}\ and\ \citenamefont
  {Wright}(2016)}]{ODonnell2016}%
  \BibitemOpen
  \bibfield  {author} {\bibinfo {author} {\bibfnamefont {R.}~\bibnamefont
  {O'Donnell}}\ and\ \bibinfo {author} {\bibfnamefont {J.}~\bibnamefont
  {Wright}},\ }in\ \href {\doibase 10.1145/2897518.2897544} {\emph {\bibinfo
  {booktitle} {Proceedings of the Forty-eighth Annual ACM Symposium on Theory
  of Computing}}},\ \bibinfo {series and number} {STOC '16}\ (\bibinfo
  {publisher} {ACM},\ \bibinfo {address} {New York, NY, USA},\ \bibinfo {year}
  {2016})\ pp.\ \bibinfo {pages} {899--912}\BibitemShut {NoStop}%
\bibitem [{\citenamefont {O'Donnell}\ and\ \citenamefont
  {Wright}(2017)}]{ODonnell2017}%
  \BibitemOpen
  \bibfield  {author} {\bibinfo {author} {\bibfnamefont {R.}~\bibnamefont
  {O'Donnell}}\ and\ \bibinfo {author} {\bibfnamefont {J.}~\bibnamefont
  {Wright}},\ }in\ \href {\doibase 10.1145/3055399.3055454} {\emph {\bibinfo
  {booktitle} {Proceedings of the 49th Annual ACM SIGACT Symposium on Theory of
  Computing}}},\ \bibinfo {series and number} {STOC 2017}\ (\bibinfo
  {publisher} {ACM},\ \bibinfo {address} {New York, NY, USA},\ \bibinfo {year}
  {2017})\ pp.\ \bibinfo {pages} {962--974}\BibitemShut {NoStop}%
\bibitem [{\citenamefont {Aaronson}(2018)}]{Aaronson2017}%
  \BibitemOpen
  \bibfield  {author} {\bibinfo {author} {\bibfnamefont {S.}~\bibnamefont
  {Aaronson}},\ }in\ \href {\doibase 10.1145/3188745.3188802} {\emph {\bibinfo
  {booktitle} {Proceedings of the 50th Annual ACM SIGACT Symposium on Theory of
  Computing}}},\ \bibinfo {series and number} {STOC 2018}\ (\bibinfo
  {publisher} {ACM},\ \bibinfo {address} {New York, NY, USA},\ \bibinfo {year}
  {2018})\ pp.\ \bibinfo {pages} {325--338},\ \Eprint
  {http://arxiv.org/abs/1711.01053} {arXiv:1711.01053} \BibitemShut {NoStop}%
\bibitem [{\citenamefont {Cramer}\ \emph {et~al.}(2010)\citenamefont {Cramer},
  \citenamefont {Plenio}, \citenamefont {Flammia}, \citenamefont {Somma},
  \citenamefont {Gross}, \citenamefont {Bartlett}, \citenamefont
  {Landon-Cardinal}, \citenamefont {Poulin},\ and\ \citenamefont
  {Liu}}]{Cramer2010a}%
  \BibitemOpen
  \bibfield  {author} {\bibinfo {author} {\bibfnamefont {M.}~\bibnamefont
  {Cramer}}, \bibinfo {author} {\bibfnamefont {M.~B.}\ \bibnamefont {Plenio}},
  \bibinfo {author} {\bibfnamefont {S.~T.}\ \bibnamefont {Flammia}}, \bibinfo
  {author} {\bibfnamefont {R.}~\bibnamefont {Somma}}, \bibinfo {author}
  {\bibfnamefont {D.}~\bibnamefont {Gross}}, \bibinfo {author} {\bibfnamefont
  {S.~D.}\ \bibnamefont {Bartlett}}, \bibinfo {author} {\bibfnamefont
  {O.}~\bibnamefont {Landon-Cardinal}}, \bibinfo {author} {\bibfnamefont
  {D.}~\bibnamefont {Poulin}}, \ and\ \bibinfo {author} {\bibfnamefont {Y.-K.}\
  \bibnamefont {Liu}},\ }\href {\doibase 10.1038/ncomms1147} {\bibfield
  {journal} {\bibinfo  {journal} {Nat Commun}\ }\textbf {\bibinfo {volume}
  {1}},\ \bibinfo {pages} {149} (\bibinfo {year} {2010})},\ \Eprint
  {http://arxiv.org/abs/1101.4366} {arXiv:1101.4366} \BibitemShut {NoStop}%
\bibitem [{\citenamefont {Holz\"{a}pfel}\ \emph {et~al.}(2015)\citenamefont
  {Holz\"{a}pfel}, \citenamefont {Baumgratz}, \citenamefont {Cramer},\ and\
  \citenamefont {Plenio}}]{Holzapfel2015}%
  \BibitemOpen
  \bibfield  {author} {\bibinfo {author} {\bibfnamefont {M.}~\bibnamefont
  {Holz\"{a}pfel}}, \bibinfo {author} {\bibfnamefont {T.}~\bibnamefont
  {Baumgratz}}, \bibinfo {author} {\bibfnamefont {M.}~\bibnamefont {Cramer}}, \
  and\ \bibinfo {author} {\bibfnamefont {M.~B.}\ \bibnamefont {Plenio}},\
  }\href {\doibase 10.1103/PhysRevA.91.042129} {\bibfield  {journal} {\bibinfo
  {journal} {Phys. Rev. A}\ }\textbf {\bibinfo {volume} {91}},\ \bibinfo
  {pages} {042129} (\bibinfo {year} {2015})},\ \Eprint
  {http://arxiv.org/abs/1411.6379} {arXiv:1411.6379} \BibitemShut {NoStop}%
\bibitem [{\citenamefont {Merkel}\ \emph {et~al.}(2013)\citenamefont {Merkel},
  \citenamefont {Gambetta}, \citenamefont {Smolin}, \citenamefont {Poletto},
  \citenamefont {C{\'o}rcoles}, \citenamefont {Johnson}, \citenamefont {Ryan},\
  and\ \citenamefont {Steffen}}]{Merkel2012}%
  \BibitemOpen
  \bibfield  {author} {\bibinfo {author} {\bibfnamefont {S.~T.}\ \bibnamefont
  {Merkel}}, \bibinfo {author} {\bibfnamefont {J.~M.}\ \bibnamefont
  {Gambetta}}, \bibinfo {author} {\bibfnamefont {J.~A.}\ \bibnamefont
  {Smolin}}, \bibinfo {author} {\bibfnamefont {S.}~\bibnamefont {Poletto}},
  \bibinfo {author} {\bibfnamefont {A.~D.}\ \bibnamefont {C{\'o}rcoles}},
  \bibinfo {author} {\bibfnamefont {B.~R.}\ \bibnamefont {Johnson}}, \bibinfo
  {author} {\bibfnamefont {C.~A.}\ \bibnamefont {Ryan}}, \ and\ \bibinfo
  {author} {\bibfnamefont {M.}~\bibnamefont {Steffen}},\ }\href {\doibase
  10.1103/PhysRevA.87.062119} {\bibfield  {journal} {\bibinfo  {journal} {Phys.
  Rev. A}\ }\textbf {\bibinfo {volume} {87}},\ \bibinfo {pages} {062119}
  (\bibinfo {year} {2013})},\ \Eprint {http://arxiv.org/abs/1211.0322}
  {1211.0322} \BibitemShut {NoStop}%
\bibitem [{\citenamefont {Blume-Kohout}\ \emph {et~al.}(2016)\citenamefont
  {Blume-Kohout}, \citenamefont {Gamble}, \citenamefont {Nielsen},
  \citenamefont {Rudinger}, \citenamefont {Mizrahi}, \citenamefont {Fortier},\
  and\ \citenamefont {Maunz}}]{Blume-Kohout2016}%
  \BibitemOpen
  \bibfield  {author} {\bibinfo {author} {\bibfnamefont {R.}~\bibnamefont
  {Blume-Kohout}}, \bibinfo {author} {\bibfnamefont {J.~K.}\ \bibnamefont
  {Gamble}}, \bibinfo {author} {\bibfnamefont {E.}~\bibnamefont {Nielsen}},
  \bibinfo {author} {\bibfnamefont {K.}~\bibnamefont {Rudinger}}, \bibinfo
  {author} {\bibfnamefont {J.}~\bibnamefont {Mizrahi}}, \bibinfo {author}
  {\bibfnamefont {K.}~\bibnamefont {Fortier}}, \ and\ \bibinfo {author}
  {\bibfnamefont {P.}~\bibnamefont {Maunz}},\ }\href {\doibase
  10.1038/ncomms14485} {\bibfield  {journal} {\bibinfo  {journal} {Nature
  Communications}\ }\textbf {\bibinfo {volume} {8}},\ \bibinfo {pages} {{}}
  (\bibinfo {year} {2016})},\ \Eprint {http://arxiv.org/abs/1605.07674}
  {arXiv:1605.07674} \BibitemShut {NoStop}%
\bibitem [{\citenamefont {Kimmel}\ \emph {et~al.}(2014)\citenamefont {Kimmel},
  \citenamefont {da~Silva}, \citenamefont {Ryan}, \citenamefont {Johnson},\
  and\ \citenamefont {Ohki}}]{Kimmel2013}%
  \BibitemOpen
  \bibfield  {author} {\bibinfo {author} {\bibfnamefont {S.}~\bibnamefont
  {Kimmel}}, \bibinfo {author} {\bibfnamefont {M.~P.}\ \bibnamefont
  {da~Silva}}, \bibinfo {author} {\bibfnamefont {C.~A.}\ \bibnamefont {Ryan}},
  \bibinfo {author} {\bibfnamefont {B.~R.}\ \bibnamefont {Johnson}}, \ and\
  \bibinfo {author} {\bibfnamefont {T.}~\bibnamefont {Ohki}},\ }\href {\doibase
  10.1103/PhysRevX.4.011050} {\bibfield  {journal} {\bibinfo  {journal} {Phys.
  Rev. X}\ }\textbf {\bibinfo {volume} {4}},\ \bibinfo {pages} {011050}
  (\bibinfo {year} {2014})},\ \Eprint {http://arxiv.org/abs/1306.2348}
  {arXiv:1306.2348} \BibitemShut {NoStop}%
\bibitem [{\citenamefont {Roth}\ \emph {et~al.}(2018)\citenamefont {Roth},
  \citenamefont {Kueng}, \citenamefont {Kimmel}, \citenamefont {Liu},
  \citenamefont {Gross}, \citenamefont {Eisert},\ and\ \citenamefont
  {Kliesch}}]{Roth2018}%
  \BibitemOpen
  \bibfield  {author} {\bibinfo {author} {\bibfnamefont {I.}~\bibnamefont
  {Roth}}, \bibinfo {author} {\bibfnamefont {R.}~\bibnamefont {Kueng}},
  \bibinfo {author} {\bibfnamefont {S.}~\bibnamefont {Kimmel}}, \bibinfo
  {author} {\bibfnamefont {Y.-K.}\ \bibnamefont {Liu}}, \bibinfo {author}
  {\bibfnamefont {D.}~\bibnamefont {Gross}}, \bibinfo {author} {\bibfnamefont
  {J.}~\bibnamefont {Eisert}}, \ and\ \bibinfo {author} {\bibfnamefont
  {M.}~\bibnamefont {Kliesch}},\ }\href {\doibase
  10.1103/PhysRevLett.121.170502} {\bibfield  {journal} {\bibinfo  {journal}
  {Phys. Rev. Lett.}\ }\textbf {\bibinfo {volume} {121}},\ \bibinfo {pages}
  {170502} (\bibinfo {year} {2018})},\ \Eprint
  {http://arxiv.org/abs/1803.00572} {arXiv:1803.00572} \BibitemShut {NoStop}%
\bibitem [{\citenamefont {Bresler}\ \emph {et~al.}(2014)\citenamefont
  {Bresler}, \citenamefont {Gamarnik},\ and\ \citenamefont
  {Shah}}]{Bresler2014}%
  \BibitemOpen
  \bibfield  {author} {\bibinfo {author} {\bibfnamefont {G.}~\bibnamefont
  {Bresler}}, \bibinfo {author} {\bibfnamefont {D.}~\bibnamefont {Gamarnik}}, \
  and\ \bibinfo {author} {\bibfnamefont {D.}~\bibnamefont {Shah}},\ }in\
  \href@noop {} {\emph {\bibinfo {booktitle} {Proceedings of the 27th
  International Conference on Neural Information Processing Systems - Volume
  1}}},\ \bibinfo {series and number} {NIPS'14}\ (\bibinfo  {publisher} {MIT
  Press},\ \bibinfo {address} {Cambridge, MA, USA},\ \bibinfo {year} {2014})\
  pp.\ \bibinfo {pages} {1062--1070},\ \Eprint {http://arxiv.org/abs/1409.3836}
  {arXiv:1409.3836} \BibitemShut {NoStop}%
\bibitem [{\citenamefont {Chow}\ and\ \citenamefont {Liu}(1968)}]{Chow1968}%
  \BibitemOpen
  \bibfield  {author} {\bibinfo {author} {\bibfnamefont {C.}~\bibnamefont
  {Chow}}\ and\ \bibinfo {author} {\bibfnamefont {C.}~\bibnamefont {Liu}},\
  }\href {\doibase 10.1109/tit.1968.1054142} {\bibfield  {journal} {\bibinfo
  {journal} {{IEEE} Transactions on Information Theory}\ }\textbf {\bibinfo
  {volume} {14}},\ \bibinfo {pages} {462} (\bibinfo {year} {1968})}\BibitemShut
  {NoStop}%
\bibitem [{\citenamefont {Abbeel}\ \emph {et~al.}(2006)\citenamefont {Abbeel},
  \citenamefont {Koller},\ and\ \citenamefont {Ng}}]{Abbeel2006}%
  \BibitemOpen
  \bibfield  {author} {\bibinfo {author} {\bibfnamefont {P.}~\bibnamefont
  {Abbeel}}, \bibinfo {author} {\bibfnamefont {D.}~\bibnamefont {Koller}}, \
  and\ \bibinfo {author} {\bibfnamefont {A.~Y.}\ \bibnamefont {Ng}},\
  }\href@noop {} {\bibfield  {journal} {\bibinfo  {journal} {Journal of Machine
  Learning Research}\ }\textbf {\bibinfo {volume} {7}},\ \bibinfo {pages}
  {1743} (\bibinfo {year} {2006})},\ \Eprint {http://arxiv.org/abs/1207.1366}
  {arXiv:1207.1366} \BibitemShut {NoStop}%
\bibitem [{\citenamefont {Bresler}\ \emph {et~al.}(2013)\citenamefont
  {Bresler}, \citenamefont {Mossel},\ and\ \citenamefont {Sly}}]{Bresler2013}%
  \BibitemOpen
  \bibfield  {author} {\bibinfo {author} {\bibfnamefont {G.}~\bibnamefont
  {Bresler}}, \bibinfo {author} {\bibfnamefont {E.}~\bibnamefont {Mossel}}, \
  and\ \bibinfo {author} {\bibfnamefont {A.}~\bibnamefont {Sly}},\ }\href
  {\doibase 10.1137/100796029} {\bibfield  {journal} {\bibinfo  {journal}
  {{SIAM} Journal on Computing}\ }\textbf {\bibinfo {volume} {42}},\ \bibinfo
  {pages} {563} (\bibinfo {year} {2013})},\ \Eprint
  {http://arxiv.org/abs/0712.1402} {arXiv:0712.1402} \BibitemShut {NoStop}%
\bibitem [{\citenamefont {Bresler}(2015)}]{Bresler2015}%
  \BibitemOpen
  \bibfield  {author} {\bibinfo {author} {\bibfnamefont {G.}~\bibnamefont
  {Bresler}},\ }in\ \href {\doibase 10.1145/2746539.2746631} {\emph {\bibinfo
  {booktitle} {Proceedings of the Forty-seventh Annual ACM Symposium on Theory
  of Computing}}},\ \bibinfo {series and number} {STOC '15}\ (\bibinfo
  {publisher} {ACM},\ \bibinfo {address} {New York, NY, USA},\ \bibinfo {year}
  {2015})\ pp.\ \bibinfo {pages} {771--782},\ \Eprint
  {http://arxiv.org/abs/1411.6156} {arXiv:1411.6156} \BibitemShut {NoStop}%
\bibitem [{\citenamefont {Hamilton}\ \emph {et~al.}(2017)\citenamefont
  {Hamilton}, \citenamefont {Koehler},\ and\ \citenamefont
  {Moitra}}]{Hamilton2017}%
  \BibitemOpen
  \bibfield  {author} {\bibinfo {author} {\bibfnamefont {L.}~\bibnamefont
  {Hamilton}}, \bibinfo {author} {\bibfnamefont {F.}~\bibnamefont {Koehler}}, \
  and\ \bibinfo {author} {\bibfnamefont {A.}~\bibnamefont {Moitra}},\ }in\
  \href@noop {} {\emph {\bibinfo {booktitle} {Proceedings of the 31st
  International Conference on Neural Information Processing Systems}}},\
  \bibinfo {series and number} {NIPS'17}\ (\bibinfo  {publisher} {Curran
  Associates Inc.},\ \bibinfo {address} {USA},\ \bibinfo {year} {2017})\ pp.\
  \bibinfo {pages} {2460--2469},\ \Eprint {http://arxiv.org/abs/1705.11107}
  {arXiv:1705.11107} \BibitemShut {NoStop}%
\bibitem [{\citenamefont {Klivans}\ and\ \citenamefont
  {Meka}(2017)}]{Klivans2017}%
  \BibitemOpen
  \bibfield  {author} {\bibinfo {author} {\bibfnamefont {A.}~\bibnamefont
  {Klivans}}\ and\ \bibinfo {author} {\bibfnamefont {R.}~\bibnamefont {Meka}},\
  }in\ \href {\doibase 10.1109/FOCS.2017.39} {\emph {\bibinfo {booktitle} {2017
  IEEE 58th Annual Symposium on Foundations of Computer Science (FOCS)}}}\
  (\bibinfo {year} {2017})\ pp.\ \bibinfo {pages} {343--354},\ \Eprint
  {http://arxiv.org/abs/1706.06274} {arXiv:1706.06274} \BibitemShut {NoStop}%
\bibitem [{\citenamefont {Gottesman}(1997)}]{Gottesman1997}%
  \BibitemOpen
  \bibfield  {author} {\bibinfo {author} {\bibfnamefont {D.}~\bibnamefont
  {Gottesman}},\ }\emph {\bibinfo {title} {Stabilizer codes and quantum error
  correction}},\ \href {https://thesis.library.caltech.edu/2900/2/THESIS.pdf}
  {Ph.D. thesis},\ \bibinfo  {school} {Caltech}, \bibinfo {address} {Pasadena,
  California} (\bibinfo {year} {1997}),\ \Eprint
  {http://arxiv.org/abs/quant-ph/9705052} {quant-ph/9705052} \BibitemShut
  {NoStop}%
\bibitem [{\citenamefont {Wootters}\ and\ \citenamefont
  {Fields}(1989)}]{Wootters1989}%
  \BibitemOpen
  \bibfield  {author} {\bibinfo {author} {\bibfnamefont {W.~K.}\ \bibnamefont
  {Wootters}}\ and\ \bibinfo {author} {\bibfnamefont {B.~D.}\ \bibnamefont
  {Fields}},\ }\href {\doibase 10.1016/0003-4916(89)90322-9} {\bibfield
  {journal} {\bibinfo  {journal} {Annals of Physics}\ }\textbf {\bibinfo
  {volume} {191}},\ \bibinfo {pages} {363} (\bibinfo {year}
  {1989})}\BibitemShut {NoStop}%
\bibitem [{\citenamefont {Dankert}(2005)}]{DankertThesis}%
  \BibitemOpen
  \bibfield  {author} {\bibinfo {author} {\bibfnamefont {C.}~\bibnamefont
  {Dankert}},\ }\emph {\bibinfo {title} {Efficient Simulation of Random Quantum
  States and Operators}},\ \href@noop {} {Ph.D. thesis},\ \bibinfo  {school}
  {University of Waterloo} (\bibinfo {year} {2005}),\ \Eprint
  {http://arxiv.org/abs/quant-ph/0512217} {quant-ph/0512217} \BibitemShut
  {NoStop}%
\bibitem [{\citenamefont {Magesan}(2012)}]{Magesan2012a}%
  \BibitemOpen
  \bibfield  {author} {\bibinfo {author} {\bibfnamefont {E.}~\bibnamefont
  {Magesan}},\ }\emph {\bibinfo {title} {Characterizing Noise in Quantum
  Systems}},\ \href
  {https://uwspace.uwaterloo.ca/bitstream/handle/10012/6832/Magesan_Easwar.pdf}
  {Ph.D. thesis},\ \bibinfo  {school} {University of Waterloo}, \bibinfo
  {address} {Waterloo, Ontario, Canada} (\bibinfo {year} {2012})\BibitemShut
  {NoStop}%
\bibitem [{\citenamefont {Proctor}\ \emph {et~al.}(2019)\citenamefont
  {Proctor}, \citenamefont {Carignan-Dugas}, \citenamefont {Rudinger},
  \citenamefont {Nielsen}, \citenamefont {Blume-Kohout},\ and\ \citenamefont
  {Young}}]{Proctor2018}%
  \BibitemOpen
  \bibfield  {author} {\bibinfo {author} {\bibfnamefont {T.~J.}\ \bibnamefont
  {Proctor}}, \bibinfo {author} {\bibfnamefont {A.}~\bibnamefont
  {Carignan-Dugas}}, \bibinfo {author} {\bibfnamefont {K.}~\bibnamefont
  {Rudinger}}, \bibinfo {author} {\bibfnamefont {E.}~\bibnamefont {Nielsen}},
  \bibinfo {author} {\bibfnamefont {R.}~\bibnamefont {Blume-Kohout}}, \ and\
  \bibinfo {author} {\bibfnamefont {K.}~\bibnamefont {Young}},\ }\href
  {\doibase 10.1103/PhysRevLett.123.030503} {\bibfield  {journal} {\bibinfo
  {journal} {Phys. Rev. Lett.}\ }\textbf {\bibinfo {volume} {123}},\ \bibinfo
  {pages} {030503} (\bibinfo {year} {2019})},\ \Eprint
  {http://arxiv.org/abs/1807.07975} {arXiv:1807.07975} \BibitemShut {NoStop}%
\bibitem [{\citenamefont {Hoeffding}(1963)}]{Hoeffding1963}%
  \BibitemOpen
  \bibfield  {author} {\bibinfo {author} {\bibfnamefont {W.}~\bibnamefont
  {Hoeffding}},\ }\href {\doibase 10.1080/01621459.1963.10500830} {\bibfield
  {journal} {\bibinfo  {journal} {Journal of the American Statistical
  Association}\ }\textbf {\bibinfo {volume} {58}},\ \bibinfo {pages} {13}
  (\bibinfo {year} {1963})}\BibitemShut {NoStop}%
\bibitem [{\citenamefont {Serfling}(1974)}]{Serfling1974}%
  \BibitemOpen
  \bibfield  {author} {\bibinfo {author} {\bibfnamefont {R.~J.}\ \bibnamefont
  {Serfling}},\ }\href {\doibase 10.1214/aos/1176342611} {\bibfield  {journal}
  {\bibinfo  {journal} {The Annals of Statistics}\ }\textbf {\bibinfo {volume}
  {2}},\ \bibinfo {pages} {39} (\bibinfo {year} {1974})}\BibitemShut {NoStop}%
\bibitem [{\citenamefont {Cand\`{e}s}\ and\ \citenamefont
  {Tao}(2006)}]{Candes2006}%
  \BibitemOpen
  \bibfield  {author} {\bibinfo {author} {\bibfnamefont {E.~J.}\ \bibnamefont
  {Cand\`{e}s}}\ and\ \bibinfo {author} {\bibfnamefont {T.}~\bibnamefont
  {Tao}},\ }\href@noop {} {\bibfield  {journal} {\bibinfo  {journal} {IEEE
  Trans. Info. Theory}\ }\textbf {\bibinfo {volume} {52}},\ \bibinfo {pages}
  {5406} (\bibinfo {year} {2006})}\BibitemShut {NoStop}%
\bibitem [{\citenamefont {Cand\`{e}s}\ \emph {et~al.}(2006)\citenamefont
  {Cand\`{e}s}, \citenamefont {Romberg},\ and\ \citenamefont
  {Tao}}]{Candes2006a}%
  \BibitemOpen
  \bibfield  {author} {\bibinfo {author} {\bibfnamefont {E.}~\bibnamefont
  {Cand\`{e}s}}, \bibinfo {author} {\bibfnamefont {J.}~\bibnamefont {Romberg}},
  \ and\ \bibinfo {author} {\bibfnamefont {T.}~\bibnamefont {Tao}},\
  }\href@noop {} {\bibfield  {journal} {\bibinfo  {journal} {IEEE Trans.
  Inform. Theory}\ }\textbf {\bibinfo {volume} {52}},\ \bibinfo {pages} {489}
  (\bibinfo {year} {2006})}\BibitemShut {NoStop}%
\bibitem [{\citenamefont {Donoho}(2006)}]{Donoho2006}%
  \BibitemOpen
  \bibfield  {author} {\bibinfo {author} {\bibfnamefont {D.}~\bibnamefont
  {Donoho}},\ }\href@noop {} {\bibfield  {journal} {\bibinfo  {journal} {IEEE
  Trans. Info. Theory}\ }\textbf {\bibinfo {volume} {52}},\ \bibinfo {pages}
  {1289} (\bibinfo {year} {2006})}\BibitemShut {NoStop}%
\bibitem [{\citenamefont {Scheibler}\ \emph {et~al.}(2015)\citenamefont
  {Scheibler}, \citenamefont {Haghighatshoar},\ and\ \citenamefont
  {Vetterli}}]{Scheibler2015}%
  \BibitemOpen
  \bibfield  {author} {\bibinfo {author} {\bibfnamefont {R.}~\bibnamefont
  {Scheibler}}, \bibinfo {author} {\bibfnamefont {S.}~\bibnamefont
  {Haghighatshoar}}, \ and\ \bibinfo {author} {\bibfnamefont {M.}~\bibnamefont
  {Vetterli}},\ }\href {\doibase 10.1109/tit.2015.2404441} {\bibfield
  {journal} {\bibinfo  {journal} {{IEEE} Transactions on Information Theory}\
  }\textbf {\bibinfo {volume} {61}},\ \bibinfo {pages} {2115} (\bibinfo {year}
  {2015})},\ \Eprint {http://arxiv.org/abs/1310.1803} {arXiv:1310.1803}
  \BibitemShut {NoStop}%
\bibitem [{\citenamefont {Cheraghchi}\ and\ \citenamefont
  {Indyk}(2017)}]{Cheraghchi2017}%
  \BibitemOpen
  \bibfield  {author} {\bibinfo {author} {\bibfnamefont {M.}~\bibnamefont
  {Cheraghchi}}\ and\ \bibinfo {author} {\bibfnamefont {P.}~\bibnamefont
  {Indyk}},\ }\href {\doibase 10.1145/3029050} {\bibfield  {journal} {\bibinfo
  {journal} {ACM Trans. Algorithms}\ }\textbf {\bibinfo {volume} {13}},\
  \bibinfo {pages} {34:1} (\bibinfo {year} {2017})},\ \Eprint
  {http://arxiv.org/abs/1504.07648} {arXiv:1504.07648} \BibitemShut {NoStop}%
\bibitem [{\citenamefont {Li}\ \emph {et~al.}(2014)\citenamefont {Li},
  \citenamefont {Bradley}, \citenamefont {Pawar},\ and\ \citenamefont
  {Ramchandran}}]{Li2015}%
  \BibitemOpen
  \bibfield  {author} {\bibinfo {author} {\bibfnamefont {X.}~\bibnamefont
  {Li}}, \bibinfo {author} {\bibfnamefont {J.~K.}\ \bibnamefont {Bradley}},
  \bibinfo {author} {\bibfnamefont {S.}~\bibnamefont {Pawar}}, \ and\ \bibinfo
  {author} {\bibfnamefont {K.}~\bibnamefont {Ramchandran}},\ }in\ \href
  {\doibase 10.1109/isit.2014.6875155} {\emph {\bibinfo {booktitle} {2014
  {IEEE} International Symposium on Information Theory}}}\ (\bibinfo
  {publisher} {{IEEE}},\ \bibinfo {year} {2014})\ \Eprint
  {http://arxiv.org/abs/1508.06336} {arXiv:1508.06336} \BibitemShut {NoStop}%
\bibitem [{\citenamefont {Lu}(2018)}]{Lu2018}%
  \BibitemOpen
  \bibfield  {author} {\bibinfo {author} {\bibfnamefont {Y.~J.}\ \bibnamefont
  {Lu}},\ }in\ \href {\doibase 10.1007/978-3-030-03405-4_9} {\emph {\bibinfo
  {booktitle} {Advances in Intelligent Systems and Computing}}}\ (\bibinfo
  {publisher} {Springer International Publishing},\ \bibinfo {year} {2018})\
  pp.\ \bibinfo {pages} {131--144},\ \Eprint {http://arxiv.org/abs/1602.00095}
  {arXiv:1602.00095} \BibitemShut {NoStop}%
\bibitem [{\citenamefont {Koller}\ and\ \citenamefont
  {Friedman}(2009)}]{Koller2009}%
  \BibitemOpen
  \bibfield  {author} {\bibinfo {author} {\bibfnamefont {D.}~\bibnamefont
  {Koller}}\ and\ \bibinfo {author} {\bibfnamefont {N.}~\bibnamefont
  {Friedman}},\ }\href@noop {} {\emph {\bibinfo {title} {Probabilistic
  graphical models: principles and techniques}}}\ (\bibinfo  {publisher} {MIT
  Press},\ \bibinfo {address} {Cambridge, MA},\ \bibinfo {year}
  {2009})\BibitemShut {NoStop}%
\bibitem [{\citenamefont {Besag}(1974)}]{Besag1974}%
  \BibitemOpen
  \bibfield  {author} {\bibinfo {author} {\bibfnamefont {J.}~\bibnamefont
  {Besag}},\ }\href {http://www.jstor.org/stable/2984812} {\bibfield  {journal}
  {\bibinfo  {journal} {Journal of the Royal Statistical Society. Series B
  (Methodological)}\ }\textbf {\bibinfo {volume} {36}},\ \bibinfo {pages} {192}
  (\bibinfo {year} {1974})}\BibitemShut {NoStop}%
\bibitem [{\citenamefont {Barahona}(1982)}]{Barahona1982}%
  \BibitemOpen
  \bibfield  {author} {\bibinfo {author} {\bibfnamefont {F.}~\bibnamefont
  {Barahona}},\ }\href {\doibase 10.1088/0305-4470/15/10/028} {\bibfield
  {journal} {\bibinfo  {journal} {Journal of Physics A: Mathematical and
  General}\ }\textbf {\bibinfo {volume} {15}},\ \bibinfo {pages} {3241}
  (\bibinfo {year} {1982})}\BibitemShut {NoStop}%
\bibitem [{\citenamefont {Jerrum}\ and\ \citenamefont
  {Sinclair}(1993)}]{Jerrum1993}%
  \BibitemOpen
  \bibfield  {author} {\bibinfo {author} {\bibfnamefont {M.}~\bibnamefont
  {Jerrum}}\ and\ \bibinfo {author} {\bibfnamefont {A.}~\bibnamefont
  {Sinclair}},\ }\href {\doibase 10.1137/0222066} {\bibfield  {journal}
  {\bibinfo  {journal} {{SIAM} Journal on Computing}\ }\textbf {\bibinfo
  {volume} {22}},\ \bibinfo {pages} {1087} (\bibinfo {year}
  {1993})}\BibitemShut {NoStop}%
\bibitem [{\citenamefont {Pinsker}(1964)}]{Pinsker1964}%
  \BibitemOpen
  \bibfield  {author} {\bibinfo {author} {\bibfnamefont {M.}~\bibnamefont
  {Pinsker}},\ }\href@noop {} {\emph {\bibinfo {title} {Information and
  information stability of random variables and processes}}}\ (\bibinfo
  {publisher} {Holden-Day},\ \bibinfo {year} {1964})\BibitemShut {NoStop}%
\bibitem [{\citenamefont {Cover}\ and\ \citenamefont
  {Thomas}(1991)}]{Cover1991}%
  \BibitemOpen
  \bibfield  {author} {\bibinfo {author} {\bibfnamefont {T.~M.}\ \bibnamefont
  {Cover}}\ and\ \bibinfo {author} {\bibfnamefont {J.~A.}\ \bibnamefont
  {Thomas}},\ }\href@noop {} {\emph {\bibinfo {title} {Elements of Information
  Theory}}}\ (\bibinfo  {publisher} {Wiley},\ \bibinfo {year}
  {1991})\BibitemShut {NoStop}%
\bibitem [{\citenamefont {Magesan}\ \emph {et~al.}(2012)\citenamefont
  {Magesan}, \citenamefont {Gambetta}, \citenamefont {Johnson}, \citenamefont
  {Ryan}, \citenamefont {Chow}, \citenamefont {Merkel}, \citenamefont
  {da~Silva}, \citenamefont {Keefe}, \citenamefont {Rothwell}, \citenamefont
  {Ohki}, \citenamefont {Ketchen},\ and\ \citenamefont
  {Steffen}}]{Magesan2012b}%
  \BibitemOpen
  \bibfield  {author} {\bibinfo {author} {\bibfnamefont {E.}~\bibnamefont
  {Magesan}}, \bibinfo {author} {\bibfnamefont {J.~M.}\ \bibnamefont
  {Gambetta}}, \bibinfo {author} {\bibfnamefont {B.~R.}\ \bibnamefont
  {Johnson}}, \bibinfo {author} {\bibfnamefont {C.~A.}\ \bibnamefont {Ryan}},
  \bibinfo {author} {\bibfnamefont {J.~M.}\ \bibnamefont {Chow}}, \bibinfo
  {author} {\bibfnamefont {S.~T.}\ \bibnamefont {Merkel}}, \bibinfo {author}
  {\bibfnamefont {M.~P.}\ \bibnamefont {da~Silva}}, \bibinfo {author}
  {\bibfnamefont {G.~A.}\ \bibnamefont {Keefe}}, \bibinfo {author}
  {\bibfnamefont {M.~B.}\ \bibnamefont {Rothwell}}, \bibinfo {author}
  {\bibfnamefont {T.~A.}\ \bibnamefont {Ohki}}, \bibinfo {author}
  {\bibfnamefont {M.~B.}\ \bibnamefont {Ketchen}}, \ and\ \bibinfo {author}
  {\bibfnamefont {M.}~\bibnamefont {Steffen}},\ }\href {\doibase
  10.1103/PhysRevLett.109.080505} {\bibfield  {journal} {\bibinfo  {journal}
  {Phys. Rev. Lett.}\ }\textbf {\bibinfo {volume} {109}},\ \bibinfo {pages}
  {080505} (\bibinfo {year} {2012})},\ \Eprint {http://arxiv.org/abs/1203.4550}
  {arXiv:1203.4550} \BibitemShut {NoStop}%
\bibitem [{\citenamefont {Ferracin}\ \emph {et~al.}(2019)\citenamefont
  {Ferracin}, \citenamefont {Kapourniotis},\ and\ \citenamefont
  {Datta}}]{Ferracin2018}%
  \BibitemOpen
  \bibfield  {author} {\bibinfo {author} {\bibfnamefont {S.}~\bibnamefont
  {Ferracin}}, \bibinfo {author} {\bibfnamefont {T.}~\bibnamefont
  {Kapourniotis}}, \ and\ \bibinfo {author} {\bibfnamefont {A.}~\bibnamefont
  {Datta}},\ }\href {\doibase 10.1088/1367-2630/ab4fd6} {\bibfield  {journal}
  {\bibinfo  {journal} {New J. Phys.}\ }\textbf {\bibinfo {volume} {21}},\
  \bibinfo {pages} {113038} (\bibinfo {year} {2019})},\ \Eprint
  {http://arxiv.org/abs/1811.09709} {arXiv:1811.09709} \BibitemShut {NoStop}%
\bibitem [{\citenamefont {Granade}\ \emph {et~al.}(2015)\citenamefont
  {Granade}, \citenamefont {Ferrie},\ and\ \citenamefont {Cory}}]{Granade2015}%
  \BibitemOpen
  \bibfield  {author} {\bibinfo {author} {\bibfnamefont {C.}~\bibnamefont
  {Granade}}, \bibinfo {author} {\bibfnamefont {C.}~\bibnamefont {Ferrie}}, \
  and\ \bibinfo {author} {\bibfnamefont {D.~G.}\ \bibnamefont {Cory}},\ }\href
  {\doibase 10.1088/1367-2630/17/1/013042} {\bibfield  {journal} {\bibinfo
  {journal} {New J. Phys.}\ }\textbf {\bibinfo {volume} {17}},\ \bibinfo
  {pages} {013042} (\bibinfo {year} {2015})},\ \Eprint
  {http://arxiv.org/abs/1404.5275} {arXiv:1404.5275} \BibitemShut {NoStop}%
\bibitem [{\citenamefont {Verstraete}\ \emph {et~al.}(2004)\citenamefont
  {Verstraete}, \citenamefont {Garc\'{\i}a-Ripoll},\ and\ \citenamefont
  {Cirac}}]{Verstraete2004}%
  \BibitemOpen
  \bibfield  {author} {\bibinfo {author} {\bibfnamefont {F.}~\bibnamefont
  {Verstraete}}, \bibinfo {author} {\bibfnamefont {J.~J.}\ \bibnamefont
  {Garc\'{\i}a-Ripoll}}, \ and\ \bibinfo {author} {\bibfnamefont {J.~I.}\
  \bibnamefont {Cirac}},\ }\href {\doibase 10.1103/PhysRevLett.93.207204}
  {\bibfield  {journal} {\bibinfo  {journal} {Phys. Rev. Lett.}\ }\textbf
  {\bibinfo {volume} {93}},\ \bibinfo {pages} {207204} (\bibinfo {year}
  {2004})},\ \Eprint {http://arxiv.org/abs/cond-mat/0406426} {cond-mat/0406426}
  \BibitemShut {NoStop}%
\bibitem [{\citenamefont {Zwolak}\ and\ \citenamefont
  {Vidal}(2004)}]{Zwolak2004}%
  \BibitemOpen
  \bibfield  {author} {\bibinfo {author} {\bibfnamefont {M.}~\bibnamefont
  {Zwolak}}\ and\ \bibinfo {author} {\bibfnamefont {G.}~\bibnamefont {Vidal}},\
  }\href {\doibase 10.1103/PhysRevLett.93.207205} {\bibfield  {journal}
  {\bibinfo  {journal} {Phys. Rev. Lett.}\ }\textbf {\bibinfo {volume} {93}},\
  \bibinfo {pages} {207205} (\bibinfo {year} {2004})},\ \Eprint
  {http://arxiv.org/abs/cond-mat/0406440} {cond-mat/0406440} \BibitemShut
  {NoStop}%
\end{thebibliography}%
\end{document}